\documentclass[a4paper,11pt]{scrartcl}

\usepackage{amssymb,amsthm,amsmath,bbm, paralist,bbding,pifont,upgreek,nicefrac}
\DeclareMathAlphabet\mathbfcal{OMS}{cmsy}{b}{n}
\usepackage{enumitem}
\usepackage{multicol}
\usepackage{graphicx}
\usepackage{xcolor}
\definecolor{plotgray}{rgb}{.4,.4,.4}
\usepackage{todonotes}
\usepackage[utf8]{inputenc}
\usepackage[normalem]{ulem}
\usepackage{pgfplots}
\usepackage{caption}

\usepackage[numbers,sort&compress]{natbib}

\usepackage{setspace}

\usepackage{url}
\usepackage{float}

\usepackage{tikz}
\usetikzlibrary{arrows}
\usepackage{pgfplots}
\usetikzlibrary{decorations.text}
\usetikzlibrary{shapes}
\usetikzlibrary{trees}
\usetikzlibrary{calc}
\usepgfplotslibrary{fillbetween}
\usepackage{ulem}
\usepackage{upgreek}

\usepackage[ruled]{algorithm2e}

\usepackage[final]{microtype}

\pgfplotsset{compat=1.16}

\newtheorem{theorem}{Theorem}[section]
\newtheorem{definition}[theorem]{Definition}

\newtheorem{lemma}[theorem]{Lemma}
\newtheorem{proposition}[theorem]{Proposition}
\newtheorem{remark}[theorem]{Remark}
\newtheorem{corollary}[theorem]{Corollary}
\newtheorem{assumption}{Assumption}

\newtheorem{assumptionaux}{Assumption}

\definecolor{hypercolor}{rgb}{0,0.2,0.7}
\usepackage[breaklinks,final]{hyperref}

\hypersetup{
  unicode,
  colorlinks,
  linkcolor=hypercolor,
  urlcolor=hypercolor,
  citecolor=hypercolor,
  bookmarksnumbered,
  bookmarksdepth=2,
}

\newcommand{\munu}{\ensuremath {{\mu\nu}}}
\newcommand{\R}{\ensuremath\mathbb{R}}
\newcommand{\CC}{\ensuremath\mathbb{C}}
\newcommand{\N}{\ensuremath\mathbb{N}}

\newcommand{\dd}{\ensuremath{\textup{d}}}

\newcommand{\Tmunuren}{\ensuremath\big\langle {T_\munu^{\textup{ren}}}\big\rangle_\omega}
\newcommand{\Tmunurenen}{\ensuremath\big\langle {T_{00}^{\textup{ren}}}\big\rangle_\omega}
\newcommand{\Tmunurenjj}{\ensuremath\big\langle {T_{jj}^{\textup{ren}}}\big\rangle_\omega}
\newcommand{\TmunurenBD}{\ensuremath\big\langle {T_\munu^{\textup{ren}}}\big\rangle_{\omega^\textup{BD}}}

\newcommand{\mathcalm}{\ensuremath\raisebox{-1pt}{\includegraphics[scale=1.1]{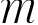}}}

\newcommand{\smalloh}{\scalebox{.6}{$\mathbfcal{O}$}}

\def\headingmath$#1${\texorpdfstring{{\rmfamily\textit{#1}}}{#1}}

\title{Cosmological de Sitter Solutions of the Semiclassical Einstein Equation}

\author{Hanno Gottschalk${}^1$, Nicolai Rothe${}^1$ and Daniel Siemssen${}^2$\\ \\ {\small ${}^1$Institute of Mathematics, TU Berlin, 10623 Berlin, Germany} \\{\small ${}^2$School of Mathematics and Natural Science \& IMACM,}\\ {\small University of Wuppertal, D-42119 Wuppertal, Germany}\\
 \texttt{\small $\{$gottschalk,rothe$\}$@math.tu-berlin.de, siemssen@uni-wuppertal.de}\\[1ex]
}

\date{{\small \today}}
\begin{document}
\maketitle

\begin{abstract}
    Exponentially expanding space-times play a central role in contemporary cosmology, most importantly in the theory of inflation and in the Dark Energy driven expansion in the late universe. In this work, we give a complete list of de Sitter solutions of the semiclassical Einstein equation (SCE), where classical gravity is coupled to the expected value of a renormalized stress-energy tensor of a free quantum field in the Bunch-Davies state. To achieve this, we explicitly determine the stress-energy tensor associated with the Bunch-Davies state using the recently proposed `moment approach' on the cosmological coordinate patch of de Sitter space. From the energy component of the SCE we thus obtain an analytic consistency equation for the model's parameters which has to be fulfilled by solutions to the SCE. Using this equation we then investigate the number of solutions and the structure of the solution set in dependency on the coupling parameter of the quantum field to the scalar curvature and renormalization constants using analytic arguments in combination with numerical evidence. We also identify parameter sets where multiple expansion rates separated by several orders of magnitude are possible. Potentially for such parameter settings, a fast (semi-stable) expansion in the early universe could be compatible with a late time `Dark Energy-like' behavior of the universe.
\end{abstract}

\section{Introduction}
\label{Introduction}
In modern cosmology, the $\Lambda$CDM is considered the standard model as it explains a large amount of observational data (see e.g.~\cite{pdg-data} and references therein). However, one of its predictions is the presence of Dark Energy or, equivalently, a positive cosmological constant. While matter in this model happens to be purely classical, it is one possible option that Dark Energy may naturally emerge from quantum effects, that is, if the matter content is modelled by a quantum field.

Physicists have put a lot of effort into deriving a satisfactory quantum theory of gravity. For an extensive survey on a wide range of quantum cosmological effects we refer to Schander and Thiemann~\cite{SchanderTiemann}. The semiclassical Einstein equation (SCE) takes a rather moderate approach. In particular, gravity is avoided to be quantized and modelled by a classical metric formalism governed by an Einstein equation. The energy content, on the other hand, is modelled by quantum fields. The latter are then coupled to classical gravity via the expected value of the stress-energy tensor in a quantum state $\omega$ and the SCE reads
\begin{equation}
    G_\munu+\Lambda g_\munu=\kappa\Tmunuren\,. \label{SCE:nonzero_lambda}
\end{equation}
Hereby, $g_\munu$ is the metric\footnote{We use the signature convention $(-,+,+,+)$}, $G_\munu=R_\munu-\tfrac{1}{2}R\,g_\munu$ is the Einstein curvature tensor (consisting of the Ricci curvature tensor $R_\munu$ and its trace $R$) and $\Lambda$ is the cosmological constant. $T_\munu^\textup{ren}$ is the renormalized quantum stress-energy (QSE) tensor of one or multiple quantum field(s), coupling to the geometry of the underlying space-time with a strength controlled by the parameter $\kappa$. We restrict to free scalar fields $\phi$ governed by the Klein-Gordon equation
\begin{equation}\label{eq:KG-eq}
(\Box+\xi R+m^2)\phi=0
\end{equation}
with mass $m$ and curvature coupling $\xi$. $\square=-g^\munu\nabla_\mu\nabla_\nu$ denotes the d'Alembertian of the metric $g_\munu$.

The SCE has been introduced in a series of articles from the late 70's by Davies, Fulling et al.\ which culminated in~\cite{Davies_etal}. The problem of finding a suitable QSE tensor was then axiomatized by Wald~\cite{Wald:1977up} and further refined by Christensen~\cite{Christensen,Christensen2}, coming up with a properly covariant regularization scheme. The SCE (for cosmological settings) was approached by numerical algorithms and special analytic solutions have been found, e.g.\ by Anderson~\cite{Anderson1,Anderson2,Anderson3,Anderson4} or Suen \& Anderson~\cite{SuenAnderson}, as well as Starobinski~\cite{Starobinsky}. Another modern view on the physical content of the SCE was provided by Flanagan \& Wald in~\cite{FlanaganWald:1996}. We refer to the monographs by Birrel \& Davies~\cite{birrell1984quantum}, Fulling~\cite{fulling1989aspects} and Wald~\cite{wald1994quantum} for a comprehensive view on the research on the SCE up to the 90's. In 2003, Moretti~\cite{Moretti:2001qh} defined a covariantly conserved QSE tensor as demanded in one of Wald's axioms, whereas in 2004, the same QSE was derived from a completely different point of view in~\cite{HollandsWald:2004}.

A mathematical theory of solutions, tailored to cosmological settings, began to be developed in the late 2000's. In their work~\cite{dappiaggi2008stable} the authors Dappiaggi, Fredenhagen \& Pinamonti observed a distinguished behavior in their solutions which they call de Sitter-type behavior and we will pick up this discussion a bit later. Ongoing, the first result towards a mathematical solution theory, providing local existence and uniqueness results for the trace of the SCE, was formulated in the seminal article~\cite{Pinamonti} by Pinamonti. This approach was further refined, particularly studying global properties of solutions and their continuability, by Pinamonti \& Siemssen in~\cite{Pinamonti:2013,daniel-diss}. These works, similarly to many of the older references cited above, focused the conformally coupled case. The review articles by Fredenhagen \& Hack~\cite{FredenhagenHack2013} and by Hack~\cite{Hack2010,Hack2015} built bridges between purely mathematical solution theory and the modern physicists' approaches to cosmology. The results of~\cite{Pinamonti:2013} were recently generalized to non-conformally coupled fields by Meda, Pinamonti \& Siemssen in~\cite{Meda:2020}. Therein, the SCE is reformulated as a fixed point equation of a certain operator on a suitable Banach space, which allows to conclude short-time existence and uniqueness of solutions by a fixed-point theorem. Moreover, in the same period of time Eltzner \& Gottschalk~\cite{Elzner-G} managed to write the SCE into one dynamical system for both the scaling factor (of the underlying Friedman-Lema\^itre-Robertson-Walker (FLRW) space-time) and for the expectation values of Wick products of the field and its derivatives. This result was further refined and reformulated in a rigorous framework in~\cite{siemssen-gottschalk}, where the authors prove (global) existence and uniqueness of solutions. The latter approach, however, has the disadvantage of a somewhat implicit definition of the initial state.

Recently, the SCE has been used to derive special cosmological models. Sanders constructed in~\cite{Sanders} maximally symmetric states on a given static (cosmological) space-time of positive (spatial) curvature. Moreover, in~\cite{Numerik-Paper} the authors of the present article have used the techniques of~\cite{siemssen-gottschalk} to study a class of cosmological expansion models driven by a massless scalar field in a Minkowski-like state, motivated by some observation on the Minkowski-vacuum state on Minkowski space. A recent paper by Juárez-Aubry~\cite{juarez2020semiclassical} studies the SCE on static and ultrastatic, not necessarily spatially homogeneous space-times as an initial value problem for the state. This work was recently extended to more generic settings and applied to so-called `quantum state collapse' scenarios in \cite{juarez2022initialvalue}. A noteworthy recent result is given in~\cite{haensel2019}, where the author studies the ordinary differential equations for $H=\frac{\dot{a}}{a}$ arising from the conformally coupled and massless SCE and to a certain extend classifies the corresponding dynamical systems in terms of the topological properties of their phase portraits. Another work on special solutions is~\cite{JUAREZAUBRY}, also by Juárez-Aubry. There the author finds solutions of the SCE which coincide with the classical vacuum solutions for a positive cosmological constant $\Lambda$, that is, the de Sitter solution with constant curvature $R=\Lambda$ or, equivalently, $a(t)=\exp(\sqrt{\nicefrac{\Lambda}{3}\,}\,t)$. In particular, the expectation value of the QSE tensor is taken with respect to the so-called Bunch-Davies state as first introduced by Bunch \& Davies in~\cite{BunchDavies} and further discussed by Allen in~\cite{Allen}. By isometrically embedding a cosmological space-time with pure de Sitter expansion into de Sitter space, this distinguished Bunch-Davies state can be pulled back and indeed yields a global state on the given de Sitter-type cosmological space-time. However, searching for these particular (vacuum) solutions of the SCE corresponds to solving $\Tmunuren=0$. (Note that $G_\munu+\Lambda g_\munu=0$ vanishes for the currently discussed metric/scale factor $a$). Thus \cite{JUAREZAUBRY} is studying the parameter set for which the presence of the Bunch-Davies vacuum state has no back-reaction effect to the space-time. While this approach is suitable for the treatment of the so-called \emph{cosmological constant problem} (as in~\cite{JUAREZAUBRY}), it does not cover all cases where a de Sitter expansion and the Bunch-Davies state on the resulting space-time yield a solution to the cosmological SCE. In particular, it omits situations when the presence of the vacuum state does have a back-reaction effect.

Another perspective on studying exponential late-time behavior of cosmological expansions deals with the topic of so-called energy conditions. These kinds of considerations are based on an observation by Wald~\cite{Wald_deSitter} from '83, namely that FLRW solutions to the Einstein equation with positive $\Lambda$ generally approach the respective exponential vacuum solution with Hubble rate $\sqrt{\nicefrac{\Lambda}{3}}$, provided the stress-energy tensor fulfills specific energy conditions in terms of two inequalities. Wald's article can be viewed as establishing a rigorous argument linking Dark Energy/a positive cosmological constant with a late-time de Sitter phase and thereby proving a no-singularity theorem (at late times) from these conditions. However, QSE tensors usually do not fulfill the conditions of \cite{Wald_deSitter} in a pointwise manner and the focus has shifted towards studying whether the observation of~\cite{Wald_deSitter} (and also other observations concerning space-time singularities) remain(s) true if the stress-energy tensor fulfills similar weaker conditions, for example where now the stress-energy tensor is averaged along time-like geodesics. For further reading we refer to some recent articles on this topic, e.g.\ to Fewster \& Kontou~\cite{Fewster:2018}, Fewster \& Smith~\cite{Fewster:2007rh}, Fewster \& Verch~\cite{Fewster:2001js} and, in particular, to the comprehensive introduction by Kontou \& Sanders in~\cite{Kontou:review}.

Apart from late time de Sitter solutions, many authors also discuss inflation, i.e.\ solutions with a de Sitter phase at early time. An inflationary phase in the universe's expansion was originally suggested as a solution to the so-called cosmic horizon problem~\cite{guth,Linde,Liddle}. The main architects of inflationary physics in its modern shape are Guth, Linde and Starobinski, with their most noteworthy articles on that topic~\cite{guth,Linde,StarobinskiII}, respectively, all from the early 1980's. In particular, Starobinski addresses the compatibility of an inflationary phase with semiclassical gravity, and with this purpose notes the existence of pure de Sitter expansion solutions to the SCE as mentioned above. Moreover, considerable progress was made by Mukhanov~\cite{Mukhanov} and others to explain the scale free spectrum of cosmological structures. For comprehensive discussions we refer to the reviews of Liddle~\cite{Liddle} and of Hack~\cite{Hack2015}, where the latter in particular discusses inflationary models in view of modern algebraic quantum field theory.

Our present work is dedicated to finding \emph{all} solutions to the cosmological SCE whose scaling factor describes a purely exponential expansion, driven by a massless or massive scalar field in the (pullback) Bunch-Davies vacuum state.

The physical motivation is to identify parameter settings, where for both an inflationary phase and a late-time de Sitter phase there are two (or more) exact solutions of the aforementioned kind which approximate the universe's expansion during these phases. Suppose that in future work one can show that the solution with a larger Hubble rate is unstable towards perturbations, that the solution with smaller rate is stable and that these two solutions are connected by a trajectory in phase space. Then one would show that a scalar quantum field is capable to drive both inflation and a late-time Dark Energy-dominated expansion. Note that such stability/instability behavior of de Sitter solutions has been found for simplified semiclassical models in~\cite{dappiaggi2008stable} and~\cite{haensel2019} (cf.\ the discussion in Section~\ref{sec:Inflationary_models}).

On the other hand, a complete list of de Sitter solutions to the cosmological SCE with the (pullback) Bunch-Davies state is interesting in its own right since these solutions frequently occur in semiclassical cosmology (e.g.~\cite{Numerik-Paper,Starobinsky,dappiaggi2008stable,degner}). Hereby the pullback Bunch-Davies state can be viewed as distinguished by its symmetry behavior on (the entire, non-cosmological) de Sitter space. Moreover, \cite{degner} finds that on cosmological de Sitter space-times any `` state of low energy'' (a physically distinguished class of states introduced in \cite{olbermann2007states}) converges to the Bunch-Davies state in a suitable sense. Note that while we are mostly interested in the cosmological setting in order to make sense of any stability features, we also obtain, as a by-product, a list of all parameters for which (the entire) de Sitter space and thereon the (non-pullback) Bunch-Davies solve the SCE.

We summarize our results in the following main theorem:
\begin{theorem}\label{thm:intro-version}
For $H>0$, consider the cosmological space-time with flat spatial sections defined by $a(t)=\exp(Ht)$ and thereon a free scalar quantum field $\phi$ in the state obtained by pulling back the Bunch-Davies vacuum on de Sitter space with radius $\frac{1}{H}$ along cosmological coordinates. The field dynamics is governed by the Klein-Gordon equation \eqref{eq:KG-eq} with parameters $m$ and $\xi$. Then:
\begin{itemize}
    \item[\textup{(i)}] The semiclassical Einstein equation \eqref{SCE:nonzero_lambda} with coupling $\kappa$ and cosmological constant $\Lambda$ for this field and state breaks down into a $($non-dynamic$)$ consistency equation for the parameters $H$, $m$, $\xi$, $\kappa$, $\Lambda$ and the renormalization constants originating in $\phi$'s stress-energy tensor. Of these parameters only four are independent.
    \item[\textup{(ii)}] Viewing the consistency equation as a constraint on $(\xi,H)$-pairs with two $($effective$)$ remaining parameters, the solution set can be parameterized by analytic curves in the $\xi$-$H$-plane. In numbers, these are one or two curves if $m=0$, or two or three curves if $m>0$.
    \item[\textup{(iii)}] The large-$H$ and small-$H$ asymptotics of the solution curves can be explicitly worked out.
\end{itemize}
\end{theorem}
{\parindent0pt Note that, by time-reflection invariance, values $H<0$ correspond to the positive values $-H>0$. The well known $H=0$-Minkowski case is not discussed here. In particular, with the knowledge of part (iii) of the theorem we can conclude:}
\begin{corollary}\label{cor:intro-version}
There exist parameter settings $m$, $\xi$, $\kappa$, $\Lambda$ and renormalization constants, such that multiple $H$-values solve the consistency equation. Moreover, the model of Theorem~\textup{\ref{thm:intro-version}} in the case $m>0$ is flexible enough such that for any two prescribed positive values of $H$ $($with a sufficiently large ratio$)$ there exist a set of remaining parameters such that both given $H$-values are solutions.
\end{corollary}
While the last statement on the $m>0$-case remains true for an arbitrary triple of positive numbers, we are mostly interested in tuning parameters to obtain two prescribed solutions for the reason discussed above. Also we note that while the massless model is not as flexible as the massive one, it remains true that for any (sufficiently) large prescribed value $H_\textup{I}$ (approximating an inflationary phase), there can be found a value of curvature coupling $\xi$ such that $H_\textup{I}$ is a solution and, moreover, a second solution can be found close to $H_\textup{vac}=\sqrt{\nicefrac{\Lambda}{3}}$ (if $\Lambda>0$).

While in the following the results are stated more precisely and tailored to the respective cases, Theorem~\ref{thm:intro-version} and Corollary~\ref{cor:intro-version} immediately follow from Proposition~\ref{prop:massless_case}, Theorems~\ref{thm:solution_set_structure}, \ref{thm:asymptotic_solution_curves} and \ref{thm:inflationary_wo_any_assumptions} as well as the discussions in Sections~\ref{Section:The_SCE_on_dS_space-time} and \ref{sec:Inflationary_models}.

Moreover, note that some intermediate results were merely accessible by numerically evaluating certain functions. Typically we need a statement of the form $h(x)>0$ for all $x\in I$ for some analytic function $h:I\to\R$ on an interval $I\subset(0,\infty)$ and we proceed as follows: As a first step, using asymptotic expansions we prove that there exist $x_1,x_2\in I$ such that $h(x)>0$ for all $x\in(\inf I,x_1)\cup(x_2,\sup I)$. Thereafter, we numerically evaluate $h$ on a sufficiently dense and sufficiently wide-spread grid in a way that we can identify any precomputed asymptotic expansion. Finally, we observe that all numeric values of $h$ in consideration are positive and, knowing that $h$ is analytic (and assuming that analytic functions cannot behave ``too wild''), we have no doubt that the respective assertions are true, although they are not rigorously proven. In order to stick with the theorem-/proof-style and the proposition labelling throughout the text, we capture the respective assertions in Assumptions \ref{ass:first}, \ref{ass:second} and \ref{ass:third}. These are then followed (at appropriate positions in the text) by ``proof-paragraphs'' which are opened by the phrase ``Numerical evidence for ...'' and closed by a rotated ``q.e.d.-box' \rotatebox{45}{$\square$} in order to distinguish them from purely analytic proofs. Thereafter, we formulate any proposition depending on the numeric evidence as that it is implied by one or more of these assumptions. We emphasise that Corollary \ref{cor:intro-version} following from Theorem (\ref{thm:inflationary_wo_any_assumptions}) does not depend on any such numerical evidence.

Our paper is organized as follows: The second section contains the derivation of the consistency equation for the special case of cosmological de Sitter space-times and free scalar fields in the (pullback) Bunch-Davies vacuum state. As a main tool, we utilize the `moments' approach to the SCE in~\cite{siemssen-gottschalk}. Note that while the consistency equation as such could have been derived faster using the results of~\cite{tadaki}, we particularly develop a viewpoint in which the de Sitter solution correspond to a phase space trajectory for the SCE as a dynamical system.

The third section is dedicated to the massless case, where the consistency equation simplifies into an explicitly solvable polynomial equation. After solving the equation and plotting the solution sets in the $\xi$-$H$-plane we discuss the asymptotics of the solution curves.

In Sections 4, 5 and 6 we study the massive case. First, in Section 4, we exploit the fact that the solution set is the zero set of an analytic function. In particular, we show that any solution to the consistency equation belongs to an analytic solution curve which is extendible into the asymptotics of the equation. Second, in Section 5 we explicitly evaluate the asymptotics of said analytic function, completing the argument as it is presented in Theorem~\ref{thm:intro-version} above. Section 6 then graphically presents the solution set as obtained by numeric evaluation, confirming the results of the previous sections.

Section 7, finally, uses the results of the previous sections in order to show the existence of parameter settings for potential inflationary models as in Corollary~\ref{cor:intro-version}.

In the last section we present our conclusion and discuss some open problems for future research.

\section{The semiclassical Einstein equation on de Sitter space-time}
\label{Section:The_SCE_on_dS_space-time}

In this section we derive the consistency condition for the parameters under which the cosmological SCE admits a solution with a pure de Sitter expansion law, driven by a scalar field in the (pullback) Bunch-Davies state.

\subsection{The energy equation as a cosmological model}
\label{Section:The-energy-equation-as-a-cosmological-model}
A priori, \eqref{SCE:nonzero_lambda} is actually a system of 16 equations and by the symmetries of both the Einstein and the stress-energy tensor, these reduce to ten independent equations. For a (flat) cosmological FLRW metric 
\begin{equation}
    g=-\textup{d}t^2+a(t)^2\big(\textup{d}y_1^2+\textup{d}y_2^2+\textup{d}y_3^2\big)\label{eq:cosmological_metric}
\end{equation}
with scaling factor $a(t)$ and a state $\omega$ that shares the space-time symmetries only two of these ten equations are independent, for example the 00- and one of the $jj$-components ($j=1,2,3$). These two independent equations can be captured in the so-called energy and trace equations,
\begin{equation}\label{abstract_trace_n_energy_equation}
    G_{00}-\Lambda=\kappa\Tmunurenen\qquad\textup{and}\qquad g^\munu G_\munu+4\Lambda=\kappa g^\munu\Tmunuren\,,
\end{equation}
respectively. Note that by the above assumptions on $g$ and $\omega$ the stress-energy tensor is of the form 
\begin{equation}\label{eq:cosmological_stress-energy_tensor}
\Tmunuren=\mathrm{diag}\Big(\langle\,\varrho\,\rangle_\omega,a^2\langle \,p\,\rangle_\omega,a^2\langle \,p\,\rangle_\omega,a^2\langle \,p\,\rangle_\omega\Big)
\end{equation}
with the (expected) energy density $\langle\,\varrho\,\rangle_\omega=\Tmunurenen$ and pressure $\langle \,p\,\rangle_\omega=\frac{1}{a^2}\Tmunurenjj$ (with a spatial index $j$) which both no longer depend on the spatial coordinates. Consequently, the metric degrees of freedom in \eqref{abstract_trace_n_energy_equation} are governed by an ODE.
Moreover, note that the condition of covariant conservedness,  $\nabla^\mu\Tmunuren=0$, for a stress-energy tensor of the form \eqref{eq:cosmological_stress-energy_tensor} can be rewritten into\footnote{We denote derivatives of $a$ with respect to cosmological time $t$ by dots, i.e.\ $\dot{a}$, $\ddot{a}$ and so on.\label{footnote:dot_prime_convention}}
\begin{equation}\notag
\langle\,\raisebox{.5pt}{$\dot{\raisebox{-.5pt}{$\varrho$}}$}\,\rangle_\omega+3\frac{\dot{a}}{a}\Big(\langle\,\varrho\,\rangle_\omega+\langle \,p\,\rangle_\omega\Big)=0\,.
\end{equation}
It is well-known that this so-called continuity equation and the energy equation imply the trace equation (whenever $\dot{a}\neq 0$) and hence, also the full Einstein equation. As in quantum field theories the stress-energy tensor is covariantly conserved by construction (cf.\ the discussion in Section~\ref{sec:The_consistency_equation_in_the_moment_based_approach}), it thus suffices to consider the energy equation instead of the full Einstein equation.

\subsection{De Sitter space and the Bunch-Davies state}
\label{sec:De_Sitter_space_and_the_Bunch-Davies_state}

De Sitter space is the four-dimensional one-sheet hyperboloid of a certain radius as a pseudo-Riemannian submanifold of five-dimensional Minkowski space $\mathbb{M}=\R^5$, oriented around the time axis of the latter. Formally, endow $\mathbb{M}$ with the coordinates $z_0,z_1,z_2,z_3,z_4:\R\to\mathbb{M}$ (with time axis along the $z_0$ coordinate), then de Sitter space is the set
\[
    \mathbbm{dS}_H=\big\{~(z_0,z_1,z_2,z_3,z_4)\in\mathbb{M}~\big|~ -z_0^2+\scalebox{1}{$\sum\limits_{\scriptscriptstyle i=1}^{\scriptscriptstyle 4}$} \,z_i^2=\tfrac{1}{H^2}~\big\}
\]
for some parameter $H>0$, with the pullback Lorentzian metric via the canonical embedding $\mathbbm{dS}_H\to\mathbb{M}$. The group of isometries of $\mathbbm{dS}_H$ is given by the full Lorenz group $O(4,1)$ and $\mathbbm{dS}_H$, as a submanifold of $\mathbb{M}$, is left invariant under this group's action. In particular, the induced pullback metric is left invariant as well.

We choose the coordinates $(t,y_1,y_2,y_3)$ such that
\begin{align*}
    z_0&=\tfrac{1}{H}\sinh(Ht)+\tfrac{1}{2}H\textup{e}^{Ht}\big(y_1^2+y_2^2+y_3^2)\,,\\
    z_i&=\textup{e}^{Ht}y_i\qquad\qquad (i=1,2,3),\\
    z_4&=\tfrac{1}{H}\cosh(Ht)-\tfrac{1}{2}H\textup{e}^{Ht}\big(y_1^2+y_2^2+y_3^2)\,,
\end{align*}
which cover $\widetilde{\mathbbm{dS}_H}=\{~(z_0,z_1,z_2,z_3,z_4)\in\mathbbm{dS}_H\,|\,z_0+z_4> 0~\}$, the so-called cosmological patch of $\mathbbm{dS}_H$. Pulling the metric of $\mathbb{M}$ back to $\widetilde{\mathbbm{dS}_H}$ through these coordinates we obtain the metric
\begin{align}
\begin{split}
    g&=-\dd t^2+\textup{e}^{2Ht}(\dd y_1^2+\dd y_2^2+\dd y_3^2)\\
    &=\tfrac{1}{H^2\tau^2}(-\dd\tau^2+\dd y_1^2+\dd y_2^2+\dd y_3^2)\label{cosmological_patch_metric}
\end{split}
\end{align}
on $\R^4$ or $\R_{\tau>0}^4$, respectively, where for the latter representation of the metric we defined the conformal time coordinate
\[
    \tau(t)=\int_t^\infty\frac{1}{a(t')}\,\dd t'{\,=\frac{1}{H}\,\textup{e}^{-Ht}}\,.
\]
Hence, (the cosmological patch of) de Sitter space can be regarded as flat FLRW-type space and comparing \eqref{cosmological_patch_metric} with \eqref{eq:cosmological_metric} we identify the scale factor $a(t)=\textup{e}^{Ht}$ in cosmological time $t$, yielding $a(\tau)=\frac{1}{H\tau}$ in conformal time $\tau$. In the following we will sloppily speak of $\R^4_{\tau>0}$, endowed with the metric \eqref{cosmological_patch_metric}, as cosmological de Sitter space or simply as de Sitter space-time.

Generally in algebraic QFT, one major difficulty is to define a state of the field algebra. On de Sitter space $\mathbbm{dS}_H$ there exists a preferred choice of such, namely the Bunch-Davies state $\omega^{\textup{BD}}$~\cite{BunchDavies,Allen}. Among all $O(4,1)$-invariant states discussed in~\cite{Allen}, the Bunch-Davies state is the only Hadamard state, suggesting it as a natural choice of vacuum state. Note that, in order for the Bunch-Davies state to exist, we have to assume that the effective de Sitter mass of the field is positive, $m^2+12\xi H^2>0$ (cf.~\cite{Allen}).

As a quasi-free state~\cite{wald1994quantum} the Bunch-Davies state is determined by its two-point function
\begin{equation}\label{eq:BD_tpf}
    \omega_2^\textup{BD}(y,z)=
    \frac{2(6\xi-1)H^2+m^2}{8\pi\cos(\pi\nu)}~{}_2F_1\big(\,\tfrac{3}{2}+\nu,~\tfrac{3}{2}-\nu;~2;~\tfrac{1}{2}(1+\mathcal{Z}(y,z))\,\big)
\end{equation}
($y,z\in\mathbbm{dS}_H$), where ${}_2F_1$ is the hypergeometric function, $\nu
=\sqrt{\frac{9}{4}-12\xi-\frac{m^2}{H^2}}$ and $\mathcal{Z}(y,z)$ is the chord length between $y$ and $z$,
\[
    \mathcal{Z}(y,z)=\mathcal{Z}\big((\tau_y,\mathbf{y}),(\tau_z,\mathbf{z})\big)=\frac{\tau_y^2+\tau_z^2-(\mathbf{y}-\mathbf{z})^2}{2\tau_y\tau_z}
\]
(in conformal-time cosmological coordinates). For this particular representation of the Bunch-Davies state's two-point function we refer to~\cite{daniel-diss}, see also~\cite{Allen} for a similar representation. 

As a remark, we note that $\nu$ is not necessarily real. In \cite{Allen} the formula \eqref{eq:BD_tpf} is derived as the solution of the Klein-Gordon equation, where in the symmetric setting of $\mathbbm{dS}_H$ the latter reduces to an ODE for $\mathcal{Z}$. By the series expansion of ${}_2F_1$ (cf.\ e.g.\ 15.1.1 in \cite{abramowitz}) we obtain
\begin{equation}\label{eq:symmetric_hypergeometric_series_representation}
{}_2F_1(a,\bar{a};2;z)=\sum_{n=0}^\infty\frac{|\Gamma(a+n)|^2}{|\Gamma(a)|^2}\,\frac{z^n}{n!(n+1)!}
\end{equation}
(utilizing $\Gamma(\bar{z})=\overline{\Gamma(z)}$, $\Gamma$ is the Gamma function) and it is immediate that ${}_2F_1(a,\bar{a};2;\cdot)$ is a real-valued function for real $z$ (wherever \eqref{eq:symmetric_hypergeometric_series_representation} converges, implying the same for any analytic continuations of \eqref{eq:symmetric_hypergeometric_series_representation}). Moreover, \eqref{eq:symmetric_hypergeometric_series_representation} shows that ${}_2F_1(a,\bar{a};2;\cdot)$ does not depend on which branch of the complex square-root we choose in the case where $\nu$ is a purely imaginary number. In fact, in Appendix~\ref{appendix-function-f} we shortly discuss on the level of the stress-energy tensors that the expression $\TmunurenBD$ (as introduced in the next section) with parameters $12\xi+\frac{m^2}{H^2}<\frac{9}{4}$ (such that $\nu$ is a positive real) is analytically continuated by the same expression with parameters $12\xi+\frac{m^2}{H^2}\ge\frac{9}{4}$ (such that $\nu$ vanishes or is either imaginary square-root).

\subsection{The consistency equation in the moment based approach}
\label{sec:The_consistency_equation_in_the_moment_based_approach}

Starting from the shape of $\omega_2^\textup{BD}$ from above, we can, in principle, evaluate all terms constituting $\Tmunuren$ following~\cite{Moretti:2001qh,Wald:1977up}. By inserting the de Sitter expansion $a(\tau)=\frac{1}{H\tau}$ and $\omega_2^\textup{BD}$ into the SCE, the dynamic aspect is eliminated and we obtain an equation for the parameters of the model, similarly as in~\cite{JUAREZAUBRY}. Note that the latter reference restricts to the case $H=\sqrt{\nicefrac{\Lambda}{3}}$, i.e.\ to solutions of $\Tmunuren=0$. 
One approach is to use the stress-energy tensors of a scalar field on (the entire) de Sitter space from~\cite{tadaki}. However, since we are mainly interested in the cosmological setting in which a formulation of the SCE as a dynamical system is elaborated, we follow the approach of~\cite{siemssen-gottschalk} and view the SCE on a flat FLRW space-time as a dynamical system for both the scaling factor $a$ and a sequence of `moments' derived from the state's two-point function via a specific `cosmological' parametrix.
Up to some details (cf.\ Remark~\ref{rem:Stucture_section_remark}), both approaches result in studying the very same consistency equations for the parameters. 

In the following we will shortly recapitulate both on the general quantization procedure of a scalar field with the particular goal of a well-defined, covariantly conserved $\Tmunuren$ and on the approach of~\cite{siemssen-gottschalk}. This rather serves as an introduction of relevant notation than as a complete description; for details we refer the reader to the pertinent literature cited in Section~\ref{Introduction}.

We start from the classical stress-energy tensor
\begin{align*}
T_\munu&=(1-2\xi)(\nabla_\mu\phi)(\nabla_\nu\phi)-\tfrac{1}{2}(1-4\xi)g_\munu(\nabla^\sigma\phi)(\nabla_\sigma\phi)
-\tfrac{1}{2}g_\munu m^2\phi^2\\
&\hspace{6cm}+\xi\big(G_\munu\phi^2-2\phi\nabla_\mu\nabla_\nu\phi-2g_\munu\phi\Box\phi\big)
\end{align*}
of a classical scalar field $\phi$ governed by the Klein-Gordon equation \eqref{eq:KG-eq}.
The QSE tensor of $\phi$ after quantization is then obtained by replacing $\phi$ and its derivatives with their respective quantum counterparts, that is, by the coincidence limit of (derivatives of) the regularized two-point function of a Hadamard state $\omega$. If we denote by $\widetilde{H}(y,z)$ the (possibly truncated) distributional kernel of the Hadamard parametrix, the coincidence limit of the regularized two-point function is given by $[\omega_2-\widetilde{H}]:=\lim_{z\to y}\big(\omega_2(y,z)-\widetilde{H}(y,z)\big)$ with the (unregularized) two-point function $\omega_2(y,z)=\omega(\phi(y)\phi(z))$. Accordingly, we have $[(\nabla_\mu\!\otimes\!\nabla_\nu)(\omega_2-\widetilde{H})]:=\lim_{z\to y}(\nabla_\mu)_y(\nabla_\nu)_z\big(\omega_2(y,z)-\widetilde{H}(y,z)\big)$ and so on. As usual, a conserved renormalization scheme such as Moretti's~\cite{Moretti:2001qh} is mandatory and thus $\Tmunuren$ additionally contains a trace anomaly term $\frac{1}{4\pi^2}g_\munu[\nu_1]$ (with the coincidence limit of the Hadamard coefficient $\nu_1$). By this scheme, the QSE tensor indeed obeys $\nabla^\mu\Tmunuren=0$. Finally, we add the renormalization freedom $c_1m^4 g_\munu+c_2m^2 G_\munu +c_3I_\munu +c_4 J_\munu$ in terms of four independent parameters $c_1,c_2,c_3,c_4$. We refer to~\cite{daniel-diss,Sanders} and references therein for precise formulas regarding $\widetilde{H},\nu_1,I_\munu$ and $J_\munu$. Note that for an explicit expression for $\widetilde{H}$ one has to introduce a length scale, the so-called Hadamard length scale, in order to make the arguments of some occurring logarithmic dependencies unit free. However, one purpose (among others) of the renormalization freedom is that changes in this length scale can be compensated by changes in $c_1$ and $c_2$. Concluding, the QSE tensor is of the form
\begin{align}
&\Tmunuren=\label{our_QSE_tensor}\\
&\quad(1-2\xi)\big[(\nabla_\mu\!\otimes\!\nabla_\nu)(\omega_2-\widetilde{H})\big]
-\tfrac{1}{2}(1-4\xi)g_\munu\big[(\nabla^\sigma\!\otimes\!\nabla_\sigma)(\omega_2-\widetilde{H})\big]
-\tfrac{1}{2}g_\munu m^2\big[\omega_2-\widetilde{H}\big]\notag\\
&\hspace{2.8cm}+\xi\big(G_\munu\big[\omega_2-\widetilde{H}\big]-2\big[(\mathbbm{1}\!\otimes\!\nabla_\mu\nabla_\nu)(\omega_2-\widetilde{H})\big]-2g_\munu\big[(\mathbbm{1}\!\otimes\!\Box)(\omega_2-\widetilde{H})\big]\big)\notag\\
&\hspace{2.8cm}+\frac{1}{4\pi^2}g_\munu\big[\nu_1\big]+c_1m^4 g_\munu+c_2m^2 G_\munu +c_3I_\munu +c_4 J_\munu\,.\notag
\end{align}

Back in the cosmological setting \eqref{eq:cosmological_metric}, under the assumption that $\omega_2(y,z)$ at points $y=(\tau,\mathbf{y})$ and $z=(\hat{\tau},\mathbf{z})$ merely depends on $r=|\mathbf{y}-\mathbf{z}|$, the derivatives of $\omega_2$ relevant for \eqref{our_QSE_tensor} are stored in a vector
\begin{equation}\label{eq:G-data}
  \mathcal{G}(\tau,r)
  :=
  \begin{pmatrix}
    \mathcal{G}_{\varphi\varphi}(\tau,r) \\
    \mathcal{G}_{(\varphi\uppi)}(\tau,r) \\
    \mathcal{G}_{\uppi\uppi}(\tau,r)
  \end{pmatrix}
  :=
  \lim_{\hat{\tau} \to \tau}
  \begin{pmatrix}
    \mathbbm{1} \\
    \frac{1}{2} (\partial_{\color{white}\hat{\color{black}\tau}} + \partial_{\hat{\tau}}) \\
    \partial_{\color{white}\hat{\color{black}\tau}} \partial_{\hat{\tau}}
  \end{pmatrix}
  a(\tau) a(\hat{\tau}) \omega_2(\tau,\hat{\tau},r)\,.
\end{equation}
The components of $\mathcal{G}$ can be viewed as `semi-coincidence limit' of $\omega_2$, i.e.\ the coincidence limit in direction of the $\tau$-coordinate, and these limits exist by the Hadamard property of $\omega$ whenever $r\neq0$. The singular structure of $\omega_2$ is then represented by the singular structure of $\mathcal{G}(\tau,\cdot)$ in the limit $r\to0$. The Klein-Gordon equation for $\omega_2$ implies for $\mathcal{G}$ that
\begin{equation}\label{eq:G-dynamics}
  \partial_\tau \mathcal{G} = \begin{pmatrix}
    0 & 2 & 0 \\ \Delta_r - V & 0 & 1 \\ 0 & 2(\Delta_r - V)  & 0
  \end{pmatrix} \mathcal{G}
\end{equation}
with the (spatial) Laplacian $\Delta_r = r^{-2} \partial_r r^2 \partial_r$ and the potential\footnote{Opposed to the convention in footnote \ref{footnote:dot_prime_convention}, we denote derivatives of $a$ with respect to conformal time $\tau$ by primes, i.e.\ $a'$, $a''$ and so on. Confusions in higher derivatives $a^{(j)}$ are excluded.} $V=(6\xi-1) \frac{a''}{a} + a^2 m^2$.

For an expression of the (full) coincidence limit of the regularized two-point function $[\omega_2-\widetilde{H}]$ one defines $\widetilde{\mathcal{H}}(\tau,r)$ by the analog of formula \eqref{eq:G-data} replacing $\omega_2$ by $\widetilde{H}$. Note that the Hadamard parametrix also depends only on $r=|\mathbf{y}-\mathbf{z}|$. The regularized two-point function and its derivatives from \eqref{our_QSE_tensor} may now be written as certain (linear combinations of) components of $\lim_{r\to 0}\big(\mathcal{G}(\cdot,r)-\widetilde{\mathcal{H}}(\cdot,r)\big)$, that is, they are obtained by completing the coincidence limit along the spatial coordinate directions.

The main innovation of~\cite{siemssen-gottschalk} is now to introduce a new `cosmological' parametrix $\mathcal{H}$. In principle it is constructed somewhat similar to $\widetilde{H}$ (or $\widetilde{\mathcal{H}}$, resp.), but adapted to cosmological coordinates. Formally, fix an arbitrary length scale $\mu$ and define for $j\in\mathbb{Z}_{\ge-1}$ so-called homogeneous distributions, i.e.\ the functions $h_{2j}:(0,\infty)\to\mathbb{R},$
\begin{equation*}
  h_{-2}(r) := -\frac{1}{\pi^2r^4},\hspace{.7cm}
  h_0(r)    := \frac{1}{2\pi^2r^2},\hspace{.7cm}
  h_{2j}(r) := \frac{(-1)^j}{2\pi^2} \,\frac{r^{2(j-1)}}{\Gamma(2j)} \Big( \log\Big(\frac{r}{\mu}\Big) - \psi^{(0)}(2j) \Big)
\end{equation*}
with the Digamma function $\psi^{(0)}=\log(\Gamma)'$. Moreover, define
\begin{equation}\label{eq:siemssen-gottschalk_parametrix}
    \mathcal{H}_n(\tau,r) :=
  \begin{pmatrix}
    \mathcal{H}_{\varphi\varphi,n}(\tau,r) \\
    \mathcal{H}_{(\varphi\uppi),n}(\tau,r) \\
    \mathcal{H}_{\uppi\uppi,n}(\tau,r)
  \end{pmatrix}
  :=
  \begin{pmatrix}
    0 \\ 0 \\ \gamma_{-1}(\tau)
  \end{pmatrix}
  h_{-2}(r)
  +
  \sum_{j=0}^n
  \begin{pmatrix}
    \alpha_j(\tau) \\ \beta_j(\tau) \\ \gamma_j(\tau)
  \end{pmatrix}
  h_{2j}(r)
\end{equation}
with some sequences $(\alpha_j)_{j{\ge0}}$, $(\beta)_{j{\ge0}}$, $(\gamma)_{j{\ge\textup{-}1}}$ of functions $\alpha_j,\beta_j,\gamma_j:(0,\infty)\to\mathbb{R}$. Likewise to the Hadamard parametrix, the sum may be truncated at a sufficiently large order without affecting the final result, so one can omit questions of convergence of \eqref{eq:siemssen-gottschalk_parametrix} as $n\to\infty$. In~\cite{siemssen-gottschalk} these functions are found such that $\mathcal{H}$ fulfills the Klein-Gordon system
\begin{equation}\label{eq:H-dynamics}
  \partial_\tau \mathcal{H}_\infty = \begin{pmatrix}
    0 & 2 & 0 \\ \Delta_r - V & 0 & 1 \\ 0 & 2(\Delta_r - V)  & 0
  \end{pmatrix} \mathcal{H}_\infty = \mathcal{O}(r^\infty)\,.
\end{equation}
Hereby, the function class $\mathcal{O}(r^\infty)$  is to be read as that a truncation of \eqref{eq:siemssen-gottschalk_parametrix} at some $n$ yields an error in $\mathcal{O}(r^{m(n)})$ with $m(n)\rightarrow\infty$ as $n\to\infty$.

As a next step, we replace the regularized expressions for the two-point function and its derivatives in \eqref{our_QSE_tensor}, i.e.\ $[\omega_2-\widetilde{H}]$, $[(\nabla_\mu\!\otimes\!\nabla_\nu)(\omega_2-\widetilde{H})]$ and so on, by the respective (linear combinations of) components of
\[
    \lim_{r\to0}\big(\mathcal{G}-\widetilde{\mathcal{H}}\big)=\lim_{r\to0}\big(\mathcal{G}-\mathcal{H}\big)+\lim_{r\to0}\big(\mathcal{H}-\widetilde{\mathcal{H}}\big)\,,
\]
where indeed both limits on the RHS do exist. $\lim_{r\to0}\big(\mathcal{H}-\widetilde{\mathcal{H}}\big)$ does not depend on the precise choice of $\omega$, and can be computed to result in a smooth function on the underlying space-time depending on $a$ and its derivatives only. $\mathcal{G}-\mathcal{H}$ may be called the regularized two-point function in the $\mathcal{H}$-regularization scheme. Finally, define the moments of $\omega$ by
\[
    \mathcalm_n :=\lim_{r\to 0} \Delta_r^n \big(\mathcal{G} - \mathcal{H} \big)\,,
\]
that is, they can be thought of as (even-order, radial) Taylor coefficient of the (radially symmetric) function $\mathcal{G}(\tau,\cdot) - \mathcal{H}(\tau,\cdot)$. With \eqref{eq:G-dynamics} and \eqref{eq:H-dynamics} also $\mathcal{G} - \mathcal{H}$ fulfills the ($\mathcal{O}(r^\infty)$-approximate) Klein-Gordon system \eqref{eq:H-dynamics} and we can reformulate the latter into a linear evolution equation for the function $\tau\mapsto\mathcalm(\tau)=(\mathcalm_0(\tau),\mathcalm_1(\tau),\mathcalm_2(\tau),\dots)$ valued in a suitable Banach space of sequences. In this setting the SCE can be written as
\begin{equation}
    \begin{cases}
        A'(\tau)&=V\big(A(\tau),\mathcalm(\tau)\big)\\
        \mathcalm'(\tau)&=W\big(A(\tau)\big)\cdot\mathcalm(\tau)
    \end{cases}\label{dynamical-system}
\end{equation}
with $A=(a,a',a'',a''')$ and some dynamic vector fields $V$ and $W$, and the authors of~\cite{siemssen-gottschalk} prove existence and uniqueness of the solutions. Note that the second line of \eqref{dynamical-system} is a mere consequence of the Klein-Gordon equation, particularly it is independent of whether $a$ is a solution to any cosmological model or not. The first line of \eqref{dynamical-system} is, usually, obtained from the traced SCE, constrained by the energy equation \eqref{abstract_trace_n_energy_equation}. However, by the discussion in Section~\ref{Section:The-energy-equation-as-a-cosmological-model} we can use the energy equation for the first line of \eqref{dynamical-system} in the context of a pure de Sitter expansions (where $\dot{a}\neq 0$ as well as $a'\neq0$ hold globally). In this case, $A$ is of the form $A=(a,a',a'')$.

Performing the replacements described above, the energy evaluates to
\begin{align}
  0 =& \left( 6 (3 c_3 + c_4) + \frac{1}{960\pi^2} - \frac{6\xi-1}{96\pi^2} - \frac{(6\xi-1)^2}{32\pi^2} \log(a \lambda_0) \right)\notag%
  \\\notag
  &\hspace{6.5cm}\cdot \left( 2 \frac{a^{(3)} a'}{a^4} - \frac{(a'')^2}{a^4} - 4 \frac{a'' (a')^2}{a^5} \right)
  \\
  & - \frac{(6\xi-1)^2}{16\pi^2} \frac{a''(a')^2}{a^5} + \frac{1}{960\pi^2} \frac{(a')^4}{a^6} + \left(\frac{\Lambda}{\kappa}-m^4 \left( c_1 + \frac{1}{32\pi^2} \log(a \lambda_0) \right)\right) a^2 \label{daniel-hanno-energy-eq}
  \\\notag
  &+ \left( -\frac{3}{\kappa} + m^2 \Bigl( 3 c_2 - \frac{1}{96\pi^2} - \frac{6\xi-1}{16\pi^2} \big(1 + \log(a \lambda_0) \big) \Big) \right) \frac{(a')^2}{a^2}
  \\\notag
  & + \frac{m^2}{2} \mathcalm_{\varphi\varphi,0} + (6\xi-1) \left( -\frac{(a')^2}{2a^4} \mathcalm_{\varphi\varphi,0} + \frac{a'}{a^3} \mathcalm_{(\varphi\uppi),0} \right)
  \\\notag
  & \hspace{7.5cm} + \frac{1}{2a^2} \big( \mathcalm_{\uppi\uppi,0} - \mathcalm_{\varphi\varphi,1} \big)\,,
\end{align}
where $\lambda_0$ is the ratio of the Hadamard length scale and the length scale $\mu$. For more details we refer the interested reader to~\cite{siemssen-gottschalk}. Note that to obtain \eqref{daniel-hanno-energy-eq} from the latter reference we have merely included a possibly non-vanishing $\Lambda$.

As a side remark, note that the representation in \eqref{daniel-hanno-energy-eq} nicely shows how precise choices of both the involved length scales do not matter and any changes can be absorbed into the renormalization constants.

A more or less straightforward computation yields the first moments of the Bunch-Davies state's two-point function \eqref{eq:BD_tpf} on a cosmological de Sitter space-time with $a(\tau)=\frac{1}{H\tau}$ as
{\allowdisplaybreaks
\begin{align*}
    \mathcalm_{\varphi\varphi,0}=&-\frac{2(1-6\xi)H^2-m^2}{16\,H^2\pi^2\tau^2}~\bigg[1+\log\big(\tfrac{\mu^2}{4\tau^2}\big)+\psi^{(0)}\big(\tfrac{3}{2}-\nu\big)+\psi^{(0)}\big(\tfrac{3}{2}+\nu\big)\bigg]\\
    \mathcalm_{\varphi\varphi,1}=&-\frac{2(1-6\xi)H^2-m^2}{128\,H^2\pi^2\tau^4}\\
    &\cdot\bigg[18+84\xi+7\tfrac{m^2}{H^2}+6(12\xi+\tfrac{m^2}{H^2})\Big(\log\big(\tfrac{\mu^2}{4\tau^2}\big)+\psi^{(0)}\big(\tfrac{3}{2}-\nu\big)+\psi^{(0)}\big(\tfrac{3}{2}+\nu\big)\Big)\bigg]\\
    \mathcalm_{(\varphi\uppi),0}=&\frac{2(1-6\xi)H^2-m^2}{16\,H^2\pi^2\tau^3}~\bigg[2+\log\big(\tfrac{\mu^2}{4\tau^2}\big)+\psi^{(0)}\big(\tfrac{3}{2}-\nu\big)+\psi^{(0)}\big(\tfrac{3}{2}+\nu\big)\bigg]\\
    \mathcalm_{\uppi\uppi,0}=&-\frac{2(1-6\xi)H^2-m^2}{128\,H^2\pi^2\tau^4}\\
    &\!\!\cdot\!\bigg[30+12\xi+\tfrac{m^2}{H^2}+2(4+12\xi+\tfrac{m^2}{H^2}) \Big(\log\big(\tfrac{\mu^2}{4\tau^2}\big)+\psi^{(0)}\big(\tfrac{3}{2}-\nu\big)+\psi^{(0)}\big(\tfrac{3}{2}+\nu\big)\Big)\bigg],
\end{align*}
}{\parindent0pt and for the combination of moments relevant in \eqref{daniel-hanno-energy-eq} we compute}
\begin{align}
    &\frac{m^2}{2} \mathcalm_{\varphi\varphi,0} + (6\xi-1) \left( -\frac{(a')^2}{2a^4} \mathcalm_{\varphi\varphi,0} + \frac{a'}{a^3} \mathcalm_{(\varphi\uppi),0} \right)
   + \frac{1}{2a^2} \big( \mathcalm_{\uppi\uppi,0} - \mathcalm_{\varphi\varphi,1} \big)\label{moment-content-of-energy-equation}\\
   &\!=\hspace{-.5pt}\frac{2(1-6\xi)H^2-m^2}{128\pi^2H^2\tau^2}\hspace{-.5pt}\bigg[6(1-6\xi)H^2\hspace{-.25pt}-\hspace{-.25pt}m^2\hspace{-.25pt}-\hspace{-.25pt}2m^2\Big(\hspace{-.5pt}\log(\tfrac{\mu^2}{4\tau^2})+\psi^{(0)}(\tfrac{3}{2}-\nu)+\psi^{(0)}(\tfrac{3}{2}+\nu)\hspace{-.5pt}\Big)\hspace{-.5pt}\bigg]\hspace{-.5pt}.\notag
\end{align}
Recall that $\nu
=\sqrt{\frac{9}{4}-12\xi-\frac{m^2}{H^2}}$ and note that the Digamma function $\psi^{(0)}$ stems from taking derivatives of the hypergeometric function ${}_2F_1$ in the Bunch-Davies two-point function \eqref{eq:BD_tpf} and not from the occurrence of $\psi^{(0)}$ in the cosmological parametrix via the homogeneous distributions $(h_{2j})_{j\ge\textup{-}1}$. Moreover, note how all moments and, in particular, the contribution \eqref{moment-content-of-energy-equation} vanish in the massless conformally coupled case $m^2=\xi-\frac{1}{6}=0$.

Finally, plugging both the de Sitter Ansatz $a(\tau)=\frac{1}{H\tau}$ and the Bunch-Davies moments \eqref{moment-content-of-energy-equation} into the energy equation \eqref{daniel-hanno-energy-eq}, we arrive at the following form of the consistency equation
\begin{align}
   0=&\left(\frac{1}{960}- \frac{(6\xi-1)^2}{32}\right) H^4-3KH^2+\Lambda K + m^4\, d_1  +   m^2H^2\, d_2
  \label{finalenergyequation}\\
  &\hspace{2cm}+\left(\frac{m^4}{64}+\frac{6\xi-1}{32}m^2H^2\right)\Big(2\log(\mu H)+\psi^{(0)}(\tfrac{3}{2}-\nu)+\psi^{(0)}(\tfrac{3}{2}+\nu)\Big)\,.\notag
\end{align}
Here we have introduced linearly transformed renormalization constants
\[
    d_1:=\frac{1}{128}-\frac{1}{32}\log(2\lambda_0)-\pi^2\,c_1\qquad\textup{as well as}\qquad d_2:=3\pi^2\,c_2 - \frac{1}{96}-\frac{6\xi-1}{16}\log(2\lambda_0)
\]
and we have set $K:=\frac{\pi^2}{\kappa}$. We note that $I_{00}=J_{00}=0$ for de Sitter expansions $a(\tau)=\frac{1}{H\tau}$, thus we are left with only two renormalization constants. Moreover, the $\log(\tau)$-terms of \eqref{moment-content-of-energy-equation} just cancel the $\log(\tau)$-terms occurring in \eqref{daniel-hanno-energy-eq}. This is not surprising if we recall that all of them originate in the cosmological parametrix $\mathcal{H}$. We are left only with the length-scale $\mu$ and observe (again) that a different choice of $\mu$ leads to additional terms which can be absorbed into the renormalization constants $d_1$ and $d_2$.

\section{De Sitter solutions for the massless field}
\label{Sec:De_Sitter_Solutions_for_the_Massless_Field}

In the massless case $m=0$, the consistency equation \eqref{finalenergyequation} breaks down into the polynomial equation
\begin{equation}
   0=\left(\frac{1}{960}- \frac{(6\xi-1)^2}{32}\right) H^4-3KH^2+\Lambda K\label{energy-eq-massless-pos-Lambda}
\end{equation}
for $\xi$ and $H$ with parameters $\Lambda$ and $K$.
\begin{remark}\label{rem:massless_remark}
\begin{itemize}
    \item[\textup{(i)}] Throughout the present section we assume $\xi>0$ in order to fulfill the existence condition for the Bunch-Davies state.
    \item[\textup{(ii)}] Some may take the standpoint that the prefactors $m^2$ $($of $G_\munu)$ and $m^4$ $($of $g_\munu)$ in the renormalization freedom of $\Tmunuren$ are only chosen to endow the respective terms with the correct unit in order to obtain unit-free renormalization constants $c_1$ and $c_2$. Consequently, in the massless case $m=0$ these prefactors should be expressed by some other mass scale $\widetilde{m}$ in order to maintain this freedom. However, replacing the coupling constant $K$ and the cosmological constant $\Lambda$ in equation \eqref{energy-eq-massless-pos-Lambda} by their renormalized analogs,
    \begin{equation*}
        \widetilde{K}=K-\frac{\widetilde{m}^2d_2}{3}\qquad\text{and}\qquad\widetilde{\Lambda}=\frac{3\Lambda K+3\widetilde{m}^4d_1}{3K-\widetilde{m}^2d_2}\,,
    \end{equation*}
    respectively, we end up with the very same equation with the only difference that  $\widetilde{K}$ can be a non-positive parameter. Note also how \textup{\eqref{energy-eq-massless-pos-Lambda}} simplifies for $K=0$ and how for $K<0$ the analysis of \textup{\eqref{energy-eq-massless-pos-Lambda}} is $($in a suitable sense$)$ inverted around the zeros $\frac{1}{6}\pm\nicefrac{1}{\sqrt{1080}}$ of the $H^4$-prefactor. However, we stick to the view that the renormalization freedom compensates ambiguities in the choice of the Hadamard length scale, that is, we always assume $K>0$.
    \item[\textup{(iii)}] Carrying out the computation to obtain \eqref{energy-eq-massless-pos-Lambda} from \eqref{daniel-hanno-energy-eq} and \eqref{moment-content-of-energy-equation} one notices that the Bunch-Davies state's moments only yield a contribution into the prefactor of $H^4$ in \eqref{energy-eq-massless-pos-Lambda}. Assuming vanishing moments $\mathcalm=0$ instead, the fraction ${(6\xi-1)^2}/\raisebox{-1pt}{$32$}$ would be replaced by ${(6\xi-1)^2}/\raisebox{-1pt}{$8$}$. Hence, the analysis of the present section can be adjusted for the Minkowski vacuum-like states in \textup{\cite{Numerik-Paper}} by squeezing any graphic by a factor of $\frac{1}{2}$ around $\xi=\frac{1}{6}$. Note in particular the similarity between the $\Lambda=0$-curve in Figure~\textup{\ref{fig:nonpos-Lambda-plot}} and the respective graphic in \textup{\cite{Numerik-Paper}}.
\end{itemize}
\end{remark}
We introduce
\[
    \xi_\textup{cc}:=\frac{1}{6}\,,\qquad\xi_{(\pm)}:=\xi_\textup{cc}\pm\frac{1}{\sqrt{1080}}\qquad\text{as well as}\qquad H_\textup{vac}:=\sqrt{\frac{\,\Lambda\,}{3}}
\]
(for $\Lambda>0$) since these particular $\xi$- and $H$-values are distinguished by the behavior of the solution set of \eqref{energy-eq-massless-pos-Lambda}. Here, $H_\textup{vac}$ represents the unique (positive) de Sitter solution of the vacuum Einstein equation $G_\munu+\Lambda g_\munu=0$ for $\Lambda>0$.

We find the following:

\begin{proposition}\label{prop:massless_case}
Let $\Lambda\in\R$, $K>0$. The set of de Sitter solutions of the SCE for these parameters with a scalar field in the Bunch-Davies state can be parameterized for $\Lambda\le0$ by one and for $\Lambda>0$ by two analytic curves in the $(\xi,H)$-parameter plane $(0,\infty)\times(0,\infty)$.

Moreover:
\begin{itemize}
    \item[\textup{(i)}] If $\frac{\Lambda}{K}>2160$ the two solution curves can be globally solved for $\xi$. Denoting $H_\textup{min}=\big(\tfrac{K}{29}\big(1440^2+29\cdot960\tfrac{\Lambda}{K}\big)^{\nicefrac{1}{2}}-\frac{1440K}{29}\big)^{\nicefrac{1}{2}}~\in(0,H_\textup{vac})$, the solution curves are the $($disjoint$)$ graphs of the functions

    $\Xi_{(+)}:(0,\infty)\to(0,\infty)$, $\Xi_{(-)}:(H_\textup{min},\infty)\to(0,\infty)$,
    \begin{equation}
        \Xi_{(\pm)}(H)=\frac{1}{6}\pm\sqrt{\frac{1}{1080}-\frac{8K}{3H^2}+\frac{8\Lambda K}{9H^4}}~.\label{xi-of-H-formula}
    \end{equation}
    Hereby, $\Xi_{(-)}$ is restricted to $(H_\textup{min},\infty)$ in order to be positive-valued.
    
    In particular, any arbitrary $H>0$ is the Hubble rate of a de Sitter solution to the SCE for one or two suitable value$($s$)$ for $\xi$, with two possible values if and only if $H>H_\textup{min}$. On the other hand, for any $\xi\in\big(\max(\Xi_{(-)}),\min(\Xi_{(+)})\big)~\big(\ni\xi_\textup{cc}\big)$ there exist no de Sitter solution $H$ at all.

    \item[\textup{(ii)}] If $\Lambda\in(0,2160K)$ the two solution curves can be globally solved for $H$, that is, they are the $($disjoint$)$ graphs of the $($analytic continuations of the$)$ functions

    $H_{(+)}:(\xi_{(-)},\xi_{(+)})\to(0,\infty)$, $H_{(-)}:(0,\infty)\to(0,\infty),$
    \begin{equation}
        H_{(\pm)}(\xi)=\sqrt{1440\,K~\frac{1\pm\sqrt{1-\frac{\Lambda}{2160\,K}(1-30(6\xi-1)^2)~}~}{1-30(6\xi-1)^2}~}~.\label{H-of-xi-formula}
    \end{equation}

    In particular, for any $\xi\in(\xi_{(-)},\xi_{(+)})$ there exist precisely two de Sitter solutions $H_{(-)}(\xi)$ and $H_{(+)}(\xi)$, while for any other $\xi>0$ there exists precisely one de Sitter solution $H_{(-)}(\xi)$. On the other hand, for any value $H>0$ with $H\notin\big(\max(H_{(-)}),\min(H_{(+)})\big)$ $\neq\emptyset$ there is a $\xi$-value such that $H$ is the Hubble rate of a de Sitter solution to the SCE.

    \item[\textup{(iii)}] If $\nicefrac{\Lambda}{K}=2160$ Equation \eqref{energy-eq-massless-pos-Lambda} is equivalent to
    \begin{equation}
        \xi-\frac{1}{6}=\pm\frac{1}{\sqrt{1080}}\left(1-2\,\frac{H_\textup{vac}^2}{H^2}\right)\label{eq:simplified_consistency_eq_case_iii}
    \end{equation}
    and thus can be globally solved for either $H$ or $\xi$ at will. Hence, the solution set is the union of the graphs of two bijective functions
    ${(0,\xi_{(+)})\to(H_\textup{min},\infty)}$ and\linebreak ${(\xi_{(-)},\infty)\to(0,\infty)}$ mapping $\xi\mapsto H(\xi)$ $(H_\textup{min}$ as in \textup{(i)}$)$ or their inverses, respectively. The mappings are $($piecewise$)$ degenerations of \textup{\eqref{xi-of-H-formula}} and \textup{\eqref{H-of-xi-formula}} for the present ratio $\frac{\Lambda}{K}$ and the graphs intersect only in $(\xi_\textup{cc},\sqrt{2}H_\textup{vac})$.

    In particular, for any $H\in (H_\textup{min},\infty)\backslash\{\sqrt{2}\,H_\textup{vac}\}$ there exist two values of $\xi$ to yield $H$ as the corresponding de Sitter solution, whereas for any $H\in(0,H_\textup{min}]\cup\{\sqrt{2}H_\textup{vac}\}$ there exists precisely one such $\xi$-value. On the other hand, for any $\xi>0$ there exists $($at least$)$ one de Sitter solution $H$ and a second solution exists if and only if $\xi\in(\xi_{(-)},\xi_{(+)})\backslash\{\xi_\textup{cc}\}$.

    \item[\textup{(iv)}] If $\Lambda\le0$ the single solution curve can be globally solved for $H$, that is, it is the graph of the function $H_{(+)}:(\xi_{(-)},\xi_{(+)})\to(0,\infty)$ defined in \textup{\eqref{H-of-xi-formula}}.

    In particular, each $H>\min(H_{(+)})$ is the de Sitter solution for precisely two $\xi\in(\xi_{(-)},\xi_{(+)})$, $H=\min(H_{(+)})$ is the unique de Sitter solution for $\xi=\xi_\textup{cc}$ and any $H<\min(H_{(+)})$ does not yield a solution of our model. On the other hand, for any $\xi\in(\xi_{(-)},\xi_{(+)})$ there exists one de Sitter solution, whereas otherwise there exists none at all.
\end{itemize}
\end{proposition}

Note that every assertion of the previous proposition follows from studying \eqref{energy-eq-massless-pos-Lambda} as a quadratic equation for $\xi$ and $H^2$ and we skip the proof. Rather, we will concentrate on a further description of the solution sets, in particular their asymptotes and some physically relevant properties.

\begin{figure}[t!]
\centering
    \includegraphics{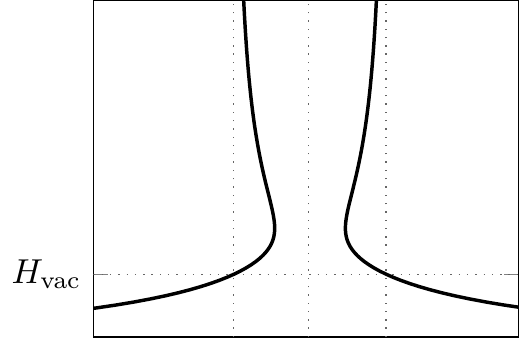}~
    \includegraphics{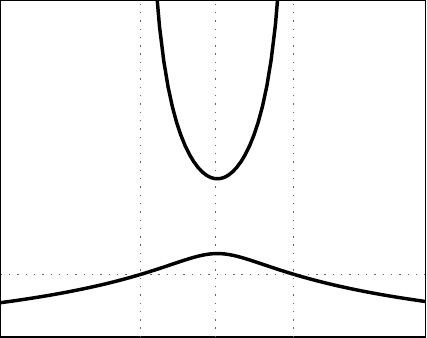}~
    \includegraphics{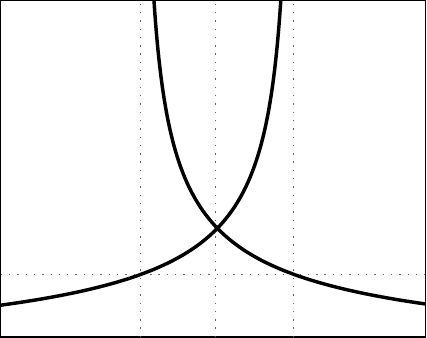}~

    \hspace{1cm}(i)  $\frac{\Lambda}{K} > 2160$
    \hspace{2cm} (ii) $0<\frac{\Lambda}{K} < 2160$
    \hspace{1.5cm} (iii) $\frac{\Lambda}{K} = 2160$
    \hspace{.5cm}

    \begin{minipage}{.9\textwidth}
        \captionsetup{format=plain, labelfont=bf}
        \caption{Schematic plots for the quartic solution curves of \eqref{energy-eq-massless-pos-Lambda} in Cases (i), (ii) and (iii) of Proposition~\ref{prop:massless_case}. The horizontal axis marks $\xi$ and the vertical dotted lines lie at $\xi\in\{\xi_\textup{cc},\xi_{(\pm)}\}$. The vertical axis marks $H$, normalized to $H_\textup{vac}$. Note that in each graphic the solution curves intersect the $\xi=0$-axis at $(\xi,H)=(0,H_\textup{min})$ with the respective value of $H_\textup{min}$.} \label{schematic-H-of-xi}
    \end{minipage}
\end{figure}

Figure~\ref{schematic-H-of-xi} shows a plot of the solution set in the Cases (i) - (iii) of Proposition~\ref{prop:massless_case}. The vertical axis was rescaled by $H_\textup{vac}$ in order to show how for any $\Lambda>0$ the respective only solution at $\xi=\xi_{(\pm)}$ lies at $H_\textup{vac}$, independently of $K$. More general, note how $H_{(\pm)}/H_\textup{vac}$ from \eqref{H-of-xi-formula} only depends on the ratio $\frac{\Lambda}{K}$, that is, the qualitative shape of the solution sets also only depends on that single parameter.

Figure~\ref{fig:nonpos-Lambda-plot} shows the solution sets for Case (iv) of Proposition~\ref{prop:massless_case}. The horizontal axis remains as in Figure~\ref{schematic-H-of-xi}, but the vertical axis is now rescaled by $\sqrt{\kappa}$. The thick curve marks the boundary case $\Lambda=0$, whereas the thin curves mark some solution curves for larger and larger negative $\frac{\Lambda}{K}$. The gray curve marks one solution set for parameters obeying Case (ii) in Proposition~\ref{prop:massless_case} for reference. Note how the shape of the $H_{(+)}$-branch of solutions is rather unaffected by $\Lambda$ changing from positive to negative and the boundary case $\Lambda=0$ is (in a suitable sense) continuously embedded.

\begin{figure}
    \centering
    \begin{minipage}{.441\textwidth}
    \includegraphics{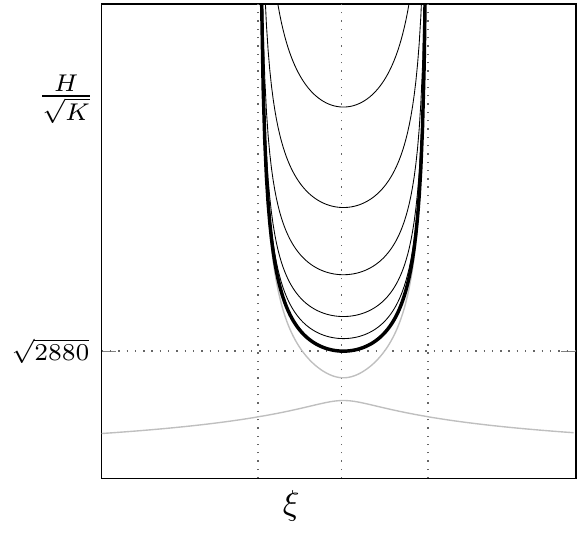}
    \end{minipage}
    \hspace{.3cm}
    \begin{minipage}{.52\textwidth}
    \captionsetup{format=plain, labelfont=bf}
    \caption{The curves $\nicefrac{H}{\sqrt{K}}$ as a function of $\xi$ with different values of $\frac{\Lambda}{K}<0$ (Case (iv) of Proposition~\ref{prop:massless_case}). The thick curve shows the $\Lambda=0$ case, whereas the other black curves show the curves for\label{fig:nonpos-Lambda-plot}}

    \vspace{-.6cm}
    \[
    \frac{\Lambda}{K}\in\{\textup{-}2160,\textup{-}8649,\textup{-}34560,\textup{-}138240,\textup{-}552960\},
    \]

    \vspace{-.2cm}
    as from bottom to top. For reference, the gray curve shows a positive-$\Lambda$ curve with $\frac{\Lambda}{K}=2025$ (Case (iii)). The vertical dotted lines mark the same distinguished $\xi$-values as in Figure~\ref{schematic-H-of-xi}. Note the similarity of the $\Lambda=0$-curve with the respective graphic for the `tow-in' states in~\cite{Numerik-Paper}.

    \end{minipage}

\end{figure}

On the one hand, for $\Lambda>0$ the solution set has the asymptote $H=0$ as $\xi\to\infty$ which to leading order is given by
\[H_{(-)}(\xi)=\left(\frac{8K}{\Lambda}\right)^{\nicefrac{1}{4}}\frac{H_\textup{vac}}{\sqrt{\xi}}+\mathcal{O}(\xi^{-\nicefrac{3}{2}})\]
in said limit. In particular, the stronger a scalar field couples to the metric's curvature, the more it compensates the effect of a fixed positive value $\Lambda>0$ to yield classical `Dark Energy' solutions with expansion rate $H_\textup{vac}$.

On the other hand we have, for any value of $\Lambda$, the asymptotes $\xi=\xi_{(\pm)}$ and $H$ diverges as $\xi$ approaches $\xi_{(-)}$ from above or as $\xi$ approaches $\xi_{(+)}$ from below, respectively. In particular, by tuning the parameter $\xi$ around said values, one obtains arbitrarily large values of $H$ to yield a de Sitter solution of the SCE. As noted above, for positive $\Lambda$, there exists a second (continuous) solution branch around these $\xi$-values defined by $H_{(-)}$ which in particular fulfills $H_{(-)}(\xi_{(\pm)})=H_\textup{vac}$. This observation suggest referring to the lower solution branch (around $H_\textup{vac}$) to the (semi)classical solution branch in the sense of a classical solution plus quantum corrections, which notably exists if and only if the classical solution exists. Moreover, this observation suggest referring to the upper (divergent) solution branch as quantum solution branch in the sense that they exist independently of the presence of the classical solution (i.e.\ for all $\Lambda$) and that they have no classical analog. Note that while the classical and quantum solution branches are clearly separated whenever $\frac{\Lambda}{K}<2160$ (Cases (ii) and (iv)), they degenerate in $\xi_\textup{cc}$ for $\frac{\Lambda}{K}=2160$ (Case (iii)) and even annihilate each other around $\xi_\textup{cc}$ if $\frac{\Lambda}{K}>2160$ (Case (i)). However, around the values $\xi_{(\pm)}$ the separation remain valid in all cases. We will pick up these solutions branches in the discussion of Section~\ref{sec:Inflationary_models}.

Related results have been observed in the literature: The solutions at $\xi=\xi_ {(\pm)}$, namely $H=H_\textup{vac}$, were previously found in~\cite{JUAREZAUBRY} together with the fact that $H=H_\textup{vac}$ is a solution only for the aforementioned $\xi$-values. On the other hand, de Sitter solutions at $\xi=\xi_\textup{cc}$ were found before by many authors employing a variety of states or approximations of states. For example, Starobinski~\cite{Starobinsky} found what we called ``quantum solution'' for $\xi=\xi_\textup{cc}$ using the Bunch-Davies state (synonymously referring to it as ``de Sitter state''). Moreover, the authors of~\cite{DEfromQFT} found the analogs of both what we called the quantum and the classical solution using approximate KMS states. At third, in~\cite{Numerik-Paper} the authors of the present work found the quantum solution using massless Minkowski-like vacuum states (i.e.\ states with $\mathcalm=0$), and by introducing a positive cosmological constant also the classical solution would appear (cf.\ Remark~\ref{rem:massless_remark}.(iii)). Note that for a conformally coupled field, however, the state merely contributes geometric terms, hence the precise choice of a (Hadamard) state does not matter. How the two regimes around $\xi_{(\pm)}$ and $\xi_\textup{cc}$ in turn are connected was, to the authors' knowledge, not observed before.

\section{The solution set of de Sitter solutions for the massive field}
\label{positive-mass-simplification}
In this section we also consider the energy equation \eqref{finalenergyequation} as a consistency constraint on points $(\xi,H)\in\R\times(0,\infty)$ with parameters $m,~K,~\Lambda,~\mu,~d_1$ and $d_2$. In contrast to the previous section, this consistency equation is no more an explicitly solvable polynomial equation and one major task is the treatment of the incomparably more complicated dependency of the energy density $\Tmunurenen$ on the parameters via the Bunch-Davies moments \eqref{moment-content-of-energy-equation}. We also remark that negative values for $\xi$ are allowed as long as $m^2+12\xi H^2>0$ (cf.\ Section~\ref{sec:De_Sitter_space_and_the_Bunch-Davies_state}).

At first we reformulate the consistency equation \eqref{finalenergyequation} into a shape tailored to the massive case. In particular we identify two (effective) parameters $e_1$ and $e_2$ which influence the shape of the solution set in the $(\xi,H)$-plane and we point out how the remaining parameters merely rescale the solution set in said plane.

As a next step, in analogy to Proposition~\ref{prop:massless_case}, we prove how the solution set in the $(\xi,H)$-plane can be parameterized by analytic curves and how each such curve must hit the boundary of admissible $(\xi,H)$-points at both ends. Partial results are relegated to Subsections~\ref{Section-Overall-solution-counting} - \ref{section-variety-reduction} and the proof is concluded in Subsection \ref{section_Proof_of_structure_thm}. The asymptotics of the solution set, particularly how many curves constitute the solution set, will be analyzed in the next section, Section~\ref{Asymptotic-Behavior-of-the-Solution-Set}.

In order to simplify the consistency equation \eqref{finalenergyequation}, denote
\begin{equation}
    f:(0,\infty)\to\R,~x\mapsto\psi^{(0)}\Big(\tfrac{3}{2}-\sqrt{\tfrac{9}{4}-x}\,\Big)+\psi^{(0)}\Big(\tfrac{3}{2}+\sqrt{\tfrac{9}{4}-x}\,\Big)\,.\label{Digamma-dependency}
\end{equation}
A plot of $f$ on the relevant domain is shown in Figure~\ref{plot-of-f}.(i) below and a few useful properties of it are listed in Appendix~\ref{appendix-function-f}.
By introducing a shifted curvature coupling
\begin{equation}
    x=12\xi+\frac{m^2}{H^2}\label{eq:def_of_x}
\end{equation}
(note how the field equation reads $(\Box+xH^2)\phi=0$\,), we can rewrite the Digamma function terms occurring in the consistency equation \eqref{finalenergyequation} into $\psi^{(0)}(\tfrac{3}{2}-\nu)+\psi^{(0)}(\tfrac{3}{2}+\nu)=f(x)$. Moreover, regarding \eqref{finalenergyequation} as an equation for $x$ instead of $\xi$ simplifies the domain of our problem to such $(x,H)$-points where both $x>0$ and $H>0$. Therefore we note that $x>0$ is equivalent to the positivity of the effective de Sitter mass $xH^2=m^2+12\xi H^2$, and hence equivalent to the existence of the unique $O(4,1)$-invariant Bunch-Davies state on the de Sitter space-time encoded by $H$ (as discussed in Section \ref{sec:De_Sitter_space_and_the_Bunch-Davies_state}).

We recall that changes in the length scale $\mu$ can be absorbed into the renormalization constants $d_1$ and $d_2$. Hence, we eliminate the parameters $m^2$ and $\mu$ by setting $\mu=m$ and rewriting the energy equation in terms of $h=\frac{H}{m}$.

Finally, we define new parameters
\begin{equation}\label{eq:massive_renormalization_parameters_e1_e2}
    e_1=960\Big(d_1+\frac{\Lambda K}{m^4}\Big)-\tfrac{15}{2}\qquad\textup{and}\qquad e_2=960\Big(d_2-3\frac{K}{m^2}\Big)
\end{equation}
and note that $e_1$ and $e_2$ are still linear transformations of the original renormalization freedoms $c_1$ and $c_2$, respectively. In particular, they can be arbitrary real numbers.

Using the above simplifications, we rewrite the consistency equation \eqref{finalenergyequation} into
\begin{align}
    \hspace{1.5cm}F(x,h)&=0\label{energy-equation-massive-case}
\intertext{with the functions $F,F_1,F_2:(0,\infty)\times(0,\infty)\to\R,$}
    F_1(x,h)&=\Big(\frac{x^2}{4}-x+\frac{29}{30}\Big)h^2-\Big(\frac{x}{2}+\frac{e_2}{30}-1\Big)-\frac{e_1}{30}\,\frac{1}{h^2}\notag\\[2pt]
    &=\Big(\frac{x}{2}-1\Big)^2h^2-\Big(\frac{x}{2}-1\Big)-\frac{1}{30h^2}\big(h^4+e_2h^2+e_1\big)\notag\\[2pt]
    &=\frac{h^2}{4}\,x^2-\Big(h^2+\frac{1}{2}\Big)x+\frac{29h^2}{30}-\frac{e_2}{30}-\frac{e_1}{30h^2}+1\,,\notag\\[3pt]
    F_2(x,h)&=2\log(h)+f(x)\,,\notag\\[2pt]
    F(x,h)&= F_1(x,h)-\big(\tfrac{x}{2}-1\big)F_2(x,h)\,.\notag
\end{align}
We denote the zero set of $F$, that is, the solution set of the consistency equation \eqref{energy-equation-massive-case}, by
\begin{align*}
    \mathcal{S}_{e_1,e_2}&:=\big\{\,(x,h)\in(0,\infty)\times(0,\infty)\,\big|\,F(x,h)=0\textup{ for the parameters }e_1,e_2\,\big\}\\[2pt]
    &\phantom{:}\subset(0,\infty)\times(0,\infty)\,.
\end{align*}
Due to the analyticity of $f$ discussed in Appendix~\ref{appendix-function-f}, also $F_1,F_2$ and $F$ are analytic and $\mathcal{S}_{e_1,e_2}$ is an analytic variety.
\begin{remark}\label{rem:Stucture_section_remark}
    We have remarked in Section~\textup{\ref{sec:The_consistency_equation_in_the_moment_based_approach}} that the consistency equation can also be derived by using the system of the $($non-pullback$)$ Bunch-Davies state on $($the entire$)$ de Sitter space as an Ansatz for the SCE, using the stress-energy tensor derived in \textup{\cite{tadaki}}.  While in the massless case this is straightforward, for a positive mass, which does not eliminate the renormalization freedoms $c_1$ and $c_2$/$d_1$ and $d_2$, the parameters $e_1$ and $e_2$ need to be defined different from \eqref{eq:massive_renormalization_parameters_e1_e2}. 
\end{remark}

In the following we study solutions of \eqref{energy-equation-massive-case}, where the present section is dedicated to showing that the analytic variety $\mathcal{S}_{e_1,e_2}$ can be decomposed into non-singular subvarieties.

We introduce the distinguished $x$-values
\[
    x_{(\pm)}:=2\pm\sqrt{\frac{2}{15}}
\]
and note that $x(\xi_{(\pm)})=12\xi_{(\pm)}+\frac{1}{h^2}\to x_{(\pm)}$ in the limit $h\to\infty$, that is, in said limit these distinguished $x$-values correspond to the $\xi$-values distinguished in the massless case (cf.\ Section~\ref{Sec:De_Sitter_Solutions_for_the_Massless_Field}).

Moreover, inspired by graph theory we introduce the following notion. Note that, whenever we speak of a \emph{connected} set we mean a path-connected set, that is, a set such that any pair of points from that set is connected by a continuous curve contained in that set.
\begin{definition}
A subset $\mathcal{S}\subset(0,\infty)\times(0,\infty)$ is called tree-like if it is connected and for each $s\in \mathcal{S}$ the set $\mathcal{S}\backslash\{s\}$ is disconnected with finitely many connected components.
\end{definition}

In order to exclude pathological counter examples of curves that ``turn around'' (such as e.g.\ the analytic curve $(-\varepsilon,\varepsilon)\to(0,\infty)\times(0,\infty),\,x\mapsto(1,1+x^2)$\,) we restrict to regularly parameterized curves.  A parameterization is called regular, if the absolute value of its derivative is bounded away from zero.

The following theorem characterizes the analytic variety $\mathcal{S}_{e_1,e_2}$. Note again how contrary to the massless case the occurrence of the function $f$ prevents a simple, closed expression for solutions in  $\mathcal{S}_{e_1,e_2}$ and we approach the equation using certain properties of $f$ such as its analyticity, its asymptotic behavior or certain bounds on its derivatives. 

Moreover, we need a technical assumption on the Hessian of $F$ in order to proof the following theorem.
\begin{assumption}\label{ass:first}
    Let $e_1,e_2\in\mathbb{R}$. Suppose that for all $(x,h)\in(0,\infty)\times(0,\infty)$ either $\nabla F(x,h)\neq0$ or $\det\big(\textup{Hess}\,F(x,h)\big)<0$.
\end{assumption}
\begin{remark}
In Section \textup{\ref{Section-Non-existence-of-local-extrema}}, we will substantiate the previous assumption by proving it in the asymptotic regimes of large and small $h$-values and, thereafter, underpinning it with numerical evidence in the intermediate regime. In Section \textup{\ref{Section-Non-existence-of-local-extrema}} we study Assumption \textup{\ref{ass:firstprime}} as an equivalent version of Assumption~\textup{\ref{ass:first}} $($cf.\ the numerical evidence for Assumption \textup{\ref{ass:firstprime})} which is better accessible by numeric means  to a level which leaves no reasonable doubt. 
\end{remark}

To this end, we can state the main theorem of the present section.

\begin{theorem}\label{thm:solution_set_structure}
Let $e_1,e_2\in\R$ and suppose that Assumption \textup{\ref{ass:first}} holds. The analytic variety $\mathcal{S}_{e_1,e_2}$ can be parameterized by finitely many inextendible analytic curves. Each such curve $\gamma:I\to(0,\infty)\times(0,\infty)$
defined on an open $I\subset\R$ leaves any given compact subset of $(0,\infty)\times(0,\infty)$, i.e.\ for all compact $K\subset(0,\infty)\times(0,\infty)$ there exist $t_1,t_2\in I$ with 
\[
    \gamma(t)\in\big((0,\infty)\times(0,\infty)\big)\backslash K
\]
for all $t\in(\inf I,t_1)\cup(t_2,\sup I)$. Moreover, $\mathcal{S}_{e_1,e_2}$ is the disjoint union of tree-like subsets.
\end{theorem}
The proof is split into Subsections \ref{Section-Overall-solution-counting} - \ref{section-variety-reduction} and, finally, concluded in Subsection~\ref{section_Proof_of_structure_thm}.
\begin{remark}\label{rem:remark_after_structure_thm}
\begin{itemize}
    \item[\textup{(i)}]Note that if any inextendible curve leaves any compact subset of $(0,\infty)\times(0,\infty)$, then in particular, $\mathcal{S}_{e_1,e_2}$ can have no compact connected component. Moreover, all solutions can be found by studying the asymptotics of $F$ in the limits $x\to0$, $x\to\infty$, $h\to0$ and $h\to\infty$ and continuating the solution curves found there. The asymptotic analysis of $F$ is done Section~\textup{\ref{Asymptotic-Behavior-of-the-Solution-Set}}.
    \item[\textup{(ii)}] In the present section we skip the proof that the number of solution curves as in the theorem is at most finite. Note that a direct proof is quite difficult as one has to exclude an accumulation of curves. For example, the zero set of 
    \[
    (0,\infty)\times(0,\infty)\to\mathbb{R},~(x,h)\mapsto(h-1)\sin\left(\log(x)\right)
    \] 
    is an analytic variety which is still a tree-like set, but it consists of infinitely many curves that accumulate in both the limits $x\to0$ and $x\to\infty$. In turn, this can be proven from the asymptotic behavior of the function $F$.  Whenever we use Theorem~\textup{\ref{thm:solution_set_structure}} in Section~\textup{\ref{Asymptotic-Behavior-of-the-Solution-Set}}, this particular assertion is not needed.
\end{itemize}
\end{remark}

\begin{figure}
\centering
\includegraphics{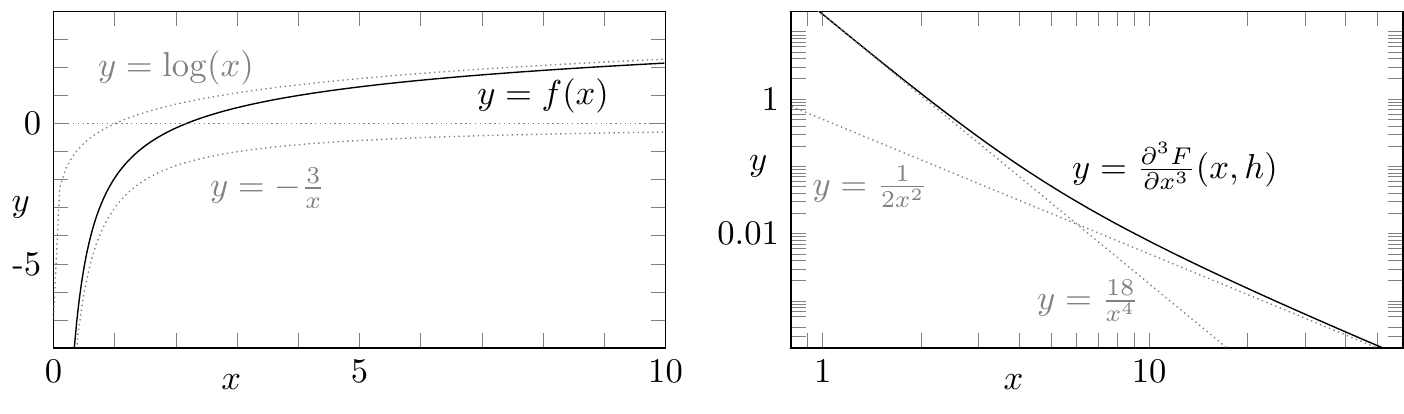}

\hspace{.5cm}(i) Plot of $f$ and its asymptotics\hspace{2cm} (ii) Plot of $\frac{\partial^3F}{\partial x^3}$ and its asymptotics

\begin{minipage}{.9\textwidth}
\captionsetup{format=plain, labelfont=bf}
\caption{Part (i) shows a plot of the function $f$ together with its asymptotics as asserted in the text. Part (ii) shows the third derivative of $F(\cdot,h)$ (for any fixed $h$) in a double logarithmic plot together with its asymptotics as derived in Lemma~\ref{lem:solution_count_lemma_h_fixed}.\label{plot-of-f}}
\end{minipage}
\end{figure}

\subsection{Counting of solutions}
\label{Section-Overall-solution-counting}

In this section we exploit that a strictly convex/concave function has at most two zeros. More generally, a smooth function whose $n$-th derivative has $\widetilde{n}$ zeros has itself at most $n+\widetilde{n}$ zeros, $n,\widetilde{n}\in\N_0$, which is an immediate consequence of the fundamental theorem of calculus.

Note that an analytic, convex/concave, but \emph{not} strictly convex/concave function is already linear on some open set, and thus everywhere. Hence, in the following we suppress the prefix \emph{strictly} as any relevant function studied for strict convexity/concavity is both obviously analytic and obviously not linear.

By these means, we obtain upper bounds on the number of solutions of \eqref{energy-equation-massive-case} for fixed $h$- and for fixed $x$-values, respectively. Moreover, we can identify regions in which solution curves must lie. In particular, we show that Theorem~\ref{thm:solution_set_structure} is not a theorem treating the empty set. Note that most assertions in the following lemmata can be read off from the different representation of $F_1$ in \eqref{energy-equation-massive-case} or follow by simple computations.

In order to establish an upper bound on the number of solutions for fixed $h>0$ we need the positivity of the function 
\begin{equation}\label{eq:third_derivative_of_F_in_x_direction}
    \frac{\partial^3F}{\partial x^3}(\cdot,h):~(0,\infty)\to\R,~~x\mapsto -\tfrac{3}{2}f''(x)-\big(\tfrac{x}{2}-1\big)f'''(x)
\end{equation}
(which is independent of $h$, $e_1$ and $e_1$). We have not found an analytical proof so far, but below we present numerical evidence in combination with an asymptotic analysis which leaves us no doubt about this positivity. So the general strategy (also discussed in Section \ref{Introduction}) is to assume the latter in order to be able to proof a lemma on a desired upper bound of solutions and carefully track which assertion in the following argumentation depends on this assumption/on this numerical evidence.
\begin{lemma}\label{lem:first_assumtion_asymptotics}
    Let $e_1,e_2\in\mathbb{R}$. There exist $x_1,x_2\in(0,\infty)$ such that 
    \[
        \frac{\partial^3F}{\partial x^3}(x,h)>0
    \]
    for all $x\in(0,x_1)\cup(x_2,\infty)$ and all $h>0$.
\end{lemma}
\begin{proof}
Using the asymptotics of $f$ and its series expansion in the limit $x\to\infty$ from Appendix~\ref{appendix-function-f} we get
\[
    \frac{\partial^3F}{\partial x^3}(x,h)\sim  \frac{18}{x^4}\quad\textup{as}~x\to 0\qquad\textup{as well as}\qquad \frac{\partial^3F}{\partial x^3}(x,h)\sim\frac{1}{2x^2} \quad\textup{as}~x\to \infty
\]
(independently of $h$, $e_1$ and $e_1$). These asymptotics imply the lemma.
\end{proof}

\begin{assumption}\label{ass:second}
Assume that the $(h$-, $e_1$- and $e_2$-independent$)$ function $\frac{\partial^3F}{\partial x^3}(\cdot,h)$ is positive on its entire domain $(0,\infty)$.
\end{assumption}
\begin{proof}[Numerical evidence for Assumption~\textup{\ref{ass:second}}]
We refer to Figure~\ref{plot-of-f}.(ii) which shows a double logarithmic plot of $\frac{\partial^3F}{\partial x^3}(\cdot,h)$ together with its proposed small-$x$- and large-$x$-asymptotics as in the proof of Lemma \ref{lem:first_assumtion_asymptotics}. The asymptotic behavior asserted in that lemma is well visible in Figure~\ref{plot-of-f}.(ii). 
\renewcommand\qedsymbol{\rotatebox{45}{$\square$}}
\end{proof}

\begin{lemma}\label{lem:solution_count_lemma_h_fixed}
    Let $e_1,e_2\in\mathbb{R}$ and suppose that Assumption \textup{\ref{ass:second}} holds. For any fixed $h>0$ the function $(0,\infty)\to\R,~x\mapsto F(x,h)$ has at most three zeros.
\end{lemma}
\begin{proof}
    The assumption on $\frac{\partial^3F}{\partial x^3}(\cdot,h)$ and the argument presented in the beginning of the present section immediately imply the lemma.
\end{proof}
\begin{remark}
    Note that we will not need Assumption \textup{\ref{ass:second}} anymore in the remainder of Section \textup{\ref{positive-mass-simplification}}, in particular, the main theorem of the present section $($Theorem \textup{\ref{thm:solution_set_structure}}$)$ 
    does not depend on it. However, some of the arguments in Section \textup{\ref{Asymptotic-Behavior-of-the-Solution-Set}} rely on the previous lemma again.
\end{remark}
If we fix the first argument of $F$ we obtain the following.
\begin{lemma}\label{lem:solution_count_lemma_x_fixed}
Let $e_1,e_2\in\R$.
For any fixed $x>0$ the function $(0,\infty)\to\R,~h\mapsto F(x,h)$ has at most three solutions and, depending on $e_1$ $($only$)$, this bound can be lowered according to Table~\textup{\ref{tabular-on-max-amount}}.
\end{lemma}

\begin{table}
    \centering
    \setlength{\unitlength}{25pt}
    \begin{tabular}{clllllll}
       ~\\
       ~& \begin{picture}(1,1)\put(0,0){\rotatebox{45}{$x<x_{\scriptscriptstyle(-)}$}}\end{picture}
       & \begin{picture}(1,1)\put(0,0){\rotatebox{45}{$x=x_{\scriptscriptstyle(-)}$}}\end{picture}
       & \begin{picture}(1,1)\put(0,0){\rotatebox{45}{$x_{\scriptscriptstyle(-)}<x<2$}}\end{picture}
       & \begin{picture}(1,1)\put(0,0){\rotatebox{45}{$x=2$}}\end{picture}
       & \begin{picture}(1,1)\put(0,0){\rotatebox{45}{$2<x<x_{\scriptscriptstyle(+)}$}}\end{picture}
       & \begin{picture}(1,1)\put(0,0){\rotatebox{45}{$x=x_{\scriptscriptstyle(+)}$}}\end{picture}
       & \begin{picture}(1,1)\put(0,0){\rotatebox{45}{$x>x_{\scriptscriptstyle(+)}$}}\end{picture}
       \\[7pt]
       if $e_1>0$:
       & $= 1$    & $= 1$    & $\leq 2$ & $\leq 2$ & $\leq 2$ & $\leq 2$ & $\leq 3$
       \\[3pt]if $e_1=0$:
       & $= 1$    & $= 1$    & $\leq 2$ & $\leq 1$ & $\leq 1$ & $= 1$    & $\leq 2$
       \\[3pt]if $e_1<0$:
       & $\leq 2$ & $\leq 2$ & $\leq 3$ & $= 1$    & $= 1$    & $= 1$    & $\leq 2$
    \end{tabular}
    \begin{minipage}{.9\textwidth}
    \captionsetup{format=plain, labelfont=bf}
    \caption{Collection of the respective counts of solutions. Note that the upper bounds are sharp in the sense that for each entry we can find $x$ and $e_1$ values that yield the respective amount of $h$-values to solve $F(x,h)=0$ in the respective regime. Example plots are shown in Section~\ref{sec:Numerics}. \label{tabular-on-max-amount}}
    \end{minipage}
\end{table}

\begin{proof}
Consider for any fixed $x>0$
\begin{equation}\label{eq:del_h_g_at_fixed_x}
        h^3\frac{\partial F}{\partial h}(x,h)=2\Big(\frac{x^2}{4}-x+\frac{29}{30}\Big)\,h^4+(2-x)h^2+\frac{e_1}{15}
\end{equation}
as a function of $h$. Read as a quadratic polynomial in $h^2$ it has at most two positive zeros and thus, as a polynomial in $h$, it has at most two positive zeros as well. Hence, the same holds for $\partial_hF(x,\cdot)$ and $F(x,\cdot)$ has at most three zeros by the argument in the beginning of this section.

Given $x<x_{(-)}$, the RHS of \eqref{eq:del_h_g_at_fixed_x}, as a polynomial in $h^2$, has precisely one positive zero if $e_1<0$, hence $F(x,\cdot)$ has at most two zeros in said case. If, on the other hand, $e_1\ge0$, the RHS of \eqref{eq:del_h_g_at_fixed_x} has no positive zero at all and we have at most one zero of $F(x,\cdot)$. Taking into account the limits
\[
    \lim_{h\to0}F(x,h)=-\infty\qquad\textup{and}\qquad\lim_{h\to\infty}F(x,h)=+\infty
\]
for $e_1\ge0$ we can, for such $e_1$, replace ``at most'' by ``exactly''.

If $x=x_{(\pm)}$ the leading coefficient in \eqref{eq:del_h_g_at_fixed_x} vanishes. Hence, $\partial_hF(x_{(-)},\cdot)$ has precisely one zero if $e_1<0$ and no zero at all if $e_1\ge0$, implying that $F(x_{(-)},\cdot)$ has at most two or at most one zero, respectively. The very same argument can be applied to $F(x_{(+)},\cdot)$ by reversing the sign of $e_1$. Taking also into account the limits
\begin{align*}
    \lim_{h\to0}F\big(x_{(-)},h)=-\infty&\qquad\textup{and}\qquad\lim_{h\to\infty}F\big(x_{(-)},h)    =+\infty &\textup{if }e_1\ge0
    \intertext{as well as}
    \lim_{h\to0}F\big(x_{(+)},h)=+\infty&\qquad\textup{and}\qquad\lim_{h\to\infty}F\big(x_{(+)},h)=-\infty&\textup{if }e_1\le0
\end{align*}
we see that in the latter case ``at most'' can be replaced by ``exactly''.

If $x_{(-)}<x<2$ the RHS of \eqref{eq:del_h_g_at_fixed_x} has precisely one zero if $e_1\ge0$, implying that $F(x,\cdot)$ has at most two zeros.

At $x=2$ the RHS of \eqref{eq:del_h_g_at_fixed_x} has no positive zero if $e_1\le0$ and precisely one if $e_1>0$. Consequently, $F(2,\cdot)$ has at most one or at most two zeros, respectively. Taking, moreover, into account that
\[
    \lim_{h\to0}F(2,h)=+\infty\qquad\textup{and}\qquad\lim_{h\to\infty}F(2,h)=-\infty
\]
for $e_1<0$ in this case ``at most'' can again be replaced by ``exactly''.

For $2<x<x_{(+)}$ the RHS of \eqref{eq:del_h_g_at_fixed_x} has no positive zero if $e_1\le0$ and precisely one positive zero if $e_1>0$. Consequently, $F(x,\cdot)$ has at most one or at most two zeros, respectively. Taking into account the limits
\[
    \lim_{h\to0}F(x,h)=+\infty\qquad\textup{and}\qquad\lim_{h\to\infty}F(x,h)=-\infty
\]
for $e_1<0$ also here ``at most'' can be replaced by ``exactly''.

Finally, if $x>x_{(+)}$ the RHS of \eqref{eq:del_h_g_at_fixed_x} has precisely one positive zero if $e_1\le0$, allowing at most two positive zeros for $F(x,\cdot)$.

This finishes the proof of every entry shown in Table~\ref{tabular-on-max-amount}.
\end{proof}
Note that we will continue to study the zeros of $\partial_hF(x,\cdot)$ at fixed $x$ in the subsequent section.

\begin{corollary}\label{cor:non-empty_solution_set}
The solution set $\mathcal{S}_{e_1,e_2}$ of $F(x,h)=0$ is non-empty for any choice of parameters $e_1$ and $e_2$.
\end{corollary}
\begin{proof}
For each row of Table~\ref{tabular-on-max-amount} we have at least one entry with exactly one solution at the respective $x$-value.
\end{proof}

After we have demonstrated how many solutions exist at most if we fix one variable we can give a more accurate location of the three solutions at a fixed $h$. Note that the solutions of $F(2,h)=0$ are given by the zeros of the polynomial
\begin{equation}\label{eq:g_at_x_equals_two}
    p(h)=-30\,h^2F(2,h)=h^4+e_2h^2+e_1\,.
\end{equation}

\begin{lemma}\label{lem:solution_count_lemma_two}
Let $e_1,e_2\in\R$ and suppose that the polynomial function $p(h)=h^4+e_2h^2+e_1$ has a zero $h_0$.
If $h_0>\exp(\gamma_\textup{E}-1)$ $($approximately $\approx0.6552)$  with the Euler-Mascheroni number $\gamma_\textup{E}$, then $F(x,h_0)$ has at least three solutions $x_1,x_2,x_3\in(0,\infty)$ with $x_1\in(0,2)$, $x_2=2$ and $x_3\in(2,\infty)$.
\end{lemma}
\begin{proof}
The map $(0,\infty)\to\R,~x\mapsto F(x,h_0)$ fulfills
\[
\lim_{x\to0}F(x,h_0)=-\infty\,,\qquad\lim_{x\to\infty}F(x,h_0)=+\infty\qquad\textup{and}\qquad F(2,h_0)=0\,.
\]
Moreover, its derivative in $x=2$ fulfills
\[
\frac{\partial}{\partial x}\,F(2,h_0)=-\frac{1}{2}-\log(h_0)-\frac{1}{2}f(2)<0\,,
\]
where we used that $f(2)=1-2\gamma_\mathrm{E}$. Consequently, $F(\cdot,h_0)$ is positive on an interval of the form $(2-\varepsilon,2)$ and negative on an interval of the form $(2,2+\varepsilon)$. In combination with the limits above, this implies the existence of zeros of $F(\cdot,h_0)$ as asserted in the lemma.
\end{proof}
\begin{remark}
\begin{itemize}
    \item[\textup{(i)}] Note that, although being concerned with the case of a fixed $h$-value, the previous lemma does not depend on Assumption \textup{\ref{ass:second}}. However, in combination with Lemma~\textup{\ref{lem:solution_count_lemma_h_fixed}.(i)} $($and Assumption \textup{\ref{ass:second}}$)$ the assertion ``at least three'' in Lemma~\textup{\ref{lem:solution_count_lemma_two}} can be replaced by ``exactly three''.
    \item[\textup{(ii)}] If in Lemma~\textup{\ref{lem:solution_count_lemma_two}} we claim $h_0<\exp(\gamma_\mathrm{E}-1)$ instead, we still have the solution at $(2,h_0)$, but if two more $h=h_0$-solutions exist at all, they must be either both larger than \textup{2} or both smaller than \textup{2}.
\end{itemize}
\end{remark}
Complementary to Lemma~\ref{lem:solution_count_lemma_two} concerning zeros of $p$ we find the following lemma concerning $h$-values where $p(h)\neq0$.
\begin{lemma}\label{lem:solution_count_lemma_three}
Let $h>0$, $e_1,e_2\in\R$ and $p(h)=h^4+e_2h^2+e_1$ $($cf.\ \eqref{eq:g_at_x_equals_two}$)$.
\begin{itemize}
    \item[\textup{(i)}] If $p(h)>0$, the equation $F(x,h)=0$ has at least one solution $x\in(2,\infty)$.
    \item[\textup{(ii)}] If $p(h)<0$, the equation $F(x,h)=0$ has at least one solution $x\in(0,2)$.
\end{itemize}
\end{lemma}
\begin{proof}
At $x\neq2$ the equation $F(x,h)=0$ is equivalent to the function
\[
    (0,\infty)\setminus\{2\}\to\R,~x\mapsto \frac{F(x,h)}{\frac{x}{2}-1}
\]
having a zero. For this function we observe
\[
\lim_{x\to0}\frac{F(x,h)}{\frac{x}{2}-1}=\lim_{x\to\infty}\frac{F(x,h)}{\frac{x}{2}-1}=+\infty\,.
\]
Moreover, it has a first order pole at $x=2$, and by \eqref{eq:g_at_x_equals_two} this pole's residue has the opposite sign than the value $p(h)$. Consequently,
\[
    \lim_{\begin{smallmatrix}x\to2\\x>2\end{smallmatrix}}\frac{F(x,h)}{\frac{x}{2}-1}=-\textup{sgn}\,p(h)\cdot\infty\qquad\textup{and}\qquad\lim_{\begin{smallmatrix}x\to2\\x<2\end{smallmatrix}}\frac{F(x,h)}{\frac{x}{2}-1}=\textup{sgn}\,p(h)\cdot\infty
\]
and the lemma follows.
\end{proof}
\begin{remark}
Lemma~\textup{\ref{lem:solution_count_lemma_three}} particularly shows that the equation $F(x,h)$ for any fixed $h>0$ has at least one solution for $x$ $($recall that $F(2,h)=0$ for $p(h)=0)$.
\end{remark}

We postpone a further location of solution curves via the asymptotics of $F$ to Section~\ref{Asymptotic-Behavior-of-the-Solution-Set}.

\subsection{Possible non-{\normalfont\bfseries\itshape h}-solvable points}
\label{section-Possible-non-h-solvable-points}
The implicit function theorem (in its analytical version) tells us that, whenever we have a solution $F(x,h)=0$ such that $\nabla F(x,h)=\big(\partial_x F(x,h),\partial_hF(x,h)\big)\neq0$, then there exists an open neighborhood of $(x,h)$ in which all solutions of $F(x,h)=0$ are collected in an analytic curve. Moreover, any such ``piece of solution curve'' can be continuated either until it leaves the domain $(0,\infty)\times(0,\infty)$ of $F$ or until it runs into a point where $\nabla F(x,h)=0$.

The present section is dedicated to studying (a necessary condition on) points $(x,h)$ in which the gradient of $F$ vanishes. Since $\partial_xF$ involves derivatives of $f$ in a poorly manageable combination, we use $\partial_hF=0$ as a necessary condition. The latter in turn is (equivalent to) a polynomial equation and is explicitly solvable. Note that by this weaker criterion we also identify points in which the solution curves ``turn around'', i.e.\ points where they are not solvable for $h$, but possibly for $x$.

\begin{figure}
\flushright
    \begin{minipage}{.55\textwidth}
        \includegraphics{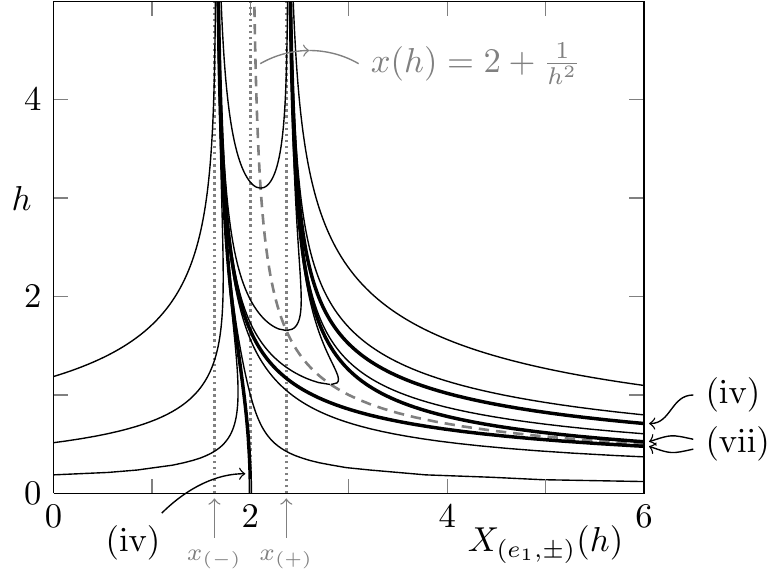}
    \end{minipage}
    \begin{minipage}{.44\textwidth}
    \captionsetup{format=plain, labelfont=bf}
        \caption{The plot shows the curves of points in the $(x,h)$-plane in which the zero set of $F$ is possibly not solvable for $h$ for different values of $e_1$. Everywhere else the zero set of $F$ is representable as the graph of a function $h(x)$. We picked the following $e_1$-values:}

        \begin{center}
            \begin{tabular}{llll}
                 (i) & $e_1=-100$\!\! & (vi)$^*$ & $e_1=6$   \\
                 (ii) & $e_1=-10$ & \textbf{(vii)} & {\normalfont\bfseries\itshape e}$\mathbf{{}_1=\frac{15}{2}}$   \\
                 (iii)\! & $e_1=-1$ & (viii)$^*$\!\! & $e_1=9$   \\
                 \textbf{(iv)} & {\normalfont\bfseries\itshape e}$\mathbf{{}_1=0}$ & (ix) & $e_1=15$    \\
                 (v) & $e_1=1$ & (x) & $e_1=100$\!
            \end{tabular}
        \end{center}
        \label{possible-bad-points-figure}
        The dashed curve marks the symmetry
    \end{minipage}

    \vspace{1pt}
    \begin{minipage}{.95\textwidth}
         line of conformal coupling, $x=2+\frac{1}{h^2}$. Any curve fragment above/right to this symmetry line corresponds to $X_{(e_1,+)}$, any fragment below/left to this line corresponds to $X_{(e_1,-)}$. The thick lines mark the distinguished values of $e_1$. The remaining curves are assigned to the remaining $e_1$-values in a monotonous fashion. For the cases marked with $^*$ we plotted only $X_{(e_1,-)}$ to avoid an overload.
    \end{minipage}

\end{figure}

\begin{lemma}\label{lem:def_of_Xpm}
Let $e_1,e_2\in\R$ and denote
\[
    h_\textup{min}=\begin{cases}
        0&\textup{if }e_1\le\frac{15}{2}\\
        (e_1-\frac{15}{2})^{\nicefrac{1}{4}}&\textup{if }e_1>\frac{15}{2}
    \end{cases}~.
\]
\begin{itemize}
    \item[\textup{(i)}] The mapping
    \[
        h\mapsto X_{(e_1,\pm)}(h)=2+\tfrac{1}{h^2}\pm\sqrt{\tfrac{2}{15}+\tfrac{1}{h^4}\big(1-\tfrac{2e_1}{15}\big)}
    \]
    defines real-valued functions $X_{(e_1,\pm)}:(0,\infty)\cap[h_\textup{min},\infty)\to\R$.
    \item[\textup{(ii)}] Any $(x,h)\in(0,\infty)\times(0,\infty)$ with $\nabla F(x,h)=0$ fulfills $x\in\{X_{(e_1,+)}(h),X_{(e_1,-)}(h)\}$.
\end{itemize}
\end{lemma}
\begin{proof}
Finding the zeros of $\partial_hF$ is equivalent to solving the polynomial equation (for $x$)
\begin{equation}
    0=x^2-\Big(4+\frac{2}{h^2}\Big)x+\frac{58}{15}+\frac{4}{h^2}+\frac{2e_1}{15h^4} \label{equation-for-possible-bad-points}
\end{equation}
and doing so results in the functions $X_{(e_1,\pm)}$ in the lemma. The domain which yields real functions is obtained by requiring the radicand to be non-negative. From this observation both assertions of the lemma follow immediately.
\end{proof}

A visualization of the graphs of $X_{(e_1,\pm)}$ is given in Figure~\ref{possible-bad-points-figure} for a few values of $e_1$.  Note how $X_{(e_1,-)}$ becomes negative if $e_1<0$, that is, the curve defined by $h\mapsto\big(X_{(e_1,-)}(h),h\big)$ leaves the domain $(0,\infty)\times(0,\infty)\ni(x,h)$ of our model at small $h$.

We collect a few properties of the functions $X_{(e_1,\pm)}$.
\begin{lemma}\label{lem:non-h-solvable_point-properties}
Denote for $e_1\in\R$
\[
    M_{e_1}=(0,\infty)\backslash\Big(\textup{ran}\,X_{(e_1,+)}\cup\textup{ran}\,X_{(e_1,-)}\Big)\subset(0,\infty)\,.
\]
\begin{itemize}
    \item[\textup{(i)}] The functions $X_{(e_1,\pm)}$ admit the asymptotic expansion
    \[
        X_{(e_1,\pm)}(h)=x_{(\pm)}+\frac{1}{h^2}+\mathcal{O}\Big(\frac{1}{h^4}\Big)
    \]
    in the limit $h\to\infty$.
    \item[\textup{(ii)}] If $e_1\le\frac{15}{2}$ the functions $X_{(e_1,\pm)}$ admit the asymptotic expansion
    \[
        X_{(e_1,\pm)}(h)=\frac{1\pm\sqrt{1-\frac{2e_1}{15}}\,}{h^2}+2\,\pm\,\frac{h^2}{15\sqrt{1-\frac{2e_1}{15}}\,}+\mathcal{O}\big(h^2\big)
    \]
    in the limit $h\to0$. In particular, $X_{(0,-)}(h)\to2$ as $h\to0$ and $X_{(0,-)}(h)$ is a bounded, monotonous function.
    \item[\textup{(iii)}] If $e_1<0$ the function $X_{(e_1,-)}$ attains its global maximum
    \[
    \max_{h>0}\big(X_{(e_1,-)}(h)\big)\,=\,X_{(e_1,-)}\Big(\big(\tfrac{2e_1^2}{15}-e_1\big)^{\nicefrac{1}{4}}\Big)\,=\,2-\big(\tfrac{15}{2}-\tfrac{225}{4e_1}\big)^{-\nicefrac{1}{2}}\quad\in(x_{(-)},2)
    \]
    and is unbounded from below. $X_{(e_1,+)}$, in turn, is strictly decreasing and bijective as a function $(0,\infty)\to(x_{(+)},\infty)$. Consequently,
    $
        M_{e_1}=\big(\,2-\big(\tfrac{15}{2}-\tfrac{225}{4e_1}\big)^{-\nicefrac{1}{2}}\,,\,x_{(+)}\,\big].
    $
    \item[\textup{(iv)}] If $e_1=0$ the mappings $X_{(0,\pm)}$ define strictly decreasing, bijective functions
    \[
    X_{(0,-)}:(0,\infty)\to(x_{(-)},2) \qquad\textup{and}\qquad X_{(0,+)}:(0,\infty)\to(x_{(+)},\infty)\,.
    \]
    Consequently, $M_0=(0,x_{(-)}]\cup[2,x_{(+)}]$.
    \item[\textup{(v)}] If $0<e_1\le\frac{15}{2}$ the mappings $X_{(0,\pm)}$ define strictly decreasing, bijective functions
    \[
    X_{(0,\pm)}:(0,\infty)\to(x_{(\pm)},\infty)\,.
    \]
    Consequently, $M_{e_1}=(0,x_{(-)}]$.
    \item[\textup{(vi)}] If $e_1>\frac{15}{2}$ both functions $X_{(e_1,\pm)}$ are bounded on their domains $[h_\textup{min},\infty)$. $X_{(e_1,-)}$ is bounded from below by its infimum $x_{(-)}$ and pointwise bounded from above by $X_{(e_1,+)}$. $X_{(e_1,+)}$, in turn, attains its global maximum
    \[
    \max_{h>0}\big(X_{(e_1,+)}(h)\big)\,=\,X_{(e_1,+)}\Big(\big(\tfrac{2e_1^2}{15}-e_1\big)^{\nicefrac{1}{4}}\Big)\,=\,2+\big(\tfrac{15}{2}-\tfrac{225}{4e_1}\big)^{-\nicefrac{1}{2}}\quad\in(x_{(+)},\infty)\,.
    \]
    Consequently, $M_{e_1}=(0,x_{(-)}]\cup\big(2+\big(\tfrac{15}{2}-\tfrac{225}{4e_1}\big)^{-\nicefrac{1}{2}},\infty\big)$.
\end{itemize}
\end{lemma}

We skip the proof since any of the assertions can be obtained by straightforward computations.
\begin{remark}
Note that the implicit function theorem provides us around any $x_0\in M_{e_1}$ with $F(x_0,h)=0$ an analytic curve of the form $x\mapsto\big(x,h(x)\big)$ defined on a neighborhood $U\ni x_0$. Up to the possibility that $h(x)\to0$ or $h(x)\to\infty$ if $x$ approaches the boundary of $U$, such a solution curve can even be extended to the whole connected component of $M_{e_1}$ containing $x_0$. We will continue to study these possibilities in Section~\textup{\ref{Asymptotic-Behavior-of-the-Solution-Set}} using the asymptotics of $F$.
\end{remark}

\subsection{Non-existence of local extrema}
\label{Section-Non-existence-of-local-extrema}

In the previous section we have located where the solution set of $F(x,h)=0$ is potentially not locally solvable for $h$. In the present section we show that $F$ has no local extrema. More precisely, we show that at any critical point at which $\nabla F(h,x)=\big(\partial_x F(x,h),\partial_hF(x,h)\big)=0$ the function $F$ has a saddle. This is equivalent to the fact that $\det(\textup{Hess}\,F)<0$ in all critical points, where $\textup{Hess}\,F$ denotes the Hessian matrix of $F$. To show this, we study the analytic functions
\begin{align}
    Y_{(e_1,\pm)}:&~(X_{(e_1,\pm)}\big)^{-1}(\mathbb{R}_{>0})\to\R,\notag\\
    h\mapsto~&\frac{\partial^2F}{\partial h^2}\big(X_{(e_1,\pm)}(h),h\big)\cdot\frac{\partial^2F}{\partial x^2}\big(X_{(e_1,\pm)}(h),h\big)-\Big(\frac{\partial^2F}{\partial h\partial x}\big(X_{(e_1,\pm)}(h),h\big)\Big)^2\label{charpol-offset-functon}\\
    =&~\Big[\frac{1}{2}\Big(X_{(e_1,\pm)}(h)-2+\frac{1}{h^2}\Big)^2-\frac{1}{2h^4}-\frac{e_1}{5h^4}-\frac{1}{15}\Big]\notag\\
    &\hspace{2.2cm}\cdot\Big[\frac{h^2}{2}-f'\circ X_{(e_1,\pm)}(h)-\Big(\frac{1}{2}X_{(e_1,\pm)}(h)-1\Big)\cdot f''\circ X_{(e_1,\pm)}(h)\Big]\notag\\
    &\hspace{6cm}-\Big[\Big(X_{(e_1,\pm)}(h)-2\Big)\,h-\frac{1}{h}\Big]^2\notag
\end{align}
on the (possibly $e_1$-dependent) maximal domains of $X_{(e_1,\pm)}$ to yield positive values (specified in Lemma~\ref{lem:def_of_Xpm} and to be refined in Lemma \ref{lem:analytical_Hessian_lemma} below).
We study the functions $Y_{(e_1,\pm)}$ in terms of their asymptotics and by numerical means to show that they are mostly negative, and if not, then $\partial_xF(h,X_{(e_1,\pm)}(h))\neq 0$ and thus the point in question is not critical.

\begin{lemma}\label{lem:analytical_Hessian_lemma}
Let $e_1\in\R$.
\begin{itemize}
    \item[\textup{(i)}] $Y_{(e_1,\pm)}(h)<0$ for sufficiently large $h$.
    \item[\textup{(ii)}] If $e_1<\frac{15}{2}$, then $Y_{(e_1,+)}(h)<0$ for sufficiently small $h>0$.
    \item[\textup{(iii)}]If $e_1\in[0,\frac{15}{2}]$, then $Y_{(e_1,-)}(h)<0$ for sufficiently small $h>0$.
    \item[\textup{(iv)}]For $e_1<0$ denote $h_\textup{crit}=\big(\big((\tfrac{15}{29})^2-\tfrac{e_1}{29}~\big)^{\nicefrac{1}{2}}-\tfrac{15}{29}\big)^{\nicefrac{1}{2}}$. Then $X_{(e_1,-)}$ is positive on $(h_\textup{crit},\infty)$ and there exists $\varepsilon>0$ such that $Y_{(e_1,-)}$ is negative on $(h_\textup{crit},h_\textup{crit}+\varepsilon)$.
    \end{itemize}
\end{lemma}
\begin{proof}
By Lemma~\ref{lem:non-h-solvable_point-properties}.(i) we have $X_{(e_1,\pm)}(h)\to x_{(\pm)}$ as $h\to\infty$ and since $f$ is smooth (i.e.\ $f'$ and $f''$ are continuous) we can read off from \eqref{charpol-offset-functon} that $Y_{(e_1,\pm)}\to-\infty$ as $h\to\infty$. More precisely, identifying the dominant terms we find
\begin{equation}\label{eq:power_law_expansion_one}
    \frac{1}{h^2}Y_{(e_1,\pm)}(h)\to-\frac{2}{15}\qquad\textup{as }h\to\infty\,.
\end{equation}
This proves (i).

In order to show (ii), recall from Section~\ref{section-Possible-non-h-solvable-points} that both $X_{(e_1,+)}$ for $e_1<\frac{15}{2}$ and $X_{(e_1,-)}$ for $e_1\in(0,\frac{15}{2})$ are defined on $(0,\infty)$ and that $X_{(e_1,\pm)}(h)\to+\infty$ as $h\to0$ for the given respective $e_1$-values. Expanding each occurrence of $X_{(e_1,\pm)}$ and $f$ in $Y_{(e_1,\pm)}$ from \eqref{charpol-offset-functon} to a sufficiently high 
order in $h$ (cf.\ Lemma~\ref{lem:non-h-solvable_point-properties}.(ii) and Appendix~\ref{appendix-function-f}) we find that
\begin{equation}\label{eq:power_law_expansion_two}
    Y_{(e_1,\pm)}(h)\,=-\,\frac{2(1-\frac{2e_1}{15})}{15\Big(1\pm\sqrt{1-\frac{2e_1}{15}}~\Big)^2}~h^2+\mathcal{O}(h^{4})
\end{equation}
as $h\to0$. Note that the functions $z\mapsto\frac{z}{(1\pm\sqrt{z\,})^2}$ are positive for the relevant domains. This proves (ii) and, moreover, (iii) for $e_1\in(0,\frac{15}{2})$.

Recall that $X_{(\frac{15}{2},\pm)}(h)=x_{(\pm)}+\frac{1}{h^2}$ for $h>0$. By an expansion to sufficiently high order we find that
\begin{equation}\label{eq:Ypm_small_h_asymptotics}
    Y_{(\frac{15}{2},\pm)}(h)=\Big(-\frac{8}{225}\pm\frac{16\sqrt{30}}{1575}\Big)\,h^6+\mathcal{O}(h^7)\,.
\end{equation}
In particular, the leading order term in $h$ of $Y_{(\frac{15}{2},-)}(h)$ is negative, proving (iii) for $e_1=\frac{15}{2}$.

In order to complete (iii) note that $X_{(0,-)}$ is defined on all of $(0,\infty)$, but now approaches the limit $X_{(0,-)}(h)\to2$ as $h\to0$. Hence we obtain
\[
    h^2\,Y_{(0,-)}(h)\to-1
\]
in said limit, showing (iii) for $e_1=0$.

Finally, if $e_1<0$, the function $X_{(e_1,-)}$ is positive only if we restrict it to $(h_\textup{crit},\infty)$. In particular, we have $X_{(e_1,-)}(h)\to0$ as $h\to h_\textup{crit}$. Expanding the respective occurrences of $f$ and its derivatives to sufficient high order in $x$ we find
\[
    \lim_{h\to h_\textup{crit}}~\frac{\partial^2F}{\partial h^2}\big(X_{(e_1,-)}(h),h\big)~=~\frac{29}{15}-\frac{2}{h_\textup{crit}^2}-\frac{e_1}{5h_\textup{crit}^4}>0\,.
\]
Note that inserting $h_\textup{crit}$ as defined in the lemma the positivity of the latter expression is to be seen in a straightforward computation. Moreover, we find that
\[
    \frac{\partial^2F}{\partial x^2}\big(X_{(e_1,-)}(h),h\big)=\frac{h^2}{2}-f'\circ X_{(e_1,-)}(h)-\Big(\frac{1}{2}X_{(e_1,-)}(h)-1\Big)f''\circ X_{(e_1,-)}(h)\quad\to-\infty\,,
\]
as $h\to h_\textup{crit}$, where for the limit we note that $f'(x)\to+\infty$ and $f''(x)\to-\infty$ as $x\to0$. At last,
\[
\lim_{h\to h_\textup{crit}}~\frac{\partial^2F}{\partial h\partial x}\big(X_{(e_1,-)}(h),h\big)=-2h_\textup{crit}-\frac{1}{h_\textup{crit}}\,,
\]
in particular, this limit exists. Together these three limits imply $Y_{(e_1,-)}(h)\to-\infty$ as $h\to h_\textup{crit}$, proving (iv).
\end{proof}

\begin{figure}
\centering
    \begin{minipage}{.495\textwidth}
        \centering
            \includegraphics{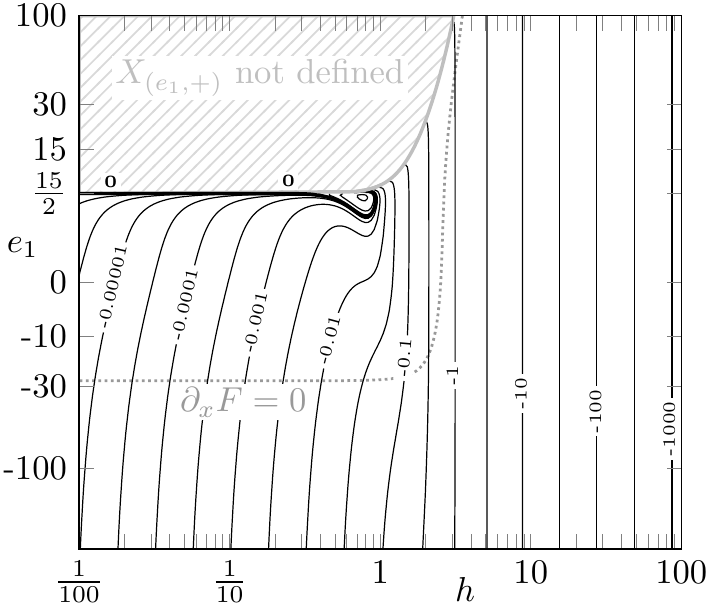}
            (i)~~$(h,e_1)\mapsto Y_{(e_1,+)}(h)$
    \end{minipage}
    \hfill
    \begin{minipage}{.495\textwidth}
        \centering
            \includegraphics{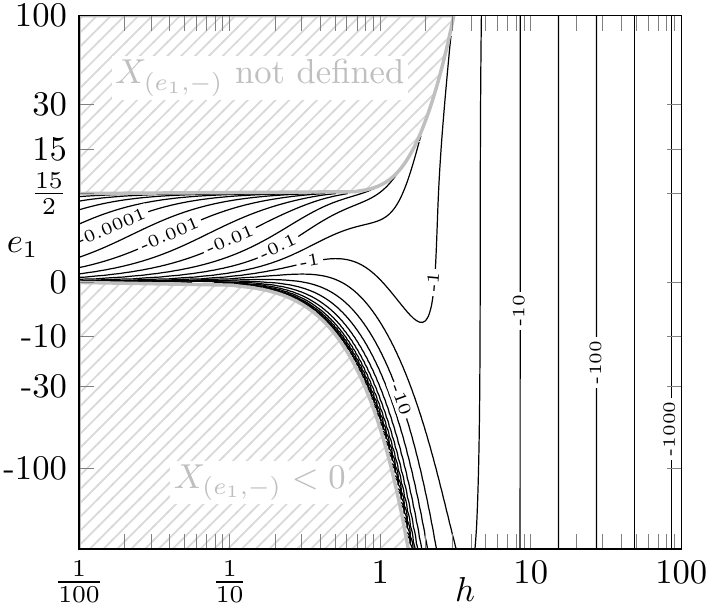}
            (ii)~~$(h,e_1)\mapsto Y_{(e_1,-)}(h)$
    \end{minipage}

    \begin{minipage}{.495\textwidth}
        \centering
            \includegraphics{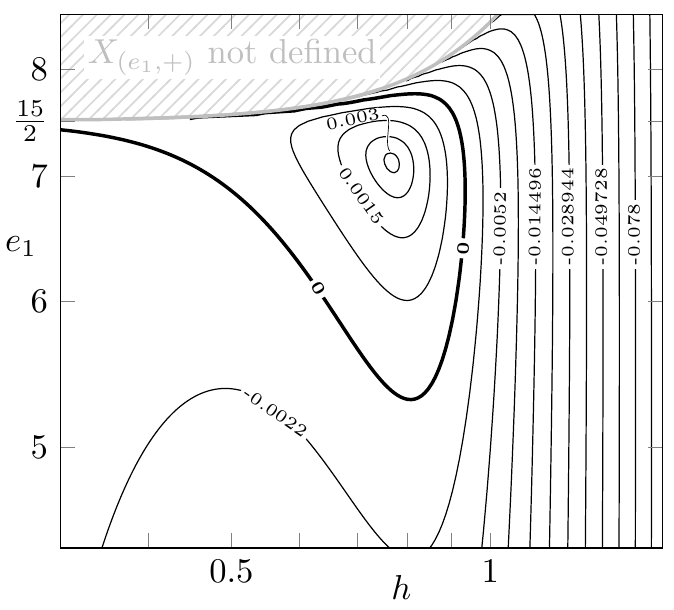}
            (iii)~~$(h,e_1)\mapsto Y_{(e_1,+)}(h)$
    \end{minipage}
    \hfill
    \begin{minipage}{.495\textwidth}
        \centering
            \includegraphics{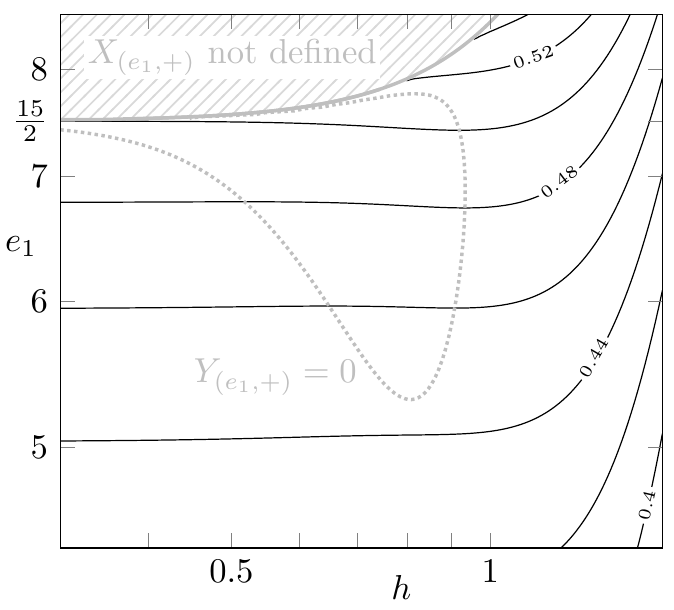}
            (iv)~~$(h,e_1)\mapsto \partial_xF\big(X_{(e_1,+)}(h),h\big)$
    \end{minipage}
    \begin{minipage}{.9\textwidth}
    \captionsetup{format=plain, labelfont=bf}
    \caption{The graphics (i)-(iii) show the level sets and particularly the zero sets (if non-empty as thick lines) of the functions $Y_{(e_1,\pm)}$ in dependence of both $e_1$ and $h$, where (iii) is a zoom into (i) at specific values. (iv) shows the level sets of $\partial_xF$ along the graph of $X_{(e_1,+)}$, again in dependence of $e_1$ and $h$ in the region where $Y_{(e_1,+)}$ takes non-negative values. While we chose a logarithmic scaling for the horizontal axes, the vertical axes are rescaled by a 3rd-order polynomial which is approximately linear around the distinguished values $\{0,\frac{15}{2}\}\ni e_1$, but strongly compresses on the ends of large absolute values. The dotted gray lines in (i) and (iv) mark the zero set of the respective other graphic for orientation.\label{hessianplots}}
    \end{minipage}
\end{figure}

\begin{remark}
\begin{itemize}
    \item[\textup{(i)}] Note that $Y_{(\frac{15}{2},+)}$ is not considered for small $h$ in the lemma. Indeed, due to $-\frac{8}{225}+\frac{16\sqrt{30}}{1575}>0$, cf.\ \eqref{eq:Ypm_small_h_asymptotics}, a claim for $Y_{(\frac{15}{2},+)}$ similar to parts \textup{(ii)} or \textup{(iii)} of the lemma is false. This can also, to some extend, be observed in Figures \textup{\ref{hessianplots}.(i)} and \textup{(iii)}.
    \item[\textup{(ii)}] If $e_1>\frac{15}{2}$, both $X_{(e_1,\pm)}$ are defined on $\big[ (e_1-\tfrac{15}{2})^{\nicefrac{1}{4}},\infty\big)$ and bounded, hence $Y_{(e_1,\pm)}$ possesses a limit as $h\to(e_1-\tfrac{15}{2})^{\nicefrac{1}{4}}$ and there is no need to determine an asymptotic behavior.
\end{itemize}
\end{remark}

We make an assumption on the analog statement on the functions $Y_{(e_1,\pm)}$ away from their asymptotics. We have not found an analytic proof for this assertion, however, we present numerical evidence which leaves us no doubt about the assumption to be true, although it is not proven analytically.
\begin{assumptionaux}\label{ass:firstprime}
    Let $e_1\in\R$. Assume that:
    \begin{itemize}
        \item[\textup{(i)}] For all $h\in(h_\textup{min},\infty)$ $($i.e.\ all $h$-values for which $X_{(e_1,+)}$ is defined$)$ we have either $Y_{(e_1,+)}(h)<0$ or $\partial_xF\big(X_{(e_1,+)}(h),h\big)\neq0$.
        \item[\textup{(ii)}] For all $h\in(h_\textup{min},\infty)$ such that $h>h_\textup{crit}$ $($i.e.\ all $h$-values for which $X_{(e_1,-)}$ is defined and yields a positive value$)$ we have $Y_{(e_1,-)}(h)<0$.
    \end{itemize}
    
\end{assumptionaux}
\begin{proof}[Numerical evidence for Assumption~\textup{\ref{ass:firstprime}}]
In Figures~\ref{hessianplots}.(i)-(iii) we show plots of the functions $(h,e_1)\mapsto Y_{(e_1,\pm)}(h)$ in terms of their level sets, where (iii) is a zoom into (i). Note that we can observe all analytical assertions on the asymptotics from Lemma~\ref{lem:analytical_Hessian_lemma}.
In particular, if a straight line of constant $e_1$ is intersected by the level sets of $Y_{(e_1,\pm)}$ in (approximately) equidistant points, this corresponds to the asymptotic power-law expansions that were proven in Lemma \ref{lem:analytical_Hessian_lemma} (e.g.\ \eqref{eq:power_law_expansion_one},\eqref{eq:power_law_expansion_two} or \eqref{eq:Ypm_small_h_asymptotics}).
Figure~\ref{hessianplots}.(iv) shows the function $(h,e_1)\mapsto \partial_xF\big(X_{(e_1,+)}(h),h\big)$ in terms of its level sets.

Observing that all the asymptotic assertions of Lemma~\ref{lem:analytical_Hessian_lemma} are already visible in Figure~\ref{hessianplots}, we have no doubt that $Y_{(e_1,-)}$ is negative wherever it is defined. Moreover, we observe that the zeros of $Y_{(e_1,+)}$ and $\partial_xF(X_{(e_1,+)}(\cdot),\cdot)$ are widely separated by a considerable margin in the $(e_1,h)$-plane and that $\partial_xF\big(X_{(e_1,+)}(h),h\big)\neq0$ whenever $Y_{(e_1,+)}$ is non-negative. Thus we conclude the assertion of the assumption from numerical evidence.

Note that it is clear from our considerations of the present and the preceding section that Assumptions \ref{ass:firstprime} and \ref{ass:first} are equivalent. We will continue to refer to Assumption \ref{ass:first} in the following.
\renewcommand\qedsymbol{\rotatebox{45}{$\square$}}
\end{proof}
\begin{lemma}\label{lem:indef_Hess_in_crit_points_new}
Let $e_1,e_2\in\mathbb{R}$ and suppose that Assumption \textup{\ref{ass:first}} holds. Then the function $F$ cannot have a local extremum.
\end{lemma}
\begin{proof}
    A necessary condition for a local extremum in $(x,h)\in(0,\infty)\times(0,\infty)$ is that $\nabla F(x,h)=0$ and, moreover, that the Hessian in such a point is at least semi-definite (i.e. definite or singular). By the assumption, in any point with $\nabla F(x,h)\neq0$ the Hessian's determinant is negative implying its indefiniteness.
\end{proof}

As a consequence we obtain the following:

\begin{proposition}\label{prop:bounded_connected_components}
Let $e_1,e_2\in\mathbb{R}$ and suppose that Assumption \textup{\ref{ass:first}} holds. The open set
\[
    M=\big\{\,(x,h)\in(0,\infty)\times(0,\infty)\,\big|\,F(x,h)\neq0\,\big\}\,=\,\big((0,\infty)\times(0,\infty)\big)\backslash\mathcal{S}_{e_1,e_2}
\]
of non-solutions to the massive consistency equation $F(x,h)=0$ possesses no $($non-empty$)$ connected component whose closure $($w.r.t.\ $\R^2)$ is a compact subset of $(0,\infty)\times(0,\infty)$.
\end{proposition}
\begin{proof}
Suppose that the contrary assertion holds, that is, let $N$ be a connected component of $M$ such that $\overline{N}\subset(0,\infty)\times(0,\infty)$ is compact. As a connected component of the open set $M$, $N$ is itself open, hence $\overline{N}=N\cup\partial N$ is a disjoint union. On the other hand, $\partial N\subset(0,\infty)\times(0,\infty)$, that is, $F$ is defined on $\partial N$ with $F\big|_{\partial N}=0$. Finally, $F$ attains both maximum and minimum on the compact set $\overline{N}$ and assuming that $N$ is non-empty we have either $\max F\big|_{\overline{N}}>0$ or $\min F\big|_{\overline{N}}<0$. In both cases, $F$ necessarily has a local extremum in $N$ which contradicts Lemma~\ref{lem:indef_Hess_in_crit_points_new}.
\end{proof}

\subsection{A lemma on the reduction of analytic varieties}
\label{section-variety-reduction}
In the present section we argue why the analytic variety defined as the zero set of an analytic function of two variables which has no local extremum must decompose into non-singular subvarieties, that is, into the union of inextendible (regularly parameterized) analytic curves. The argument was provided by Robert L.\ Bryant from Duke University, Durham, North Carolina by private communication. We are grateful towards him for his willingness to discuss this topic.

Note that the following lemma holds in more generality than just for our function $F$ as defined in \eqref{energy-equation-massive-case}. However, since we merely apply it to this function (possibly under an affine linear coordinate transformation), we keep using the same symbol.

\begin{lemma}\label{lem:variety_decomposition_lemma}
Let $F:\R^2\to\R$ be analytic such that $F(0,0)=0$ and let $F_n$ be the lowest-order non-vanishing homogeneous term in $F$'s Taylor expansion around $(0,0)$, say $F_n$ is a homogeneous polynomial of degree $n\in\N$. Suppose that all linear polynomials occurring in the factorization of $F_n$ into irreducibles $($over $\R)$ are pairwise distinct, say these are $m\le n$ in number. Then there exists a neighborhood $U\ni(0,0)$ such that
\[
    U\cap\big\{\,(y,z)\in\mathbb{R}^2\,\big|\,F(y,z)=0\,\big\}=\bigcup_{i=1}^m \gamma_i(J_i)
\]
with regularly parameterized analytic curves $\gamma_i:J_i\to U$ $($defined on some intervals $J_i\subset\R)$, $i=1,\dots,m$, which only intersect in $(0,0)$ and which are each linearly approximated around $(0,0)$ by the zero set of one of the linear factors of $F_n$.
\end{lemma}

\begin{proof}
By the assumptions of the lemma, one finds a linear coordinate transformation after which $F$ takes the form
\[
    F(y,z)=y\cdot F_{n-1}(y,z)+\sum_{k=n+1}^\infty F_k(y,z),
\]
where each $F_k$ is a homogeneous polynomial of order $k$, $k=n-1$ or $k\ge n+1$, and where $y$ does not divide $F_{n-1}$ in $\R[y,z]$. Note that $F_n(y,z)=y\cdot F_{n-1}(y,z)$ is the $F_n$ as labeled in the lemma and, moreover, that this power series for $F$ converges on a neighborhood of the origin.

By a blow-up substitution $y\mapsto yz$ we obtain
\[
    F(yz,z)=z^n\cdot\bigg(y\cdot \widetilde{G}(y)+\sum_{k=n+1}^\infty G_k(y,z)\bigg)=:z^n\cdot G(y,z)\,,
\]
where 
\[
G_k(y,z):=\frac{F_k(yz,z)}{z^n}\,=z^{k-n}F_k(y,1)\,,\qquad k\ge n+1\,,
\]
(using the homogeneity of $F_k$) are again polynomials of two variables and we set 
\[
\widetilde{G}(y)=\frac{F_{n-1}(yz,z)}{z^{n-1}}=F_{n-1}(y,1)\,.
\]
The map $(y,z)\mapsto G(y,z)$
defines another analytic function around the origin with $G(0,0)=0$ and the assumption that $y$ does not divide $F_{n-1}(y,z)$ implies that $\partial_yG(0,0)\neq0$. Consequently, the zero set of $G$ is, locally around $(0,0)$, given by the graph of an analytic curve of the form
\[
    J_1\mapsto\R^2,~z\mapsto\big(g(z),z\big)\,,
\]
$J_1\subset\R$ and $g:J_1\to\R$ analytic. Moreover, inverting the blow-up substitution yields that $F$ vanishes on the graph of the analytic curve
\begin{equation}\label{eq:solution_curve_rep_in_appendix_lemma}
    \gamma_1:J_1\mapsto\R^2,~z\mapsto\big(z\cdot g(z),z\big)\,,
\end{equation}
hence the analytic function
\[
    L_1:(y,z)\mapsto y-z\cdot g(z)
\]
is a prime factor of $F$ in the ring
\[
    \mathcal{R}:=\big\{\,K\in\R[\hspace{-1pt}[y,z]\hspace{-1pt}]\,\big|\,K\textup{ converges on some open neighborhood of }(0,0)\,\big\}
\]
of formal, locally convergent power series. Note that $G(0,0)=0$ implies $g(0)=0$ and thus $\gamma_1$ from \eqref{eq:solution_curve_rep_in_appendix_lemma} is, up to order $\mathcal{O}(z^2)$, approximated by the zero set of $(y,z)\mapsto y$, that is, of the factor of $F$'s lowest-order (non-vanishing) homogeneous term (LOHT) in consideration. Moreover, note that by \eqref{eq:solution_curve_rep_in_appendix_lemma} we have $|\gamma_1'|\ge1$, in particular, $|\gamma_1'|$ is bounded away from zero and \eqref{eq:solution_curve_rep_in_appendix_lemma} is indeed a regular parameterization.

However, since $L_1$ is a factor of $F$ we can decomposed $F$ into a product
\[
    F(y,z)=L_1(y,z)\cdot K_1(y,z)\,,
\]
for some $K_1\in\mathcal{R}$. By this product representation, multiplying the LOHTs of $L_1$ and $K_1$ must result in the LOHT of $F$ (as a product in $\R[y,z]$). Since the LOHT of $L_1$ is just $y$, the LOHT of $K_1$ consequently equals $F_{n-1}(y,z)$.

Finally, one can linearly transform the coordinates $y$ and $z$ to single out one of the remaining linear factors of $F_{n-1}(y,z)$ as $y$ again and repeat the above factorization procedure. Note that in the step specifying $G$'s zero set around $(0,0)$ we particularly restricted our considerations to the stripe defined by $z\in J_1$ (and even just an open subset of this stripe around $(0,0)$), and this new domain for the subsequent factorization step is linearly transformed as well.

Eventually, after $m$ steps we end up with a factorization
\begin{equation}
    F(y,z)=K(y,z)\cdot\prod_{j=1}^m L_j(y,z)\label{factorization}
\end{equation}
$(K=K_m)$, where each $L_i$ vanishes on an analytic curve $\gamma_i:J_i\to\R^2$. Hereby, we can find a sufficiently small open neighborhood $U$ of $(0,0)$ such that each of the power series in \eqref{factorization} converges on $U$ and such that for each $i\in\{1,\dots,m\}$ and each $(y,z)\in U$
\[
    L_i(y,z)=0\qquad\textup{if and only if}\qquad(y,z)\in\gamma_i(J_i)\,.
\]

Again, multiplying the LOHTs of each factor on the RHS of \eqref{factorization} must result in the LOHT of $F$, called $F_n$ above. By construction of the $L_i$, each of them has a LOHT equal to the respective linear prime factor of $F_n$. Consequently, the LOHT of $K$ is the product of the remaining non-linear prime factors of $F_n$. Hence, either $F_n$ has no second-order prime factors, then $K$ is a unit in $\R[\hspace{-1pt}[y,z]\hspace{-1pt}]$ and as such is non-zero on a neighborhood of $(0,0)$, w.l.o.g.\ on $U$, or $F_n$ has second-order prime factors, then the LOHT of $K$ is precisely the product of these. In the latter case, $K$ is non-zero on a $(0,0)$-pointed neighborhood of $(0,0)$, w.l.o.g.\ on $U\backslash\{(0,0)\}$, where we used that each second-order prime factor is non-zero on a pointed neighborhood of $(0,0)$. In any case, $K$ is non-zero on $U\backslash\{(0,0)\}$.

Concluding, by the factorization representation \eqref{factorization} of $F$ and the aforementioned properties of the factors, the zeros of $F$ in the open neighborhood $U$ are precisely the ranges of the analytic curves $\gamma_1,\dots,\gamma_m$.
\end{proof}

\begin{remark}\label{rem:variety_decomposition_remark}
\begin{itemize}
    \item[\textup{(i)}] In other words, the lemma states that the analytic variety defined as the zero set of $F$ can, locally around $(0,0)$, be decomposed into $m$ non-singular analytic subvarieties which are, to linear order, determined by the lowest order non-vanishing Taylor coefficients of $F$.
    \item[\textup{(ii)}]The assumption of the lemma is clearly imposed by the results of Section~\textup{\ref{Section-Non-existence-of-local-extrema}} $($and Assumption \textup{\ref{ass:first})}. That is, an indefinite Hessian of $F$ has two distinct eigenvectors, which allows for an affine linear coordinate transformation such that
    \[
    F(y,z)= yz+\mathcal{O}(\|(y,z)\|^3)\,.
    \]
    This representation of $F$, moreover, shows that there is an open neighborhood $U$ of $(0,0)$ in the given coordinates such that $\nabla F\neq0$ on $U\backslash\{(0,0)\}$, that is, $(0,0)$ is an isolated zero of $\nabla F$.
    \item[\textup{(iii)}] Note that the claim of distinct linear factors in $F_n$ is necessary, otherwise the analytic variety defined by $y^2=z^3$ provides a counter example. We note that this analytic variety is singular in the sense that neither of its two $($analytic$)$ branches $(0,\infty)\mapsto\mathbb{R}^2,~y\mapsto(\pm\, y,y^{\nicefrac{2}{3}})$ possesses an analytic continuation beyond the singular point $(0,0)$ $($which is approached as $y\to0)$, in particular, they are not the analytic continuations of one another.
    \item[\textup{(iv)}] For $n=1$ the lemma specializes into the implicit function theorem and the proof specializes into applying the latter $(\widetilde{G}$ is just constant in that case$)$.
\end{itemize}
\end{remark}

\subsection{Proof of Theorem~\ref{thm:solution_set_structure}}
\label{section_Proof_of_structure_thm}

In order to conclude the proof of Theorem~\ref{thm:solution_set_structure} we need a few more lemmata. Recall that $\mathcal{S}_{e_1,e_2}$ denotes the analytic variety defined as the zero set of $F$.

\begin{lemma}\label{lem:open_nbh_around_each_solution}
Let $e_1,e_2\in\R$ and suppose that Assumption \textup{\ref{ass:first}} holds. For any $s\in\mathcal{S}_{e_1,e_2}$ there exist an open neighborhood $U\ni s$ $($open in $(0,\infty)\times(0,\infty))$ such that $U\cap\mathcal{S}_{e_1,e_2}$ can be regularly parameterized by one or two analytic curves $I\to(0,\infty)\times(0,\infty)$.
\end{lemma}
\begin{proof}
Let $s\in\mathcal{S}_{e_1,e_2}\subset(0,\infty)\times(0,\infty)$, then either $\nabla F(s)\neq0$ or $\nabla F(s)=0$.

In the first case an open neighborhood $U$ and one (unique) curve as in the lemma are provided by the implicit function theorem in its analytic version.

In the second case $\nabla F(s)=0$ we conclude from the results of Section~\ref{Section-Non-existence-of-local-extrema} that $\textup{Hess}\,F(s)$ is indefinite (wherefore we need Assumption \textup{\ref{ass:first}}) and Lemma~\ref{lem:variety_decomposition_lemma} (cf.\ also Remark~\ref{rem:variety_decomposition_remark}) provides an open neighborhood $U$ and precisely two curves as in the lemma.

Note that both the implicit function theorem and Lemma~\ref{lem:variety_decomposition_lemma} represent the solution curves in a way such that $|\gamma'|$ can be uniformly bounded away from zero, imposing that all curves are regularly parameterized.
\end{proof}
\begin{lemma}\label{lem:solution_set_closed}
Let $e_1,e_2\in\R$. $\mathcal{S}_{e_1,e_2}$ is closed in $(0,\infty)\times(0,\infty)$.
\end{lemma}
\begin{proof}
$\mathcal{S}_{e_1,e_2}$ is the zero set of a continuous function.
\end{proof}

\begin{lemma}\label{lem:injectivity_lemma}
Let $e_1,e_2\in\R$ and suppose that $\gamma:I\to\mathcal{S}_{e_1,e_2}$ is an inextendible, regularly parameterized analytic solution curve. Moreover, suppose that Assumption \textup{\ref{ass:first}} holds. Then $\gamma$ is injective.
\end{lemma}
\begin{proof}
Suppose $\gamma$ is not injective, then there exist $a,b\in I$, $a<b$, such that $\gamma(a)=\gamma(b)$ and $\gamma$ is continuous on $[a,b]$ and analytic on $(a,b)\neq\emptyset$.

We first study the case $\nabla F(\gamma(a))=0$. Let $U$ be the open neighborhood of $\gamma(a)$ provided by Lemma~\ref{lem:open_nbh_around_each_solution} (wherefore we need Assumption \textup{\ref{ass:first}}). Moreover, let $\eta_j:(-\delta,\delta)\to U$, $j\in\{1,2\}$, $\delta>0$, be the regular parameterizations of $U\cap\mathcal{S}_{e_1,e_2}$ from the aforementioned lemma with $\eta_1(0)=\eta_2(0)=\gamma(a)$. Note that, if necessary, we can regularly reparameterize them to be defined on the same symmetric interval $(-\delta,\delta)$.

$\gamma$ is a continuous curve with $\gamma'(a)\ne0$. Thus we can assume that $U$ is small enough such that $\gamma$ takes at least one value outside $U$. Consequently, there exists $\varepsilon>0$ such that $\gamma\big|_{[a,a+\varepsilon)}$ coincides with precisely one of the four solution branches
\begin{equation}\label{eq:four_solution_branches}
    \eta_1\big|_{[0,\delta)}\,,\qquad\eta_1\big|_{(-\delta,0]}\,,\qquad\eta_2\big|_{[0,\delta)}\qquad\textup{or}\qquad\eta_2\big|_{(-\delta,0]}\,,
\end{equation}
up to reparameterization. W.l.o.g.\ we can label the $\eta_i$'s such that $\gamma\big|_{[a,a+\varepsilon)}$ coincides with $\eta_1\big|_{[0,\delta)}$. In particular, we can assume that $\gamma(t)\neq\gamma(a)$ for all $t\in(a,b)$, otherwise we replace $b$ by the smallest such point. Therefore note that, since $\gamma\big|_{[a,a+\varepsilon)}$ coincides with $\eta_1\big|_{[0,\delta)}$, such points $t$ with $\gamma(t)=\gamma(a)$ do not accumulate in $a$.

By the same argument as above, there exists $\widetilde{\varepsilon}>0$ such that $\gamma\big|_{(b-\widetilde{\varepsilon},b]}$ coincides with one of the four solution branches in \eqref{eq:four_solution_branches}. We go through the cases.

In the first case the curves $\gamma\big|_{[a,a+\varepsilon)}$ and $\gamma\big|_{(b-\widetilde{\varepsilon},b]}$ coincide, up to reparameterization. Explicitly, there exist an analytic reparameterization $\theta:(a,a+\varepsilon)\to(b-\widetilde{\varepsilon},b)$ which, by $\gamma(a)=\gamma(b)$, is monotonously decreasing. Moreover, $\gamma\big|_{(a,b)}$ represents an analytic continuation of both $\gamma\big|_{(a,a+\varepsilon)}$ and $\gamma\big|_{(b-\widetilde{\varepsilon},b)}$, hence the analytic reparameterization $\theta$ can be continuated to a monotonously decreasing reparameterization $\widehat{\theta}:(a,b)\to(a,b)$. Such a map has a fixed point $t_0\in(a,b)$, $\widehat{\theta}(t_0)=t_0$, in which
\[
\gamma'(t_0)=\widehat{\theta}'(t_0)\cdot\gamma'\big(\widehat{\theta}(t_0)\big)=\widehat{\theta}'(t_0)\cdot\gamma'(t_0)
\]
holds. $\widehat{\theta}'(t_0)<0$ implies $\gamma'(t_0)=0$ yielding a contradiction to $\gamma$ being regular.

In the other three cases of \eqref{eq:four_solution_branches}, that is, $\gamma\big|_{[a,a+\varepsilon)}$ coincides with $\eta_1\big|_{(-\delta,0]}$, $\eta_2\big|_{[0,\delta)}$ or $\eta_2\big|_{(-\delta,0]}$ up to reparameterization, we can make $U$ and $\delta$ smaller, such that $\mathcal{S}_{e_1,e_2}\cap\partial U$ consists precisely of the four points
\[
    \big\{\eta_1(\delta)=\gamma(a+\varepsilon)\,,\,\eta_1(-\delta)\,,\,\eta_2(\delta)\,,\,\eta_2(-\delta)\big\}=\mathcal{S}_{e_1,e_2}\cap\partial U\,.
\]
In particular, $\gamma(t)\notin\overline{U}=U\cup\partial U$ for all $t\in(a+\varepsilon,b-\widetilde{\varepsilon})$ and $\gamma(b-\widetilde{\varepsilon})\in\{\eta_1(-\delta)\,,\,\eta_2(\delta)\,,$ $\eta_2(-\delta)\}$. Thereby we obtain two distinct continuous curves from $\gamma(a+\varepsilon)$ to $\gamma(b-\widetilde{\varepsilon})$, one going through $U$ passing $\gamma(0)$ via $\gamma\big|_{[a,a+\varepsilon]}$ and $\gamma\big|_{[b-\widetilde{\varepsilon},b]}$ and the other outside of $\overline{U}$ along $\gamma\big|_{[a+\varepsilon,b-\widetilde{\varepsilon}]}$. By concatenating them we obtain a closed continuous curve along which $F$ vanishes and which encloses at least one point in which $F$ does not vanish, say $s_0\in U$ with $F(s_0)\neq0$. Thereby, the continuous curve in construction encloses $s_0$'s whole connected component of non-zeros of $F$, hence this connected component's closure is a compact subset of $(0,\infty)\times(0,\infty)$. The existence of such a set is excluded in Proposition~\ref{prop:bounded_connected_components}.

Concluding, if $\nabla F\big(\gamma(a)\big)=0$, any of the above possibility yields a contradiction. If, on the other hand, $\nabla F\big(\gamma(a)\big)\neq0$ the analytic implicit function theorem provides us an open neighborhood $U$ and a single curve $\eta:(-\delta,\delta)\to U$ to regularly parameterize $U\cap\mathcal{S}_{e_1,e_2}$. Replacing \eqref{eq:four_solution_branches} by the two branches $\eta\big|_{[0,\delta)}$ and $\eta\big|_{(-\delta,0]}$ these cases imply contradictions by same arguments as above.
\end{proof}

\begin{lemma}\label{lem:continuously_extending_to_limit_point}
Let $e_1,e_2\in\R$ and suppose that Assumption \textup{\ref{ass:first}} holds. Moreover, let
\[
\gamma=\big(\gamma^{(x)},\gamma^{(h)}\big):~(a,b]\to\mathcal{S}_{e_1,e_2}
\] be continuous, analytic on $(a,b)$ and regularly parameterized, in particular, $\gamma'(t)\ne0$ for all $t\in(a,b)$. Then either the limit $\lim\limits_{t\to a}\gamma(t)$ exists in $(0,\infty)\times(0,\infty)$, and thus $\gamma$ is continuously extendible to $[a,b]$,
or $\gamma(t)$ leaves any compact subset of $(0,\infty)\times(0,\infty)$, i.e.\ for all compact $K\subset(0,\infty)\times(0,\infty)$ there exists $\varepsilon>0$ such that 
\[
\gamma\big((a,a+\varepsilon)\big)\subset\big((0,\infty)\times(0,\infty)\big)\backslash K\,. 
\]
\end{lemma}
\begin{proof}
    Consider $\gamma$ as a map into the one-point compactification
    \[
    \overline{(0,\infty)\times(0,\infty)}=\Big((0,\infty)\times(0,\infty)\Big)\cup\{p\}
    \]
    of $(0,\infty)\times(0,\infty)$, then either $\gamma(t)$ converges to $p$ as $t\to a$ or not.

    Suppose $\gamma(t)\to p$ as $t\to a$. Given any compact $K\subset(0,\infty)\times(0,\infty)\subset\overline{(0,\infty)\times(0,\infty)}$, its complement $\overline{(0,\infty)\times(0,\infty)}\backslash K$ is an open neighborhood of $p$ and the limiting behavior of $\gamma$ provides us an $\varepsilon$ as in the lemma.

    If, on the other hand, $\gamma(t)$ does not converge to $p$ as $t\to a$, then there exist an open neighborhood $U$ of $p$ such that arbitrarily close to $a$ the curve $\gamma$ takes values in $K:=\overline{(0,\infty)\times(0,\infty)}\backslash U$. More precisely, there exist a sequence $(a_n)_{n\in\mathbb{N}}$ with $a_n\to a$ as $n\to \infty$ and $s_n:=\gamma(a_n)\in K$ for all $n$. Since $K$ is a compact set, the sequence $(s_n)_{n\in\mathbb{N}}$ accumulates is some $\widehat{s}$.

    In the following we show $(s_n)_{n\in\mathbb{N}}$ must converge to $\widehat{s}$. Note that we will pass to a subsequence of $(a_n)_{n\in\N}$ several times, but we will omit to give them a new label each time.

    Suppose that $(s_n)_{n\in\mathbb{N}}$ with $s_n=\gamma(a_n)$ as constructed above accumulates in $\widehat{s}$, but it does not converge to $\widehat{s}$.

By possibly passing to a subsequence, we can find an open neighborhood $U\subset(0,\infty)\times(0,\infty)$ of $\widehat{s}$ such that $s_{2n}\to\widehat{s}$ as $n\to\infty$ and $s_{2n+1}\notin U$ for all $n$. Moreover, we can achieve that the limit $a_n\to a$ is strictly monotonous.

Note that, since $s_n\in\mathcal{S}_{e_1,e_2}$ and $\mathcal{S}_{e_1,e_2}$ is closed (Lemma~\ref{lem:solution_set_closed}), also $\widehat{s}\in\mathcal{S}_{e_1,e_2}$, in particular, $F(\widehat{s})=0$. By Lemma~\ref{lem:open_nbh_around_each_solution} (which is where Assumption \textup{\ref{ass:first}} enters the argument) and possibly making $U$ smaller we find that $U\cap\mathcal{S}_{e_1,e_2}$ is parameterized by one regular curve $\eta:(-\delta,\delta)\to U$ with $\eta(0)=\widehat{s}$ (in the case $\nabla F(\widehat{s})\neq0$) or by two curves $\eta_i:(-\delta,\delta)\to U$, $i\in\{1,2\}$ (in the case $\nabla F(\widehat{s})=0$) with $\eta_1(0)=\eta_2(0)=\widehat{s}$.

By possibly passing to a subsequence again we can achieve that $s_{2n}\in U$ for all $n$ and that still $s_{2n}\to\widehat{s}$ and $s_{2n+1}\notin U$ for all $n$. Hence, all $(s_{2n})_{n\in\N}$ are contained in the range of $\eta$ or in the ranges of $\eta_1$ and $\eta_2$, respectively, according to the cases of $\nabla F(\widehat{s})=0$ or $\nabla F(\widehat{s})\neq0$. In the case of $\nabla F(\widehat{s})=0$, the range of at least one of $\eta_1$ and $\eta_2$ contains infinitely many of the $(s_{2n})_{n\in\N}$ and, by relabeling the $\eta_i$ and possibly passing to a subsequence again, we can achieve that all of the $(s_{2n})_{n\in\N}$ are contained in the range of $\eta:=\eta_1$.

Then there exists a sequence $(c_n)_{n\in\N}$ such that $\eta(c_n)=s_{2n}=\gamma(a_{2n})$. Moreover, since both $\gamma$ and $\eta$ are regularly parameterized analytic curves, there exist positive sequences $(\varepsilon^{(-)}_{2n})_{n\in\N}$ and $(\varepsilon^{(+)}_{2n})_{n\in\N}$ such that for all $n$ the curves
\[
    \gamma\big|_{(a_{2n}-\varepsilon_{2n}^{(-)},a_{2n}+\varepsilon_{2n}^{(+)})}\qquad\textup{and}\qquad\eta
\]
coincide, up to reparameterization. Since $U$ does not contain any $s_{2n+1}$ and since $a_n\to a$ monotonously, on each interval $(a_{2n+2},a_{2n})$ the curve $\gamma$ also takes values outside of $U$, for instance $\gamma(a_{2n+1})$ with $a_{2n+1}\in(a_{2n+2},a_{2n})$. Consequently, the intervals
\[
(a_{2n}-\varepsilon_{2n}^{(-)},a_{2n}+\varepsilon_{2n}^{(+)})\qquad\textup{with}\qquad n\in\N
\]
are pairwise disjoint. But then, for some $n\in\N$ the curves
\[
\gamma\big|_{(a_{2n}-\varepsilon_{2n}^{(-)},a_{2n}+\varepsilon_{2n}^{(+)})},\qquad\gamma\big|_{(a_{2n+2}-\varepsilon_{2n+2}^{(-)},a_{2n+2}+\varepsilon_{2n+2}^{(+)})}\qquad\textup{and}\qquad\eta
\]
have the same range and, in particular, $\gamma$ is not injective on $(a,b)$. This contradicts Lemma~\ref{lem:injectivity_lemma}.

Concluding, we have shown, that the values $\gamma(t)$ with $t\to a$ either leaves any compact subset of $(0,\infty)\times(0,\infty)$ (as in the lemma), or it accumulates in some $\widehat{s}\in(0,\infty)\times(0,\infty)$. Moreover, we have shown that in the latter case this accumulation point is already a limit.
\end{proof}

\begin{corollary}\label{cor:continuation_beyond_corollary}
Let $e_1,e_2\in\R$ and suppose that Assumption \textup{\ref{ass:first}} holds. Moreover, let $\gamma:[a,b]\to\mathcal{S}_{e_1,e_2}$ be the continuous extension of $\gamma:(a,b]\to\mathcal{S}_{e_1,e_2}$ from Lemma \textup{\ref{lem:continuously_extending_to_limit_point}} with $\gamma(a)=\lim\limits_{t\to a}\gamma(t)\in(0,\infty)\times(0,\infty)$. Then $\gamma$ is analytically continuable to $(a-\varepsilon,b]$ for some $\varepsilon>0$.
\end{corollary}
\begin{proof}
Analogously to the proof of Lemma~\ref{lem:injectivity_lemma} we find an open neighborhood $U\ni\gamma(a)$ and an analytic curve $\eta:(-\delta,\delta)\to U$ together with $\widetilde{\varepsilon}>0$ such that $\gamma\big|_{[a,a+\widetilde{\varepsilon})}$ and $\eta\big|_{[0,\delta)}$ coincide, up to reparameterization. But then $\gamma$ can be concatenated with (a commensurable reparameterization of) $\eta\big|_{(-\delta,0]}$ to obtain an analytic continuation of $\gamma$ beyond the point $\gamma(a)$ on some interval $(a-\varepsilon,b]$.
\end{proof}

\begin{lemma}\label{lem:tree-like-lemma}
Let $e_1,e_2\in\R$ and suppose that Assumption \textup{\ref{ass:first}} holds. Each connected component of $\mathcal{S}_{e_1,e_2}$ is tree-like.
\end{lemma}
\begin{proof}
We adjust the argument of Lemma~\ref{lem:injectivity_lemma}.

Let $\mathcal{S}$ be a connected component of $\mathcal{S}_{e_1,e_2}$. If $\mathcal{S}$ is not tree-like, there exists $s\in\mathcal{S}$ such that $\mathcal{S}\backslash\{s\}$ is still connected.
Let $U\ni s$ be the open neighborhood provided by Lemma~\ref{lem:open_nbh_around_each_solution} and let $\eta:(-\delta,\delta)\to U$ be \emph{the} (if $\nabla F(s)\neq0$) or \emph{one} of the two (if $\nabla F(s)=0$) corresponding regular parameterization(s) of $U\cap\mathcal{S}_{e_1,e_2}$ with $\eta(0)=s$. Since the range of $\eta$ is connected, $\eta$ is a map $(-\delta,\delta)\to U\cap\mathcal{S}$.

Since $\mathcal{S}\backslash\{s\}$ is still connected, there exists a continuous curve $\gamma:[a,b]\to\mathcal{S}\backslash\{s\}$ with $\gamma(a)=\eta(-\nicefrac{\delta}{2})$ and $\gamma(b)=\eta(\nicefrac{\delta}{2})$. On the other hand, $\eta\big|_{[-\nicefrac{\delta}{2},\nicefrac{\delta}{2}]}$ is a continuous curve in $\mathcal{S}$ connecting $\gamma(a)=\eta(-\nicefrac{\delta}{2})$ and $\gamma(b)=\eta(\nicefrac{\delta}{2})$ as well, and since $\gamma$ cannot take the value $s$, we have indeed two different continuous curves in $\mathcal{S}$ connecting the aforementioned two points. Hence, by concatenation, we obtain a closed continuous curve in $\mathcal{S}$ which, viewed as curve in $(0,\infty)\times(0,\infty)$, must enclose at least one non-solution (with $F(x,h)\neq0$) and thus its whole connected component of non-solutions. This connected component cannot exist as we have demonstrated in Lemma~\ref{prop:bounded_connected_components}.
\end{proof}

With the previous lemmata we can finally conclude Theorem~\ref{thm:solution_set_structure}. 
\begin{proof}[Proof of Theorem~\textup{\ref{thm:solution_set_structure}}]
In Lemma~\ref{lem:open_nbh_around_each_solution} we have seen that any solution $s\in\mathcal{S}_{e_1,e_2}$ belongs to the graph of an analytic curve. In Lemma~\ref{lem:continuously_extending_to_limit_point} and Corollary~\ref{cor:continuation_beyond_corollary} we have seen that any such solution curve has an analytic continuation until it leaves any compact subset of $(0,\infty)\times(0,\infty)$. In particular, any inextendible curve leaves any given compact subset of $(0,\infty)\times(0,\infty)$ both at  the infimum and the supremum of its domain.
Finally, the tree-like structure is proven in Lemma~\ref{lem:tree-like-lemma}. We have not shown that there are only finitely many solution curves. This assertion is postponed to Section~\ref{Asymptotic-Behavior-of-the-Solution-Set} (cf.\ Remark~\ref{rem:remark_after_structure_thm}), where we note that by the previous argument, any inextendible curve must be seen in the asymptotics of $F$. Note that Assumption \ref{ass:first} enters the theorem via the respective lemmata.
\end{proof}

\begin{remark}
    Lemma~\textup{\ref{lem:open_nbh_around_each_solution}} may be viewed as a generalized implicit function theorem for functions which, like $F$ does, fulfill that whenever $F(x,h)=0$ and $\nabla F(x,h)=0$ for some $(x,h)$, then $\textup{Hess}\,F(x,h)$ is indefinite. Hence, in the usual fashion of continuating solution curves provided by the implicit function theorem until that theorem is no longer applicable, it is already clear from Lemma~\textup{\ref{lem:open_nbh_around_each_solution}} that $\mathcal{S}_{e_1,e_2}$ consists of only ``infinitely long'' solution curves. Note that from Lemma~\textup{\ref{lem:variety_decomposition_lemma} (}in particular the power series representation in Remark~\textup{\ref{rem:variety_decomposition_remark})} it is easy to see that common zeros of both $F$ and $\nabla F$ cannot accumulate. However, then one is still burdened with excluding any kind of strange behavior an $($even analytic$)$ ``infinitely long'' curve can show, such as spiraling into a point $($with a non-integrable $|\gamma'|)$ or a ``topologist's sine curve''-like behavior $($e.g.\ via $x\mapsto\sin(\nicefrac{1}{x})\,)$. If one, in the end, aims to show that any such behavior must be visible in the asymptotics of $F$, one is right in the middle of proving Lemma~\textup{\ref{lem:continuously_extending_to_limit_point}}.
\end{remark}

\section{Asymptotic behavior of the solution set}
\label{Asymptotic-Behavior-of-the-Solution-Set}
After we have determined the structure of the solution set to $F(x,h)=0$ in Theorem~\ref{thm:solution_set_structure}, the present section addresses the identification of asymptotic solution curves. Under an asymptotic solution curve we understand a curve which in a certain region of large or small $x$ or $h$ approximates an actual solution curve to a sufficiently high order. Hereby, sufficiently high means that by this approximation we can isolate solution curves in order to count them. Below we will be more precise about this notion.

By the results of the previous section each solution curve must eventually terminate in such an asymptotic solution curve. Conversely, since there can be no additional compact connected components of solutions we find, by continuating the solution curves in the asymptotic regimes, already all solution curves constituting the solution set of $F(x,h)=0$.

Denote for $i\in\{0,-1\}$
\begin{equation}\label{eq:e1_bounds_with_Lambert}
    \beta_i:=-\frac{15}{2}W_i\bigl(-\mathrm{e}^{-2}\bigr)^2-15 W_i\bigl(-\mathrm{e}^{-2}\bigr)\,,
\end{equation}
where $W_0$ is the principal branch of the Lambert-$W$-function and $W_{\textup{-1}}$ is the $(\textup{-}1)$-st subprincipal branch. Numerically, $\beta_0\approx2.1903$ and $\beta_{\textup{-}1}\approx-27.046$. The Lambert-$W$-function's branches are the piecewise inverses of the map $\R\to\R,~t\mapsto t\textup{e}^t$ and occur in the asymptotic behavior of $F(x,h)$ at large $x$.

We introduce the following manner of speaking in order to characterize the asymptotic behavior of $\mathcal{S}_{e_1,e_2}$.
\begin{definition}\label{def:asymptotic_solution_curve_definition}
An asymptotic solution curve $\gamma:I\to(0,\infty)\times(0,\infty)$ in some limit
\begin{equation}\label{eq:four_possible_limits_of_as_sol_curve}
    x\to 0\,,\qquad x\to\infty\,,\qquad h\to0\qquad\textup{or}\qquad h\to\infty
\end{equation}
$($or a combination of these limits$)$ is a curve for which at least one of the limits \eqref{eq:four_possible_limits_of_as_sol_curve} holds and which approximates precisely one curve constituting $\mathcal{S}_{e_1,e_2}$ $($according to Theorem~\textup{\ref{thm:solution_set_structure})} to a sufficiently high order.
\end{definition}
\begin{remark}
\begin{itemize}
    \item[\textup{(i)}]Note that $F$ does not necessarily vanish along an asymptotic solution curve, that is, an asymptotic solution curve is not a solution curve in the sense of Section~\textup{\ref{positive-mass-simplification}}, merely an approximation thereof.
    \item[\textup{(ii)}] A priori Theorem {\textup{\ref{thm:solution_set_structure}}} does not rule out an oscillating behavior of solution curves as long as the solution curve simultaneously approaches the "boundary" of $(0,\infty)\times(0,\infty)$. For instance, the curve 
    \begin{align*}
    \gamma:(1,\infty)\to&\,(0,\infty)\times(0,\infty)\,,\\
        \gamma^{(x)}(t)&=
        t\cdot\left(\arctan\left(t^2\cos\left(t^2\right)\right)+\tfrac{\pi}{2}\right)\,,
        \\
        \gamma^{(h)}(t)&=
        t\cdot\left(\arctan\left(t^2\sin\left(t^2\right)\right)+\tfrac{\pi}{2}\right)
    \end{align*}
    is analytic, does leave any compact subset of $(0,\infty)\times(0,\infty)$ in the limit $t\to\infty$ $($we ignore for now what happens around $t\to 1)$, but none of the limits \eqref{eq:four_possible_limits_of_as_sol_curve} is approached as $t\to\infty$. We see in the following that the types of asymptotic solution curves as in Definition \textup{\ref{def:asymptotic_solution_curve_definition}} suffice to characterize our solution set $\mathcal{S}_{e_1,e_2}$.
    \item[\textup{(iii)}] We intentionally kept the approximation order of an asymptotic solution curve vague in the previous definition. What is a sufficient order depends on the different regimes, in particular, whether the graphs of $h\mapsto\big(X_{(e_1,\pm)}(h),h\big)$ $($cf.\ Section~\textup{\ref{section-Possible-non-h-solvable-points})} come amiss in the respective regime.

    For example, at small $x$ we can show for any parameter setting that the graph of $h\mapsto\big(X_{(e_1,\pm)}(h),h\big)$ is bounded away from $\mathcal{S}_{e_1,e_2}$ and a rather low order suffices. On the other hand, at large $x$ there are up to three solution curves in $\mathcal{S}_{e_1,e_2}$, all approximated by $x\mapsto\big(x,\nicefrac{\alpha}{\sqrt{x}}\big)$ for some values of $\alpha>0$, and, moreover, the graphs of $h\mapsto\big(X_{(e_1,\pm)}(h),h\big)$ are $($\,for particular values of $e_1)$ approximated by such curves as well. Hence, in the latter regime, we need to determine the solution curve's asymptotics to a higher order such that we can separate them both among each other and from the graphs of $h\mapsto\big(X_{(e_1,\pm)}(h),h\big)$.

    Note that we will mostly express an approximation order employing the Landau-function classes $\mathcal{O}(\cdot)$ and $\smalloh(\cdot)$.
\end{itemize}
\end{remark}

To this end, we state the main theorem of the present section.

\begin{theorem}\label{thm:asymptotic_solution_curves}
Let $e_1,e_2\in\R$ and suppose that Assumptions \textup{\ref{ass:first}} and \textup{\ref{ass:second}} hold. The solution set $\mathcal{S}_{e_1,e_2}$ of the massive consistency equation $F(x,h)=0$ in \eqref{energy-equation-massive-case} can be parameterized by three analytic curves in the cases
\begin{equation}\label{eq:cases_for_three_solution_curves}
    e_1\in(\beta_{\textup{-1}},\beta_0)\,,\qquad e_1=\beta_{\textup{-1}}\textup{ and }e_2\ge-10\qquad\textup{or}\qquad e_1=\beta_0\textup{ and }e_2\le-10
\end{equation}
and by two analytic curves otherwise. More detailed, we have:
\begin{itemize}
    \item[\textup{(i)}] Three asymptotic solution curves at large $h$ defined by the lines of constant $x\in\{0,x_{(-)},x_{(+)}\}$. Each of them approximates an analytic solution curve which can be represented as the graph of an analytic function $x\mapsto h(x)$ on an interval of the form $(0,\varepsilon)$, $(x_{(-)},x_{(-)}+\varepsilon)$ or $(x_{(+)},x_{(+)}+\varepsilon)$ $($for a sufficiently small $\varepsilon>0)$, respectively, and $h(x)\to\infty$ as $x\to0$, $x\to x_{(-)}$ or $x\to x_{(+)}$.
    \item[\textup{(ii)}] One asymptotic solution curve at small $h$ in the following cases.
    \begin{itemize}
        \item[\textup{(a)}] If $e_1<0$, the line of constant $x=0$ approximates an analytic solution curve which is representable as the graph of an analytic function $x\mapsto h(x)$ defined on an interval of the form $(0,\varepsilon)$ $($for some sufficiently small $\varepsilon>0)$ with $h(x)\to 0$ as $x\to 0$.
        \item[\textup{(b)}] If $e_1=0$ and $e_2<0$, the line of constant $h=0$ approximates an analytic solution curve which is representable as the graph of an analytic function $x\mapsto h(x)$ defined on an interval of the form $(2-\varepsilon,2)$ $($for some sufficiently small $\varepsilon>0)$ with $h(x)\to 0$ as $x\to 2$.
        \item[\textup{(c)}] If $e_1=e_2=0$, the line of constant $x=2$ approximates an analytic solution curve which is representable as the graph of an analytic function $x\mapsto h(x)$ defined on an interval of the form $(2,2+\varepsilon)$ $($for some sufficiently small $\varepsilon>0)$ with $h(x)\to 0$ as $x\to 2$.
        \item[\textup{(d)}] If $e_1=0$ and $e_2>0$, the line of constant $h=0$ approximates an analytic solution curve which is representable as the graph of an analytic function $x\mapsto h(x)$ defined on an interval of the form $(2,2+\varepsilon)$ $($for some sufficiently small $\varepsilon>0)$ with $h(x)\to 0$ as $x\to 2$.
        \item[\textup{(e)}] If $e_1>0$, the map
        \begin{equation}\label{eq:asym_sol_curve_in_thm}
            \gamma_\alpha:(x_0,\infty)\to(0,\infty)\times(0,\infty),~x\mapsto\big(\,x\,,\,\tfrac{\alpha}{\sqrt{x}}\big)
        \end{equation}
        for an appropriate $\alpha>0$ and some $x_0>0$ approximates an analytic solution curve at large $x$ which is also representable as the graph of an analytic function $x\mapsto h(x)$ defined on $(x_0,\infty)$ with
        \begin{equation}\label{eq:approximation_at_large_x_in_thm}
        h(x)\in\frac{\alpha}{\sqrt{x}}+\smalloh(x^{-\nicefrac{1}{2}})
        \end{equation}
        as $x\to \infty$.
    \end{itemize}
    \item[\textup{(iii)}] In the cases \eqref{eq:cases_for_three_solution_curves} there are two additional analytic solution curves which are approximated at large $x$ by $\gamma_\alpha$ from \eqref{eq:asym_sol_curve_in_thm} with appropriate $($not necessarily different$)$ values of $\alpha$. Also these solution curves are representable as in Case \textup{(ii).(e)} with approximation of the form \eqref{eq:approximation_at_large_x_in_thm}.
\end{itemize}
Consequently, these six $($in the cases \eqref{eq:cases_for_three_solution_curves}$)$ or four $($otherwise$)$ curves can be extended according to Theorem~\textup{\ref{thm:solution_set_structure}} in order to connect into three or two analytic curves, respectively, in a tree-like manner.
\end{theorem}
\begin{proof}
The presence of the solution curves and their asymptotics in the theorem are studied in Sections~\ref{section-Asymptotics-at-large-x}, \ref{Section:small_x_asymptotics} and \ref{Section-Asymptotics-at-finite-x}, treating separately the regimes of large $x$, of small $x$ and of $x$-values in between, respectively. The remaining assertions follow from Theorem~\ref{thm:solution_set_structure}.

We note that Assumption \textup{\ref{ass:second}} is used to exclude the existence of more solutions in the regime of large $h$-values than the ones we have found. In turn, Assumption \textup{\ref{ass:first}} is used for applying Theorem \textup{\ref{thm:solution_set_structure}} on the structure of $\mathcal{S}_{e_1,e_2}$. In particular, we rule out that there are compact connected components of the solution set which are not seen in the asymptotics and that the solution curves we have found are pairwise continued into one another. However, the existence of the solution curves listed in the theorem exclusively relies on analytical arguments.
\end{proof}

We will, moreover, distinguish cases where there exist solutions at conformal and minimal coupling in Section~\ref{Section-Minimal-and-conformal-coupling} as follows. Therefore we need the following assumptions on the values of $F$ along the curves of minimal and conformal coupling. These assumptions will be substantiated by an asymptotic analysis (cf.\ Lemmata \ref{lem:minimal_coupling_lemma_ass} and \ref{lem:conformal_coupling_lemma_ass}) combined with numerical evidence in Section \ref{Section-Minimal-and-conformal-coupling}.
\begin{assumption}\label{ass:third}
\begin{enumerate}
    \item[\textup{(i)}] Suppose that the $(h$-, $e_1$- and $e_2$-independent$)$ function 
    \[
    (0,\infty)\to\mathbb{R},~x\mapsto\frac{\mathrm{d}^2}{\mathrm{d} x^2}\,F\Big(x,\frac{1}{\sqrt{x}}\Big)
    \]
    is negative on its entire domain.
    \item[\textup{(ii)}] Suppose that the $(h$-, $e_1$- and $e_2$-independent$)$ function 
    \[
    (2,\infty)\to(0,\infty)\times(0,\infty),~x\mapsto\frac{\mathrm{d}^2}{\mathrm{d} x^2}\,F\Big(x,\frac{1}{\sqrt{x-2}\,}\Big)
    \]
    has exactly one zero in its domain.
\end{enumerate}
\end{assumption}
\begin{proposition}
Suppose that Assumption \textup{\ref{ass:third}} holds. In general, there are at most two solutions of the consistency equation \eqref{energy-equation-massive-case} with minimal coupling $\xi=0$ and at most three solutions with conformal coupling $\xi=\frac{1}{6}$. In dependence of the parameters $e_1$ and $e_2$ these bounds can be lowered.
\end{proposition}
\begin{proof}
This proposition immediately follows from the findings of Section~\ref{Section-Minimal-and-conformal-coupling}.
\end{proof}
This analysis allows to locate the solution curves more accurately and in some cases even shows which of the four or six solution curves in the asymptotics are analytic continuations of one another. For example, suppose we have found two solution curves in the asymptotics for which $\xi<\frac{1}{6}$ holds and two such curves with $\xi>\frac{1}{6}$ for parameters at which there are no solutions with conformal coupling $\xi=\frac{1}{6}$ and for which there are only four asymptotic solution curves in total. Then we can already conclude that the two curves with $\xi<\frac{1}{6}$ analytically extend into one another and likewise do the other two curves. Note that, however, the analysis of the particular $\xi\in\{0,\frac{1}{6}\}$-cases is not carried out into every last possible detail and the results should rather be viewed as a rough localization (in some cases) of the solution subvarieties which constitute the solution set $\mathcal{S}_{e_1,e_2}$ in the ambient plane $(0,\infty)\times(0,\infty)$.

\subsection{Asymptotics at large \headingmath $x$}
\label{section-Asymptotics-at-large-x}
In the present section we study the asymptotics at large $x$ in the following lemmata. At first, we prove a quite obvious lemma which will be refined afterwards.

\begin{lemma}\label{lem:solution_set_in_stripe_union}
There exists $(x_0,h_0)\in(0,\infty)\times(0,\infty)$ such that $F$ is strictly positive on $(x_0,\infty)\times(h_0,\infty)$. Consequently, any solution of $F(x,h)=0$ lies in the union of stripes $(0,x_0)\times(0,\infty)\,\cup\,(0,\infty)\times(0,h_0)$.
\end{lemma}
\begin{proof}
This lemma simply follows by noting that for large $x$ and for values of $h$ which are bounded away from $0$ the term $\frac{h^2x^2}{4}$ in $F_1$ (cf.\ \eqref{energy-equation-massive-case}) is dominant for $F$.
\end{proof}

For small $h$, the situation is more complicated as there are multiple competing terms in $F$. We study this limit by evaluating $F$ along the curves
\begin{equation}\label{eq:gamma_alpha_def}
    \gamma_\alpha:~(0,\infty)\to(0,\infty)\times(0,\infty),~x\mapsto\big(\,x\,,\,\tfrac{\alpha}{\sqrt{x}}\,\big)\,,\quad\alpha>0\,.
\end{equation}
Alternatively, studying $F$ on these curves may be viewed as a reparameterization of $(0,\infty)\times(0,\infty)$ into the coordinates $(x,\alpha)$. If we then find a solution curve $x\mapsto\big(x,\alpha(x)\big)$ for which $\alpha(x)$ converges to some $\alpha_0$ as $x\to\infty$, then in our original coordinates $(x,h)$ we have a solution curve which to order $\smalloh(x^{-\nicefrac{1}{2}})$ (for $h(x)$ as $x\to\infty$) is approximated by the curve $\gamma_{\alpha_0}$.

We obtain the equation
\begin{align}
    0&=F\Big(x,\frac{\alpha}{\sqrt{x}}\Big)\label{large-x-limit-equation}\\ &=s(\alpha)\,x+\Big(1-\alpha^2-\frac{e_2}{30}+2\log(\alpha)-\frac{x}{2}\big(f(x)-\log(x)\big)\Big)+\Big(\frac{29\alpha^2}{30x}+f(x)-\log(x)\Big)\,,\notag
\end{align}
where we grouped the terms according to their asymptotic behaviour in the limit $x\to\infty$ into the divergent term, the bounded (but non-zero) terms and the terms vanishing in that limit. For the dominant term we introduced the function
\[
    s:(0,\infty)\to\R,~\alpha\mapsto\frac{\alpha^2}{4}-\frac{1}{2}-\frac{e_1}{30\alpha^2}-\log(\alpha)\,.
\]
For the assignment of $f-\log$ to the terms with the respective behavior we refer to the asymptotic expansion of $f$ in Appendix~\ref{appendix-function-f}.

We are particularly interested in the zeros of the function $s$ and whether $s$ changes sign in its zeros. Note that, given $\alpha_0\in(0,\infty)$ with $s(\alpha_0)=0$ and $\varepsilon>0$ with (w.l.o.g., otherwise the same follows analogously) $s<0$ on $[\alpha_0-\varepsilon,\alpha_0)$ and $s>0$ on $(\alpha_0,\alpha_0+\varepsilon]$, we observe that
\[
    F\Big(x\,,\,\frac{\alpha_0-\varepsilon}{\sqrt{x}}\Big)\to-\infty\qquad\textup{as well as}\qquad F\Big(x\,,\,\frac{\alpha_0+\varepsilon}{\sqrt{x}}\Big)\to\infty\,,
\]
in particular, there exist $x_0$ large enough such that
\[
F\Big(x\,,\,\frac{\alpha_0-\varepsilon}{\sqrt{x}}\Big)<0\quad\textup{ and}\quad F\Big(x\,,\,\frac{\alpha_0+\varepsilon}{\sqrt{x}}\Big)>0
\]
for any $x>x_0$. Consequently, for any such $x$ there exists a solution $(x,h)$ of $F(x,h)=0$ with $\frac{\alpha_0-\varepsilon}{\sqrt{x}}<h<\frac{\alpha_0+\varepsilon}{\sqrt{x}}$. Below we will show that each such solution belongs to an analytic curve for which, consequently, $\gamma_{\alpha_0}$ is an asymptotic approximation. Zeros of $s$ without sign change are more subtle and will be treated below as well.

We study the function $s$. For the following lemma recall the definition and approximate values of $\beta_0\approx2.1903\in(0,\frac{15}{2})$ and $\beta_\textup{-1}\approx-27.046<0$ from \eqref{eq:e1_bounds_with_Lambert}.
\begin{lemma}\label{lem:zeros_of_s}
\begin{itemize}
    \item[\textup{(i)}] If $e_1>\beta_0$ the functions $s$ has precisely one zero and changes sign.
    \item[\textup{(ii)}] If $e_1=\beta_0$ the function $s$ has precisely two zeros and does change sign only in the larger one.
    \item[\textup{(iii)}] If $e_1\in(0,\beta_{\textup{0}})$ the function $s$ has precisely three zeros and changes sign in each.
    \item[\textup{(iv)}] If $e_1\in(\beta_{\textup{-1}},0]$ the function $s$ has precisely two zeros and changes sign in both.
    \item[\textup{(v)}] If $e_1=\beta_{\textup{-1}}$ the function $s$ has precisely one zero and does not change sign.
    \item[\textup{(vi)}] If $e_1<\beta_{\textup{-1}}$ the function $s$ has no zero.
\end{itemize}
\end{lemma}
\begin{proof}
Suppose $e_1> 0$. Taking three derivatives we find that 
\[
\frac{\textup{d}^3}{\textup{d}\alpha^3}\,\alpha \,s(\alpha)=\frac{e_1}{5\alpha^4}+\frac{1}{\alpha^2}+\frac{3}{2}
\]
is positive, hence $\alpha\mapsto\alpha s(\alpha)$ admits at most three zeros. Note that the subsets of $(0,\infty)$ on which the function $\alpha\mapsto\alpha\, s(\alpha)$ is positive, negative or zero coincide with the respective sets for $s$. Thus, also $s$ has at most three zeros. Moreover, noting that
\begin{equation}\label{eq:s_limits_for_pos_e1}
    \lim_{\alpha\to0}s(\alpha)=-\infty\qquad\textup{and}\qquad\lim_{\alpha\to\infty}s(\alpha)=+\infty\,,
\end{equation}
we find that if $s'$ has two distinct zeros, $s$ necessarily has a local maximum at the smaller zero of $s'$ and a local minimum at the larger one, or a saddle in both these zeros. If $s'$ has only one zero, $s$ has a saddle there and hence is a monotonous function. If $s'$ has no zero at all, $s$ is monotonous as well.

The zeros of $s'$ are to be found at the solutions of
\begin{equation} \label{eq:zeros_of_s_prime}
    2\alpha^3 s'(\alpha)=\alpha^4-2\alpha^2+\frac{2e_1}{15}=0\,,\qquad\quad\textup{i.e.\ at}\quad\alpha_{(\pm)}=\sqrt{1\pm\sqrt{1-\tfrac{2e_1}{15}}\,}\,,
\end{equation}
whenever this expression yields positive reals. In our present consideration of $e_1>0$, \eqref{eq:zeros_of_s_prime} has obviously at most one solution if $e_1\ge\frac{15}{2}$. This partially proves (i) for $e_1\ge\frac{15}{2}$ ($>\beta_0$).

If $e_1\in(0,\frac{15}{2})$ \eqref{eq:zeros_of_s_prime} has indeed two distinct positive solutions $\alpha_{(-)},\alpha_{(+)}$ with $\alpha_{(-)}<1<\alpha_{(+)}$. Consider the map
\begin{equation}\label{eq:alphasquared_times_sofalpha}
    (0,\tfrac{15}{2})\to\R,~e_1\mapsto\alpha_{(\pm)}^2s(\alpha_{(\pm)})=-\frac{e_1}{15}-\frac{1}{2}\Big(1\pm\sqrt{1-\tfrac{2e_1}{15}}\Big)
    \log\Big(1\pm\sqrt{1-\tfrac{2e_1}{15}}\,\Big)\,.
\end{equation}
We read off that $\alpha_{(+)}^2s(\alpha_{(+)})<0$ for $e_1\in(0,\frac{15}{2})$, in particular $s(\alpha_{(+)})<0$ and with the limits \eqref{eq:s_limits_for_pos_e1} we conclude that $s$ has a zero in $(\alpha_{(+)},\infty)$. The ``$-$''-branch, on the other hand, may take both positive and negative values. Compute
\begin{equation}
    \frac{\textup{d}^2}{{\textup{d}e_1}^2}\big(\alpha_{(-)}^2s(\alpha_{(-)})\big)=-\frac{1}{450\big(1-\frac{2e_1}{15}\big)^{\nicefrac{3}{2}}}\left(\frac{1}{1-\sqrt{1-\frac{2e_1}{15}}}+\log\Big(1-\sqrt{1-\tfrac{2e_1}{15}}\Big)\right)\label{second-deivative-of-alpha-squared-s-alpha}
\end{equation}
and note that the map $z\mapsto\frac{1}{1-z}+\log(1-z)$ takes only positive values on the interval $(0,1)$. Consequently, $e_1\mapsto\alpha_{(-)}^2s(\alpha_{(-)})$ defines a concave function on $(0,\frac{15}{2})$. By observing that $\alpha_{(-)}^2 s(\alpha_{(-)})\to-\frac{1}{2}$ as $e_1\to\frac{15}{2}$, that $\alpha_{(-)}^2 s(\alpha_{(-)})\to0$ as $e_1\to 0$ and that
\[
    \frac{\textup{d}}{{\textup{d}e_1}}\big(\alpha_{(-)}^2s(\alpha_{(-)})\big)=-\frac{1}{15}-\frac{1}{30\sqrt{1-\frac{2e_1}{15}}}\Big(\log\Big(1-\sqrt{1-\tfrac{2e_1}{15}}\Big)+1\Big)\to +\infty
\]
as $e_1\to0$ we conclude that $e_1\mapsto\alpha_{(-)}^2 s(\alpha_{(-)})$ has exactly one zero in the open interval $(0,\frac{15}{2})$. This zero is found at
\[
e_1=\beta_0:=-\frac{15}{2}W_0\big(-\mathrm{e}^{-2}\big)^2-15 W_0\big(-\mathrm{e}^{-2}\big)\,,
\]
which lead to the definition \eqref{eq:e1_bounds_with_Lambert}. Note that the Lambert-$W$-function, as a piecewise inverse of $t\mapsto t\textup{e}^t$, occurs naturally when solving the expression $\alpha_{(\pm)}^2s(\alpha_{(\pm)})=0$ in \eqref{eq:alphasquared_times_sofalpha} for $e_1$.

In particular, if $e_1\in(\beta_0,\frac{15}{2})$, the function $\alpha\mapsto\alpha^2s(\alpha)$ is negative in the smaller zero $\alpha_{(-)}$ of $s'$. Consequently, $s$ is negative on the interval $(0,\alpha_{(+)})$ and the zero of $s$ in the interval $(\alpha_{(+)},\infty)$ we have found above is the only zero of $s$. This completes the proof of (i).

If $e_1=\beta_0$, the function $\alpha\mapsto\alpha^2s(\alpha)$ vanishes in the smaller zero $\alpha_{(-)}$ of $s'$. Hence, $s$ has two zeros, namely $\alpha_{(-)}$ where it consequently does not change sign and the zero in the interval $(\alpha_{(+)},\infty)$ found above. This shows (ii). Note that $s$ is negative on a pointed neighborhood of $\alpha_{(-)}$.

Finally, if $e_1\in(0,\beta_0)$, the function $\alpha\mapsto\alpha^2s(\alpha)$ is positive in the local maximum $\alpha_{(-)}$ of $s$. Consequently,  also $s(\alpha_{(-)})>0$ and besides the zero larger than $\alpha_{(+)}$ from above, $s$ has a zero in the interval $(\alpha_{(-)},\alpha_{(+)})$ and a zero in the interval $(0,\alpha_{(-)})$. In each zero $s$ must change sign. This shows (iii).

To show (iv)-(vi) we first note that  $s''$ is positive for $e_1\le0$. Hence, $s$ is strictly convex and admits at most two zeros. Moreover, $s'$ has at most one zero. Taking into account the limits
\[
    \lim_{\alpha\to0}s(\alpha)=\lim_{\alpha\to\infty}s(\alpha)=+\infty\,,
\]
we find that $s$ has a global minimum which, consequently, is the only local extremum to be found at the unique positive solution $\alpha_{(+)}$ of \eqref{eq:zeros_of_s_prime}. Note that the formula for $\alpha_{(-)}$ does not yield a positive real for $e_1\le0$.

Completely analogous to the above analysis of $e_1\mapsto\alpha_{(-)}^2s(\alpha_{(-)})$ we study the function
\[
(-\infty,0]\to\R,~e_1\mapsto\alpha_{(+)}^2s(\alpha_{(+)})
\]
and find an expression for $\frac{\textup{d}^2}{{\textup{d}e_1}^2}\big(\alpha_{(+)}^2s(\alpha_{(+)})\big)$ similar to \eqref{second-deivative-of-alpha-squared-s-alpha} which now is positive, hence also $e_1\mapsto\alpha_{(+)}^2s(\alpha_{(+)})$ has at most two zeros. With the limits
\[
    \lim_{e_1\to-\infty}\alpha_{(+)}^2s(\alpha_{(+)})=+\infty\qquad\textup{and}\qquad\alpha_{(+)}^2s(\alpha_{(+)})\Big|_{e_1=0}=-\log(2)<0\,,
\]
said function has exactly one zero to be found at
\[
e_1=\beta_{\textup{-}1}:=-\frac{15}{2}W_{\textup{-}1}\big(-\mathrm{e}^{-2}\big)^2-15 W_{\textup{-}1}\big(-\mathrm{e}^{-2}\big)\,,
\]
in particular, it is positive on $(-\infty,\beta_{\textup{-}1})$ and negative on $(\beta_{\textup{-}1},0]$.

Consequently, if $e_1\in(\beta_{\textup{-}1},0]$ the function $\alpha\mapsto\alpha^2s(\alpha)$ is negative in the global minimum $\alpha_{(+)}$ of $s$, thus $s(\alpha_{(+)})<0$ and $s$ has exactly two zeros, one in the interval $(0,\alpha_{(+)})$ and one in $(\alpha_{(+)},\infty)$. In both zeros it changes sign and (iv) follows.

If $e_1=\beta_{\textup{-}1}$ the function $\alpha\mapsto\alpha^2s(\alpha)$ vanishes in $\alpha_{(+)}$, hence $s$ vanishes in its global minimum $\alpha_{(+)}$. This proves (v). Note that $s$ is positive everywhere but in $\alpha_{(+)}$.

Finally, if $e_1\in(-\infty,\beta_{\textup{-}1})$ the function $e_1\mapsto\alpha_{(+)}^2s(\alpha_{(+)})$ is positive in $\alpha_{(+)}$ and thus $s$ is positive in its global minimum. This shows (vi) and completes the proof.
\end{proof}

\begin{corollary}\label{cor:sign_change_s}
If $\alpha_0$ is a zero of $s$ with sign change, then $s'(\alpha_0)\neq0$.
\end{corollary}
\begin{proof}
In the previous lemma's proof the zeros of $s'$ were labeled $\alpha_{(\pm)}$. Moreover, the values $e_1\in\{\beta_{\textup{-}1},\beta_0\}$ where the only cases in which $s$ vanishes in a zero of $s'$. In these cases, in turn, $s$ did not change sign in its respective zero.
\end{proof}

\begin{lemma}\label{lem:zero_with_sing_change_yields_solcurve}
For any zero $\alpha_0$ at which $s$ chances sign the curve $\gamma_{\alpha_0}$ in \eqref{eq:gamma_alpha_def} is an asymptotic solution curve in the limit $x\to\infty$. The values of $\gamma_{\alpha_0}$ approximate exactly one analytic solution curve.
\end{lemma}
\begin{proof}
Above we have noted, if $s$ changes sign in a zero $\alpha_0$, then we find $\varepsilon>0$  such that (w.l.o.g., otherwise interchange `$>$' and `$<$' accordingly) $s<0$ on $[\alpha_0-\varepsilon,\alpha_0)$ and $s>0$ on $(\alpha_0,\alpha_0+\varepsilon]$ and thus there is $x_0>0$ such that
\[
F\Big(x\,,\,\frac{\alpha_0-\varepsilon}{\sqrt{x}}\Big)<0\quad\textup{ and}\quad F\Big(x\,,\,\frac{\alpha_0+\varepsilon}{\sqrt{x}}\Big)>0
\]
for any $x>x_0$. This implies that there is at least one solution $(x,h)$ of $F(x,h)=0$ with $h\in\big(\gamma_{\alpha_0-\varepsilon}(x),\gamma_{\alpha_0+\varepsilon}(x)\big)$ for such $x$.

By Corollary~\ref{cor:sign_change_s} we can make $\varepsilon$ small enough such that $s$ is monotonous on the interval $[\alpha_0-\varepsilon,\alpha+\varepsilon]$. Consequently, possibly by choosing a larger $x_0$, we have that also the map
\[
    [\alpha_0-\varepsilon,\alpha+\varepsilon]\to\R,~\alpha\mapsto F\Big(x\,,\,\frac{\alpha}{\sqrt{x}}\Big)
\]
is monotonous for each $x>x_0$ and thus, it has precisely one zero. Hence we obtain a map $h:(x_0,\infty)\to(0,\infty)$ with $F\big(x,h(x)\big)=0$  and $F\big(x,\widetilde{h}\big)\neq0$ for all $\widetilde{h}\in\big(\gamma_{\alpha_0-\varepsilon}(x),\gamma_{\alpha_0+\varepsilon}(x)\big)\backslash\{h(x)\}$, at all $x>x_0$.

To this point we do not know whether $h$ is analytic. However, since in the construction of $h$ we can shrink $\varepsilon>0$ at will, we know that
\[
h(x)\in\frac{\alpha_0}{\sqrt{x}}+\smalloh(x^{-\nicefrac{1}{2}})
\]
in the limit $x\to\infty$ and hence $\gamma_{\alpha_0}$ approximates the solution curves with an error better than any $\pm\frac{\varepsilon}{\sqrt{x}}$.

Recall the asymptotic expansion of $X_{(e_1,\pm)}$ stated in Lemma~\ref{lem:non-h-solvable_point-properties}. The leading order was given by
\begin{equation}\label{eq:expansion_X_pm_with_alphas}
    X_{(e_1,\pm)}(\widetilde{h})=\frac{\alpha_{(\pm)}^2}{\widetilde{h}^2}+\mathcal{O}(1)\qquad\textup{as }\widetilde{h}\to0
\end{equation}
(compare Lemma~\ref{lem:non-h-solvable_point-properties}.(ii) and Equation \eqref{eq:zeros_of_s_prime}). In the cases where $X_{(e_1,\pm)}(\widetilde{h})\to+\infty$ as $\widetilde{h}\to0$ (i.e.\ $e_1\in(0,\frac{15}{2}]$ for $X_{(e_1,-)}$ and $e_1\le\frac{15}{2}$ for $X_{(e_1,+)}$) we can solve the curves
\[
    (0,\varepsilon)\to(0,\infty)\times(0,\infty),~\widetilde{h}\mapsto\big(X_{(e_1,\pm)}(\widetilde{h}),\widetilde{h}\big)
\]
($\varepsilon>0$) for $\widetilde{h}$ and obtain
\[
    (x_1,\infty)\to(0,\infty)\times(0,\infty),~x\mapsto\big(x,h_{(\pm)}(x)\big)
\]
for a sufficiently large $x_1$ and some functions $h_{(\pm)}:(x_1,\infty)\to(0,\infty)$. By \eqref{eq:expansion_X_pm_with_alphas} we have
\[
h_{(\pm)}(x)=\frac{\alpha_{(\pm)}}{\sqrt{x}}+\smalloh(x^{-\nicefrac{1}{2}})
\]
in the limit $x\to\infty$. Since now $\alpha_0$ was assumed to be a zero with sign change, it coincides in particular neither with $\alpha_{(+)}$ nor with $\alpha_{(-)}$. Hence, the solutions $h(x)$ which we found above and the points of the form $\big(X_{(e_1,\pm)}(h),h\big)$ in which the solution set of $F(x,h)=0$ is not locally solvable for $h$ are asymptotically, at large $x$, bounded away from each other.  Consequently, $h$ is analytic.
\end{proof}
\begin{remark}\label{rem:all_solutions_in_certain_cases}
\begin{itemize}
    \item[\textup{(i)}]Note that to this point we have found all solutions in the limit $x\to\infty$ for $e_1\in(\beta_{\textup{-}1},\beta_0)$. By ``all solutions in the limit'' we mean that there exists $x_0$ such that any solution at $x>0$ belongs to one of the curves provided by Lemma~\textup{\ref{lem:zeros_of_s}}, Parts \textup{(iii),(iv)} and \textup{(v)} via Lemma~\textup{\ref{lem:zero_with_sing_change_yields_solcurve}}. That there can be no more solutions can be concluded using the upper bounds in Lemma~\textup{\ref{lem:solution_count_lemma_x_fixed}}, matching the amount of zeros of $s$.
    \item[\textup{(ii)}]Also, we have found all solution curves for $e_1>\frac{15}{2}$ using a similar argument. That is, $s$ in that case has one zero $\alpha_0$ and we have one solution curve approximated by the respective $\gamma_{\alpha_0}$. On the other hand, by Lemma~\textup{\ref{lem:non-h-solvable_point-properties}.(vi)} the functions $X_{(e_1,\pm)}$ are bounded from above, in particular, $\partial_hF(x,\cdot)$ has no zeros for $x$ above such an upper bound on $X_{(e_1,+)}$ and thus, at any $x$ above that bound, there exist at most one solution.
\end{itemize}
\end{remark}

The following lemma shows that any solution in the limit of large $x$ is approximated by $\gamma_\alpha$ with $s(\alpha)=0$.

\begin{lemma}
Any solution of $F(x,h)=0$ in the limit $x\to\infty$ is approximated in that limit by a curve of the form $\gamma_\alpha$ from \eqref{eq:gamma_alpha_def}, with a zero $\alpha$ of $s$. The approximation is valid in the function class $\smalloh(x^{-\nicefrac{1}{2}})$.
\end{lemma}
\begin{proof}
Note that, we have
\[
    \lim_{\alpha\to\infty}s(\alpha)=+\infty,\qquad\lim_{\alpha\to\infty}s'(\alpha)=+\infty\qquad\textup{as well as}\quad\lim_{\alpha\to\infty}s''(\alpha)=\frac{1}{2}\,,
\]
hence we can find $\alpha_1\in(0,\infty)$ such that
\[
    s(\alpha_1)>0,\qquad s'(\alpha_1)>0\qquad\textup{as well as}\quad s''(\alpha)\ge\frac{1}{4}\quad\textup{for all }\alpha\ge\alpha_1
\]
holds. By applying the fundamental theorem of calculus we also have
\[
    s'(\alpha)\ge s'(\alpha_1)\qquad\textup{and}\qquad s(\alpha)\ge s(\alpha_1)
\]
for all $\alpha\ge\alpha_1$ and thus we can find $x_1\in(0,\infty)$ such that
\[
\frac{\partial^2}{\partial\alpha^2}\,F\Big(x,\frac{\alpha}{\sqrt{x}}\Big)=s''(\alpha)x-2-\frac{2}{\alpha^2}+\frac{29}{15x}>0
\]
for all $x\ge x_1$ and all $\alpha\ge\alpha_1$.

From \eqref{large-x-limit-equation} we can read off that
\[
\lim_{x\to\infty}F\Big(x,\frac{\alpha_1}{\sqrt{x}}\Big)=+\infty\,.
\]
Moreover, by taking one $x$-derivative of \eqref{large-x-limit-equation} we can also read off that
\[
    \lim_{x\to\infty}~\frac{\partial}{\partial x} ~F\Big(x,\frac{\alpha_1}{\sqrt{x}}\Big)=s(\alpha_1)\,,
\]
where we used that $f-\log$ possesses an asymptotic Puiseux expansion allowing term-wise differentiation and, consequently, showing that $\frac{\mathrm{d}}{\mathrm{d} x}\frac{x}{2}\big(f(x)-\log(x)\big)=\mathcal{O}(\frac{1}{x})$ and $\frac{\mathrm{d}}{\mathrm{d} x}\big(f(x)-\log(x)\big)=\mathcal{O}(\frac{1}{x^2})$, cf.\ Appendix~\ref{appendix-function-f}. By enlarging $x_1$, if necessary, we can guarantee that
\[
    F\Big(x_1,\frac{\alpha_1}{\sqrt{x_1}}\Big)>0\qquad\textup{and}\qquad \frac{\partial }{\partial x}\,F\Big(x,\frac{\alpha_1}{\sqrt{x}}\Big)\ge\frac{s(\alpha_1)}{2}\quad\textup{for all }x\ge x_1\,.
\]

Finally, for any $(x,\alpha)\in[x_1,\infty)\times[\alpha_1,\infty)$ we conclude that
\begin{align*}
    F\Big(x,\frac{\alpha}{\sqrt{x}}\Big)
    &=
    F\Big(x,\frac{\alpha_1}{\sqrt{x}}\Big)
    +
    \frac{\partial}{\partial \alpha}F\Big(x,\frac{\alpha}{\sqrt{x}}\Big)\bigg|_{\alpha=\alpha_1}(\alpha-\alpha_1)
    +
    \int\limits_{\alpha_1}^{\alpha}\textup{d}\widetilde{\alpha}\int\limits_{\alpha_1}^{\widetilde{\alpha}}\textup{d}
    \begin{picture}(7,9)
    \setlength{\unitlength}{1pt}
    \put(0,1.5) {$\widetilde{{\color{white}\alpha}}$}
    \put(0,0) {$\widetilde{\alpha}$}
    \end{picture}
    ~\frac{\partial^2}{\partial
    \begin{picture}(7,9)
    \setlength{\unitlength}{1pt}
    \put(0,1.5) {$\widetilde{{\color{white}\alpha}}$}
    \put(0,0) {$\widetilde{\alpha}$}
    \end{picture}
    ^2} ~F\Big(x,\frac{
    \begin{picture}(7,9)
    \setlength{\unitlength}{1pt}
    \put(0,1.5) {$\scriptstyle\widetilde{{\color{white}\alpha}}$}
    \put(0,0) {$\scriptstyle\widetilde{\alpha}$}
    \end{picture}
    }{\sqrt{x}}\big)
    \\
    &\ge
    F\Big(x,\frac{\alpha_1}{\sqrt{x}}\Big)
    \\
    &=F\Big(x_1,\frac{\alpha_1}{\sqrt{x_1}}\Big)+\int\limits_{x_1}^x\textup{d}\widetilde{x}~\frac{\partial}{\partial\widetilde{x}}\,F\Big(\widetilde{x},\frac{\alpha_1}{\sqrt{\widetilde{x}\,}\,}\Big)
    \\
    &\ge
    F\Big(x_1,\frac{\alpha_1}{\sqrt{x_1}}\Big)
    \\
    &>0\,,
\end{align*}
where we applied the fundamental theorem of calculus several times and used the above estimates on the derivatives of $F$ in the coordinates $(x,\alpha)$ in the respective regime. Consequently, such $(x,\alpha)\in[x_1,\infty)\times[\alpha_1,\infty)$ cannot be a solution to \eqref{large-x-limit-equation}.

We work out a similar argument at small $\alpha$. In this way we distinguish the cases of positive, negative or zero $e_1$.

At first, in the case $e_1<0$ we have a similar argument as above. Therefore we reparameterize \eqref{large-x-limit-equation} via $\alpha\mapsto\frac{1}{\alpha}$, that is, we consider the equation

\begin{align*}
    &0=F\Big(x,\frac{1}{\alpha\sqrt{x}}\Big)\\ 
    &\!=\!s\Big(\frac{1}{\alpha}\Big)x+\Big(1-\frac{1}{\alpha^2}-\frac{e_2}{30}-2\log(\alpha)-\frac{x}{2}\big(f(x)-\log(x)\big)\Big)+\Big(\frac{29}{30x\alpha^2}+f(x)-\log(x)\Big)\,.
\end{align*}
Since $e_1<0$ we have
\[
    \lim_{\alpha\to\infty}s\Big(\frac{1}{\alpha}\Big)=+\infty\qquad\textup{as well as}\quad\lim_{\alpha\to\infty}\frac{\mathrm{d}}{\mathrm{d}\alpha}\,s\Big(\frac{1}{\alpha}\Big)=\lim_{\alpha\to\infty}\Big(-\frac{1}{2\alpha^3}-\frac{e_1\alpha}{15}+\frac{1}{\alpha}\Big)=+\infty\,,
\]
hence we can find $\alpha_2\in(0,\infty)$ such that
\[
    s\Big(\frac{1}{\alpha}\Big)\ge s(\alpha_2)>0\qquad\textup{as well as}\quad\frac{\mathrm{d}}{\mathrm{d}\alpha}\,s\Big(\frac{1}{\alpha}\Big)\ge1
\]
for all $\alpha\ge\frac{1}{\alpha_2}$. In this way we choose $1$ as an arbitrary, but still positive, lower bound on $\frac{\mathrm{d}}{\mathrm{d}\alpha}\,s\big(\frac{1}{\alpha}\big)$ above $\frac{1}{\alpha_2}$.

Now we can find $x_2\in(0,\infty)$ such that
\[
\frac{\partial}{\partial\alpha}\,F\Big(x,\frac{1}{\alpha\sqrt{x}}\Big)=\Big(\frac{\mathrm{d}}{\mathrm{d}\alpha}\,s\Big(\frac{1}{\alpha}\Big)\Big)x~+\frac{2}{\alpha^3}-\frac{2}{\alpha}-\frac{29}{15x\alpha^3}>0
\]
for all $x\ge x_2$ and all $\alpha\ge\frac{1}{\alpha_2}$. Again, we observe that
\[
    \lim_{x\to\infty}F\Big(x,\frac{\alpha_2}{\sqrt{x}}\Big)=+\infty\qquad\textup{as well as}\qquad\lim_{x\to\infty}\frac{\partial}{\partial x}\,F\Big(x,\frac{\alpha_2}{\sqrt{x}}\Big)=s(\alpha_2)\,,
\]
and by possibly enlarging $x_2$ we can guarantee that
\[
    F\Big(x_2,\frac{\alpha_2}{\sqrt{x_2}}\Big)>0\qquad\textup{and}\qquad \frac{\partial }{\partial x}\,F\Big(x,\frac{\alpha_2}{\sqrt{x}}\Big)\ge\frac{s(\alpha_2)}{2}\quad\textup{for all }x\ge x_2\,,
\]
again using the asymptotic expansion of $f-\log$.

Finally, we compute for any $(x,\alpha)\in[x_2,\infty)\times[\frac{1}{\alpha_2},\infty)$ that
\begin{align*}
    F\Big(x,\frac{1}{\alpha\sqrt{x}}\Big)&=F\Big(x,\frac{\alpha_2}{\sqrt{x}}\Big)+\int\limits_{\frac{1}{\alpha_2}}^\alpha\textup{d}\widetilde{\alpha}~\frac{\partial}{\partial\widetilde{\alpha}}\,F\Big(x,\frac{1}{\widetilde{\alpha}\sqrt{x}}\Big)\\
    &\ge F\Big(x,\frac{\alpha_2}{\sqrt{x}}\Big)\\
    &=F\Big(x_2,\frac{\alpha_2}{\sqrt{x_2}}\Big)+\int\limits_{x_2}^x\textup{d}\widetilde{x}\,F\Big(\widetilde{x},\frac{\alpha_2}{\sqrt{\widetilde{x}}}\Big)\\
    &\ge F\Big(x_2,\frac{\alpha_2}{\sqrt{x_2}}\Big)\\
    &>0\,,
\end{align*}
showing that the equation $F\big(x,\frac{1}{\alpha\sqrt{x}}\big)=0$ has no solution in $[x_2,\infty)\times[\frac{1}{\alpha_2},\infty)\ni(x,\alpha)$. Reversing our reparameterization $\alpha\mapsto\frac{1}{\alpha}$ from above, this means that the equation $F\big(x,\frac{\alpha}{\sqrt{x}}\big)=0$ has no solution in $[x_2,\infty)\times(0,\alpha_2]\ni(x,\alpha)$.

To find $\alpha_2$ and $x_2$ in the case $e_1>0$ is much easier. At first, observe that
\[
    \lim_{\alpha\to0}s(\alpha)=-\infty\qquad\textup{as well as}\qquad\lim_{\alpha\to0}s'(\alpha)=+\infty\,,
\]
hence there exist $\alpha_2\in(0,\infty)$ such that
\[
    s(\alpha)\le s(\alpha_2)<0\qquad\textup{for all }\alpha\le\alpha_2\,.
\]
From \eqref{large-x-limit-equation} we can read off that there exist $M\in\R$ such that
\[
    F\Big(x,\frac{\alpha}{\sqrt{x}}\Big)-s(\alpha)x\le M
\]
for all $\alpha\in(0,\alpha_2)$ and all $x\ge\frac{9}{4}$ (we employ this lower bound on $x$ in favour of the explicit bound in \eqref{bound-on-f-minus-log}, Appendix~\ref{appendix-function-f}). In particular, we can find $x_2\in(0,\infty)$ such that for all $x\ge x_2$
\[
F\Big(x,\frac{\alpha}{\sqrt{x}}\Big)\le s(\alpha)x+M<0\,,
\]
for all $\alpha\in(0,\alpha_2]$. Consequently, $[x_2,\infty)\times(0,\alpha_2]$ contains no solution of $F\big(x,\frac{\alpha}{\sqrt{x}}\big)=0$.

As a side remark, note that $M$ is indeed a uniform bound. Such uniform bound cannot be found in the cases we have treated before. Moreover, note that we comment on the case $e_1=0$ at the end of the proof.

To this point, if $e_1\neq0$ we have found $x_1,x_2,\alpha_1$ and $\alpha_2$, such that any solution of $F\big(x,\frac{\alpha}{\sqrt{x}}\big)=0$ with $x\ge\max\{x_1,x_2\}$ must fulfill $\alpha\in[\alpha_2,\alpha_1]$.

Now let $\varepsilon>0$ and suppose we have $\alpha\in[\alpha_2,\alpha_1]$ such that $|s(\alpha)|\ge\varepsilon$. We assume that $\varepsilon$ is small enough, such that $s(\alpha_1)>\varepsilon$ and $|s(\alpha_2)|>\varepsilon$ (note that $s(\alpha_1)>0$ for all $e_1\neq 0$). Then, similar as above, we read off from \eqref{large-x-limit-equation} that there exists $M\in\R$ such that
\[
    \Big|F\Big(x,\frac{\alpha}{\sqrt{x}}\Big)-s(\alpha)x\Big|\le M
\]
for all $\alpha\in[\alpha_2,\alpha_1]$ and all $x\ge\frac{9}{4}$. That this bound is now uniform in $\alpha$ is a consequence of restricting the $\alpha$-values to the compact interval $[\alpha_2,\alpha_1]$. However, we can find $x_3\in(\max\{x_1,x_2\},\infty)$, such that for all $x\ge x_3$ and all $\alpha\in[\alpha_2,\alpha_1]$
\[
    \Big|F\Big(x,\frac{\alpha}{\sqrt{x}}\Big)\Big|\ge\Big|s(\alpha)x\Big|-\Big|F\Big(x,\frac{\alpha}{\sqrt{x}}\Big)-s(\alpha)x\Big|\ge\varepsilon x-M>0
\]
holds. In particular, for any solution of $F\big(x,\frac{\alpha}{\sqrt{x}}\big)=0$ with $x\ge x_3$ we have shown that necessarily $s(\alpha)<\varepsilon$ holds. As $\varepsilon>0$ can be chosen arbitrarily small, the claim of the lemma follows for our present case $e_1\neq0$.

Finally, if $e_1=0$ one could work out a similar method. However, the occurring $\log$-terms (which are dominant in the $e_1=0$-case) prevent from analog estimates on the derivatives of $s$ and $\alpha\mapsto F\big(x,\frac{\alpha}{\sqrt{x}})$ and an adaption is not straightforward. On the other hand, a more profound method is not necessary. As we have already noted in Remark~\ref{rem:all_solutions_in_certain_cases}, we have found all solutions at large $x$ in the case $e_1=0$ together with the correct approximation using Lemmata~\ref{lem:zeros_of_s}.(iv) and \ref{lem:zero_with_sing_change_yields_solcurve} as well as the respective upper bound in Lemma~\ref{lem:solution_count_lemma_x_fixed}.
\end{proof}

\begin{remark}
At this point, only the cases $e_1\in\{\beta_0,\beta_{\textup{-}1}\}$ remain open. Actually, in addition to the cases commented on in Remark~\textup{\ref{rem:all_solutions_in_certain_cases}} $($i.e.\ for $e_1\in(\beta_0,\beta_{\textup{-}1})\,)$, we have shown that for $e_1>\beta_0$ there exists precisely one asymptotic solution curve $\gamma_\alpha$ with the single zero $\alpha$ of $s$ $($cf.\ Lemma~\textup{\ref{lem:zeros_of_s}.(i)\,)}, whereas for $e_1<\beta_{\textup{-}1}$ there exists no solution at all above a certain bound on $x$-values as $s$ has no zero in that case $($cf.\ Lemma~\textup{\ref{lem:zeros_of_s}.(vi)\,)}.
\end{remark}

\begin{lemma}
\begin{itemize}
    \item[\textup{(i)}] Let $e_1=\beta_0$ and let $\alpha_{(-)}$ be the smaller zero of $s$ provided by Lemma~\textup{\ref{lem:zeros_of_s}.(ii)}, defined in Equation \eqref{eq:zeros_of_s_prime}. If $e_2\le-10$, there exist two analytic solution curves at large $x$ approximated by $\gamma_{\alpha_{(-)}}$, whereas if $e_2>-10$ there exist no solution curves approximated by said curve.
    \item[\textup{(ii)}] Let $e_1=\beta_{\textup{-}1}$ and let $\alpha_{(+)}$ be the only zero of $s$ provided by Lemma~\textup{\ref{lem:zeros_of_s}.(ii)}, defined in Equation \eqref{eq:zeros_of_s_prime}. If $e_2\ge-10$, there exist two analytic solution curves at large $x$ approximated by $\gamma_{\alpha_{(+)}}$, whereas if $e_2<-10$ there exist no solution curves approximated by said curve.
\end{itemize}
In all assertions the approximation order is given by $\frac{\alpha}{\sqrt{x}}+\smalloh(\frac{1}{\sqrt{x}}\big)$.
\end{lemma}
\begin{proof}
For $e_1=\beta_0$ the function $\alpha\mapsto F\Big(x,\frac{\alpha}{\sqrt{x}}\Big)$ has, at sufficiently large $x$, a local maximum at the smaller solution of
\begin{equation}\label{eq:minimum_defining_equation_at_large_x}
    \frac{\partial}{\partial\alpha}\,F\Big(x,\frac{\alpha}{\sqrt{x}}\Big)=0\,.
\end{equation}
On the other hand, for $e_1=\beta_{\textup{-}1}$ said function has a local minimum at the only real (positive) solution of \eqref{eq:minimum_defining_equation_at_large_x}. For these assertions we recall the shape of the dominating function $s$ from previous lemmata.

Note that \eqref{eq:minimum_defining_equation_at_large_x} is a polynomial equation (as \eqref{equation-for-possible-bad-points} was, too) and we denote its solutions as the functions $\widetilde{\alpha}_{(\pm)}:(x_0,\infty)\to(0,\infty)$ defined by
\begin{align}
    \widetilde{\alpha}_{(\pm)}(x)^2:=&\,\tfrac{
x\Big(\frac{x}{2}-1\Big)}{
2\Big(\frac{x^2}{4}-x+\frac{29}{30}\Big)
}\left[~1~\pm~\sqrt{1-\frac{2e_1}{15}+\tfrac{\tfrac{4e_1}{225}}{\big(x-2\big)^2}}~\right]\notag\\
=&\,\alpha_{(\pm)}^2+\mathcal{O}\Big(\frac{1}{x}\Big)\label{eq:alpha_tilde_asymptotic_expansion}
\end{align}
with some sufficiently large $x_0$ and the zeros $\alpha_{(\pm)}$ of $s'$ from \eqref{eq:zeros_of_s_prime}. A lengthy but straightforward computation shows the surprisingly simple result that
\begin{equation}\label{eq:long_long_annoying_computation_result}
     F\Big(x,\frac{\widetilde{\alpha}_{(\pm)}(x)}{\sqrt{x}}\Big)=-\frac{e_2+10}{30}~+~\frac{1+W_j\big(-\mathrm{e}^{-2}\big)}{30}~\frac{1}{x}~+~\mathcal{O}\Big(\frac{1}{x^2}\Big)
\end{equation}
for the presently relevant cases $(\pm,j,e_1)=(-,0,\beta_0)$ or $(\pm,j,e_1)=(+,-1,\beta_{\textup{-}1})$.

Now let $e_1=\beta_0$. In Lemma~\ref{lem:zeros_of_s}.(ii) we have shown that $s$ does not change sign at its zero $\alpha_{(-)}$. As we, moreover, have remarked in the proof of said lemma, $s$ is negative on a pointed neighborhood of $\alpha_{(-)}$. Consequently, for any sufficiently small $\varepsilon>0$, we have
\[
    \lim_{x\to\infty} F\Big(x,\frac{\alpha_{(-)}-\varepsilon}{\sqrt{x}}\Big)
    =\lim_{x\to\infty} F\Big(x,\frac{\alpha_{(-)}+\varepsilon}{\sqrt{x}}\Big)
    =\lim_{x\to\infty} s(\alpha_{(-)}\pm\varepsilon)\,x
    =-\infty\,,
\]
independently of $e_2$, in particular, the respective expressions are negative for $x$ above some certain common lower bound $x_1\in(0,\infty)$. On the other hand, from the asymptotic expansion \eqref{eq:long_long_annoying_computation_result} we read off for $e_2\le-10$ that
\[
    F\Big(x,\frac{\widetilde{\alpha}_{(-)}(x)}{\sqrt{x}}\Big)>0
\]
for sufficiently large $x$, w.l.o.g.\ (if necessary enlarge $x_1$) we assume that the latter inequality holds for all $x\ge x_1$. For the boundary case $e_2=-10$ the sub-leading order in \eqref{eq:long_long_annoying_computation_result} is decisive and we note that $\frac{1}{30}\big(1+W_0(-\mathrm{e}^{-2})\big)\approx0.02805>0$.

By possibly enlarging $x_1$ again and by recalling the asymptotic expansion \eqref{eq:alpha_tilde_asymptotic_expansion} we can, moreover, guarantee that
\[
    \widetilde{\alpha}_{(-)}(x)\in(\alpha_{(-)}-\varepsilon,\alpha_{(-)}+\varepsilon)
\]
for all $x\ge x_1$. Hence, for all $x\ge x_1$ the equation $F\big(x,\frac{\alpha}{\sqrt{x}}\big)=0$ has solutions $\alpha$ in both the open intervals
\begin{equation}\label{eq:some_open_intervals_in_some_lemma}
    \big(\alpha_{(-)}-\varepsilon,\widetilde{\alpha}_{(-)}(x)\big)\qquad\textup{and}\qquad\big(\widetilde{\alpha}_{(-)}(x),\alpha_{(-)}+\varepsilon\big)\,.
\end{equation}
Together with the third solution at such sufficiently large $x$-values (if necessary, enlarge $x_1$ once again) provided by Lemma~\ref{lem:zero_with_sing_change_yields_solcurve} around the larger (sign-changing) zero of $s$, we already exhaust the upper bound given in Lemma~\ref{lem:solution_count_lemma_x_fixed}. Consequently, we can choose $x_1$ such that each of the intervals \eqref{eq:some_open_intervals_in_some_lemma} contains precisely one solution, for all $x\ge x_1$.

Note that, not only does the curve defined by $\widetilde{\alpha}_{(-)}$ provide points in which $F$ is positive, also does its range contain all points (around $\alpha_{(-)}$, particularly, bounded away from $\alpha_{(+)}$) in which the hypothesis of the implicit function theorem fails. Since we have found our two solutions in the \emph{open} intervals \eqref{eq:some_open_intervals_in_some_lemma}, the latter theorem's hypothesis fails in none of these solutions and we obtain, for now in the coordinates $(x,\alpha)$, indeed two analytic curves as the graphs of analytic functions of $x$ which are, to order $\smalloh(1)$ (recall that $\varepsilon$ can be chosen arbitrarily small), approximated by $\alpha_{(-)}$. Inverting the reparameterization $(x,h)\mapsto(x,\alpha)$, finally, provides us with the analytic solution curves in our original coordinates $(x,h)$, approximated in $\smalloh\big(\frac{1}{\sqrt{x}}\big)$ as asserted in the lemma.

For $e_2>-10$ we read off from the asymptotic expansion \eqref{eq:long_long_annoying_computation_result} that there exists some $x_1\in(0,\infty)$ such that
\[
    F\Big(x,\frac{\widetilde{\alpha}_{(\pm)}(x)}{\sqrt{x}}\Big)<0
\]
for all $x\ge x_1$. In particular, for some given (sufficiently small) $\varepsilon>0$ the function
\begin{equation}\label{eq:function_araound local maximum}
    (\alpha_{(-)}-\varepsilon,\alpha_{(-)}+\varepsilon)\to\R,~\alpha\mapsto F\Big(x,\frac{\alpha}{\sqrt{x}}\Big)
\end{equation}
is negative in its local maximum, for each $x\ge x_1$. Since \eqref{eq:minimum_defining_equation_at_large_x} has precisely the two solutions \eqref{eq:alpha_tilde_asymptotic_expansion}, we can choose $\varepsilon$ small enough such that $\widetilde{\alpha}_{(+)}(x)\notin(\alpha_{(-)}-\varepsilon,\alpha_{(-)}+\varepsilon)$ for all $x\ge x_1$, in other words, choose $\varepsilon$ small enough such that said negative local maximum serves as a uniform upper bound on the interval $(\alpha_{(-)}-\varepsilon,\alpha_{(-)}+\varepsilon)$. From the asymptotic expansion \eqref{eq:long_long_annoying_computation_result} we read off that this uniform upper bound on \eqref{eq:function_araound local maximum} can be estimated from above by $-\frac{1}{2}\,\frac{e_2+10}{30}<0$ for all sufficiently large $x$, w.l.o.g.\ for $x\ge x_1$. Since at this point we have found a uniform, $x$-independent, negative upper bound on the values of $F\big(x,\frac{\alpha}{\sqrt{x}}\big)$, we conclude that our equation $F\big(x,\frac{\alpha}{\sqrt{x}}\big)=0$ has no solution with $x\ge x_1$ and $\alpha\in(\alpha_{(-)}-\varepsilon,\alpha_{(-)}+\varepsilon)$. In our original coordinates $(x,h)$ this imposes what we have claimed in the lemma, finishing the proof of (i).

The proof of (ii) works in principle the same. One merely has to reverse inequalities and take into account that now the upper bound on the number of solution from Lemma~\ref{lem:solution_count_lemma_x_fixed} is 2 and the two solutions in the intervals analog to \eqref{eq:some_open_intervals_in_some_lemma} already exhaust said upper bound. Moreover, for the boundary case $e_2=-10$ we note that now $\frac{1}{30}\big(1+W_{\textup{-}1}(-\mathrm{e}^{-2})\big)\approx-0.07154<0$, though the main argument remains the same. This finishes the proof.
\end{proof}

\begin{table}[t!]
    \begin{center}
    \begin{tabular}{llll}
         Case &~& $\#$ & $\alpha$\\\hline
         $e_1<\beta_{-1}$ &$e_2$ arb.& 0 & $-$ \\
         $e_1=\beta_{-1}$ &$e_2<-10$& 0 & $-$  \\
         $e_1=\beta_{-1}$ &$e_2\ge-10$& 2 & $\alpha_1=\alpha_2=\alpha_{(+)}$\\
         $e_1\in(\beta_{-1},0]$ &$e_2$ arb.& 2 & $\alpha_1<\alpha_{(+)}<\alpha_2$\\
         $e_1\in(0,\beta_{0})$ &$e_2$ arb.& 3 & $\alpha_1<\alpha_{(-)}<\alpha_2<\alpha_{(+)}<\alpha_3$\\
         $e_1=\beta_{0}$ &$e_2\le-10$& 3 & $\alpha_1=\alpha_2=\alpha_{(-)}<\alpha_{(+)}<\alpha_3$\\
         $e_1=\beta_{0}$ &$e_2>-10$& 1 & $\alpha>\alpha_{(+)}$\\
         $e_1\in(\beta_{0},\frac{15}{2})$ &$e_2$ arb.& 1 & $\alpha>\alpha_{(+)}$\\
         $e_1=\frac{15}{2}$ &$e_2$ arb.&  1 & $\alpha>1~\big(=\alpha_{(-)}=\alpha_{(+)}\,\big)$ \\
         $e_1>\frac{15}{2}$ &$e_2$ arb.&  1 & $\alpha>1$ (and $\alpha_{(\pm)}\notin\R$)
    \end{tabular}
    \begin{minipage}{.9\textwidth}
    \captionsetup{format=plain, labelfont=bf}
    \caption{The possible cases which can occur for solution curves in the limit $x\to\infty$, distinguished by value ranges of the parameters $e_1$ and $e_2$. For each case we have listed how many branches the solution set $\mathcal{S}_{e_1,e_2}$ of $F(x,h)=0$ has in the limit $x\to\infty$ (column $\#$), together with rough bounds on the coefficient $\alpha$ in terms of $\alpha_{(\pm)}$ (if possible) such that the solutions of the respective branch are approximated by $\gamma_\alpha(x)=\big(x,\frac{\alpha}{\sqrt{x}}\big)$ (column $\alpha$). If necessary, the $\alpha$-values are labeled increasingly.\label{table-large-x-asymptotics}}
    \end{minipage}
    \end{center}

\end{table}

Finally, we collect the results of the present section in a list in Table~\ref{table-large-x-asymptotics}. Recall that for $0<e_1<\frac{15}{2}$ we have $\alpha_{(-)}<1<\alpha_{(+)}$, that for $e_1=\frac{15}{2}$ we have $\alpha_{(\pm)}=1$, that for $e_1\le0$ only $\alpha_{(+)}$ is positive and for $e_1>\frac{15}{2}$ none of the $\alpha_{(\pm)}$ is a real number. Moreover, recall that in each possible case we can find $x_0>0$ such that either $F(x,h)=0$ admits no solution at all for $x>x_0$ or such that each solution of $F(x,h)=0$ is approximated by a curve $x\mapsto\gamma_\alpha(x)=\big(x,\frac{\alpha}{\sqrt{x}}\big)$ in the limit $x\to\infty$ with some $\alpha\in(0,\infty)$ for which $s(\alpha)=0$.

\subsection{Asymptotics at small \headingmath $x$}
\label{Section:small_x_asymptotics}
At first we study the situation of $h$ away from the limits of small or large $h$.
\begin{lemma}\label{lem:small_x_finite_h_regime}
For any compact interval $[a,b]\subset(0,\infty)$ there exists $\varepsilon>0$ such that $F(x,h)<0$ for all $(x,h)\ni(0,\varepsilon)\times[a,b]$. Consequently, the equation $F(x,h)=0$ has no solution in $(0,\varepsilon)\times[a,b]$.
\end{lemma}
\begin{proof}
Note that for any fixed $h\in[a,b]$ we have $F(x,h)\to-\infty$ as $x\to0$ since the dominant term in this limit is $-(\frac{x}{2}-1) f(x)$ in $F_2$. Recall that $f(x)=-\frac{3}{x}+\mathcal{O}(1)$ in said limit (cf.\ Appendix~\ref{appendix-function-f}). Any other term in $F$ can be continuated to $x=0$, that is, there exists a continuous function
\[
    \widetilde{F}:[0,1]\times[a,b]\to\R\qquad\textup{s.t.}\quad \widetilde{F}(x,h)=F(x,h)+\Big(\frac{x}{2}-1\Big) f(x)
\]
for all $x>0$ and all $h\in[a,b]$. Let $M$ be an upper bound for the continuous function $\widetilde{F}$ on its compact domain. By the asymptotic expansion of $f$ we find $\varepsilon>0$ such that
\[
-\Big(\frac{x}{2}-1\Big) f(x)<-M
\]
for all $x<\varepsilon$. Therefore, we have for all $h\in[a,b]$ and all $x\in(0,\varepsilon)$
\[
    F(x,h)=\widetilde{F}(x,h)-\Big(\frac{x}{2}-1\Big) f(x)<0
\]
and the lemma follows.
\end{proof}

By the previous lemma it remains to study the limits $h\to \infty$ and $h\to0$, which is done in the following two lemmata.

\begin{lemma}\label{lem:large-h-small-x}
For any $e_1,e_2\in\R$ we have an asymptotic solution curve in the limit $h\to\infty$ parameterized by
\[
    (0,\varepsilon)\to(0,\infty)\times(0,\infty),~x\mapsto\big(x,\sqrt{\tfrac{90}{29\,x}}\big)
\]
for some $\varepsilon>0$.
\end{lemma}
\begin{proof}
In the limit of large $h$ and small $x$ the terms
\[
\frac{29}{30}h^2+f(x)
\]
are dominant. Studying when they mutually compensate (i.e.\ setting them to zero) and solving for $x$ yields the lemma. Note that approximating $f(x)\approx-\frac{3}{x}$ is sufficient for this claim.
\end{proof}
\begin{remark}\label{rem:approximation_order_small_x_large_h}
\begin{itemize}
    \item[\textup{(i)}] Using the results from Section~\textup{\ref{section-Possible-non-h-solvable-points}} we can make $\varepsilon$ small enough such that the graphs of $X_{(e_1,\pm)}$ are bounded away from our solution curve. Hence, indeed the solution curve approximated by our asymptotic solution curve in the previous lemma is solvable to be the graph of an analytic function of $x$ on $(0,\varepsilon)$.
    \item[\textup{(ii)}] Note that, without taking into account also terms of lower order than the one considered in Lemma~\textup{\ref{lem:large-h-small-x}}, we can only conclude that the solution curve is in $\sqrt{\nicefrac{90}{29x}\,}+\smalloh(x^{-\nicefrac{1}{2}})$. However, this suffices to conclude from $\sqrt{\nicefrac{90}{29}}>1$ that the solution curve in consideration provides solutions at positive $\xi=\frac{1}{12}\big(x-\frac{1}{h^2}\big)$, at least at sufficiently small $x$ or sufficiently large $h$, respectively.
\end{itemize}
\end{remark}

\begin{lemma}\label{lem:small_h_small_x_regime}
If $e_1<0$ the map
\[
    (0,\varepsilon)\to(0,\infty)\times(0,\infty),~x\mapsto\big(x,\sqrt{-\tfrac{e_1}{90}\,x~}\,\big)
\]
is an asymptotic solution curve at small $x$ with some sufficiently small $\varepsilon>0$. If $e_1\ge0$ there exists $\varepsilon>0$ such that $(0,\varepsilon)\times(0,\varepsilon)$ contains no solution of $F(x,h)=0$.
\end{lemma}
\begin{proof}
The relevant terms in the present setting are
\begin{equation}\label{eq:dominant_terms_small_x_small_h}
-\frac{e_1}{30}\frac{1}{h^2}+2\log(h)+f(x).
\end{equation}

Let $e_1<0$. Equating the first and last term of \eqref{eq:dominant_terms_small_x_small_h} (and approximating $f(x)\approx-\frac{3}{x}$, cf.\ Appendix~\ref{appendix-function-f}) yields the asymptotic solution curve in the lemma. In this scenario the $\log$-term only contributes in higher orders.

Note that the graphs of $X_{(e_1,\pm)}$ cannot come amiss to the present consideration and we indeed obtain an analytic solution curve approximated by the asymptotic solution curve in the lemma.

If $e_1\ge0$, all the dominant terms diverge to $-\infty$ as $h\to0$ or $x\to 0$ and since the remaining contributions into $F$ that are left out in \eqref{eq:dominant_terms_small_x_small_h} are bounded in this limit we can find a negative upper bound on $F$ on some $(0,\varepsilon)\times(0,\varepsilon)$, $\varepsilon>0$.
\end{proof}

\subsection{Asymptotics at finite \headingmath $x$}
\label{Section-Asymptotics-at-finite-x}
At first we study the regime of small $h$.
\begin{lemma}
\begin{itemize}
    \item[\textup{(i)}] Let $e_1\neq0$. For any compact interval $[a,b]\subset(0,\infty)$ there exists $\varepsilon>0$ such that $F$ is bounded away from 0 on $[a,b]\times(0,\varepsilon)$. Consequently, the equation $F(x,h)=0$ has no solution in $[a,b]\times(0,\varepsilon)$.
    \item[\textup{(ii)}] Let $e_1=0$. For any compact interval $[a,b]\subset(0,\infty)$ with $2\notin[a,b]$ there exists $\varepsilon>0$ such that $F$ is bounded away from 0 on $[a,b]\times(0,\varepsilon)$. Consequently, the equation $F(x,h)=0$ has no solution in $[a,b]\times(0,\varepsilon)$.
\end{itemize}
\end{lemma}
\begin{proof}
The proof of (i) goes very similar to the proof of Lemma~\ref{lem:small_x_finite_h_regime} and we skip the details. Note that the term $-\frac{e_1}{30}\,\frac{1}{h^2}$ in $F_1$ (cf.\ \eqref{energy-equation-massive-case}) is dominant in this regime and, in particular, suppresses the influence of the $\log$-term in $F_2$.

If, on the other hand, $e_1=0$, the dominant term in the limit $h\to0$ is the $\log$-term in $F_2$. This term does come with the prefactor $-\big(\frac{x}{2}-1\big)$, that is, with a prefactor which changes sign at $x=2$. If we, however, stay away from $x=2$ and assume that $2\notin[a,b]$, then all terms other than the $\log$-term are continuable to $h=0$ on $[a,b]$. By the same argument as for Lemma~\ref{lem:small_x_finite_h_regime} (or for Part (i)) we conclude (ii).
\end{proof}
\begin{remark}
Using Lemma~\textup{\ref{lem:small_h_small_x_regime}} we can formulate Part \textup{(i)} of the previous lemma $($i.e.\ for the case $e_1>0)$ also for intervals of the form $(0,b]$ with $b<2$.
\end{remark}
Next we study the limit $(x,h)\to(2,0)$ for $e_1=0$ which leads to cases that are sensitive to values of $e_2$.

\begin{lemma}
Let $e_1=0$.
\begin{itemize}
    \item[\textup{(i)}] If $e_2>0$, there is precisely one solution curve in a sufficiently small neighborhood of $(2,0)\in\overline{(0,\infty)\times(0,\infty)}$. This solution curve is of the form
    \[
    (2,2+\varepsilon)\to(0,\infty)\times(0,\infty),~x\mapsto\big(x,h(x)\big)\quad\textup{and fulfills}\quad h(x)\in\smalloh(x-2)\,,
    \]
    for some $\varepsilon>0$.
    \item[\textup{(ii)}] If $e_2<0$, there is precisely one solution curve in a sufficiently small neighborhood of $(2,0)\in\overline{(0,\infty)\times(0,\infty)}$. This solution curve is of the form
    \[
    (2-\varepsilon,2)\to(0,\infty)\times(0,\infty),~x\mapsto\big(x,h(x)\big)\quad\textup{and fulfills}\quad h(x)\in\smalloh(2-x)\,,
    \]
    for some $\varepsilon>0$.
    \item[\textup{(iii)}] If $e_2=0$, there is precisely one solution curve in a sufficiently small neighborhood of $(2,0)\in\overline{(0,\infty)\times(0,\infty)}$. This solution curve is of the form
    \[
    (0,\varepsilon)\to(0,\infty)\times(0,\infty),~h\mapsto\big(x(h),h\big)\quad\textup{and fulfills}\quad x(h)=2+\mathcal{O}(h^2)\,,
    \]
    for some $\varepsilon>0$. Moreover, $x(h)>2$ for all $h\in(0,\delta)$.
\end{itemize}
\end{lemma}
\begin{proof}
We use polar coordinates around $(x,h)=(2,0)$ and evaluate
\[
F(2+\varrho\sin\psi,\varrho\cos\psi)\qquad\textup{with}\quad\psi\in\big(-\tfrac{\pi}{2},\tfrac{\pi}{2}\big)\quad\textup{and}\quad\varrho\in(0,2)\,.
\]
The dominant terms in the limit $\varrho\to0$ stem from the $\log$-term in $F_2$ as well as from the $e_2$-term which does not depend on $\varrho$ or $\psi$. More precisely, evaluating the lengthy expression for $F(2+\varrho\sin\psi,\varrho\cos\psi)$ we find that
\[
    F(2+\varrho\sin\psi,\varrho\cos\psi)+\varrho\sin\psi\log(\varrho\cos\psi)+\frac{e_2}{30}~\to~0\qquad\textup{as }\varrho\to0\,,
\]
uniformly in $\psi\in\big(-\frac{\pi}{2},\frac{\pi}{2}\big)$ and at least of order $\mathcal{O}(\varrho)$. Consequently, for sufficiently small $\varrho$ the solutions of $F(2+\varrho\sin\psi,\varrho\cos\psi)=0$ are approximated by solutions of $\varrho\sin\psi\log(\varrho\cos\psi)+\frac{e_2}{30}=0$.

At fixed $\varrho\in(0,1)$ the map $\psi\mapsto\varrho\sin\psi\log(\varrho\cos\psi)$ is monotonously decreasing on $\big(-\tfrac{\pi}{2},\tfrac{\pi}{2}\big)$ with $\varrho\sin\psi\log(\varrho\cos\psi)\to\mp\infty$ as $\psi\to\pm\frac{\pi}{2}$ and a zero at $\psi=0$. Since now
\[
\varrho\sin\psi\log(\varrho\cos\psi)+\frac{e_2}{30}~\to~\frac{e_2}{30}\qquad\textup{as }\varrho\to0
\]
pointwise in $\psi\in\big(-\frac{\pi}{2},\frac{\pi}{2}\big)$ and uniformly in $\psi\in\big(-\frac{\pi}{2}+\delta,\frac{\pi}{2}-\delta\big)$ for any $\delta\in(0,\frac{\pi}{2})$, we conclude that at sufficiently small $\varrho$ we have precisely one solution $\psi$, and if $e_2>0$ we have $\psi\to\frac{\pi}{2}$, whereas if $e_2<0$ we have $\psi\to-\frac{\pi}{2}$. If $e_2=0$ we have $\psi\to0$.

Back in the coordinates $(x,h)$, for $e_2>0$ we obtain a solution curve as stated in the lemma, part (i), where $\psi\to\frac{\pi}{2}$ implies that $h\to0$ faster than linear. Analogously we obtain the assertion (ii). Note that the graph of $h\mapsto\big(X_{(0,-)}(h),h\big)$ is parameterized in the $(\varrho,\psi)$-coordinates by a curve with $\psi\to0$ and $\psi<0$ as $\varepsilon\to0$. In particular, the set where $F(x,h)=0$ is not solvable for $h$ is bounded away from our solution curves, thus they are indeed representable as the graph of a function of $x$.

Finally, for the assertion (iii) note that we have $\psi\to0$ from above as $\varrho\to0$, in particular, the graph of $h\mapsto\big(X_{(0,-)}(h),h\big)$ is also bounded away from the solution curve for $e_2=0$. Moreover, since the map $\psi\mapsto\varrho\sin\psi\log(\varrho\cos\psi)$ has a linear zero in $\psi=0$ for sufficiently small $\varrho$, that is, we have
\[
\varrho\,\sin\psi\,\log(\varrho\cos\psi)=\varrho\,\psi\,\log\varrho+\mathcal{O}(\psi^3)\,,
\]
in the limit $\psi\to0$ as $\varrho\to0$ is also achieved at least linearly. Transformation back into the coordinates $(x,h)$ adds another power and we obtain assertion (iii).
\end{proof}

Next we study the situation at large $h$.

\begin{lemma}
Let $e_1,e_2\in\R$ and suppose that Assumption \textup{\ref{ass:second}} holds. The straight lines defined by $x\in\{x_{(\pm)}\}$ are asymptotic solution curves. Moreover, the solution curves approximated by these asymptotic solution curves are representable as graphs of analytic functions
\[
(x_{(\pm)},x_{(\pm)}+\varepsilon)\to(0,\infty)\times(0,\infty),~x\mapsto\big(x,h(x)\big)
\]
for some $\varepsilon>0$.
\end{lemma}
\begin{proof}
We evaluate the function $F$ at $x_{(\pm)}$ and at $x_{(\pm)}+2\delta$ for some $\delta\in(0,\nicefrac{1}{\sqrt{30}})$ and obtain
\begin{equation}
    F(x_{(\pm)},h)=\mp\frac{1}{\sqrt{30}}-\frac{e_2}{30}-\frac{e_1}{30}\,\frac{1}{h^2}\mp\frac{1}{\sqrt{30}}\big(2\log(h)+f(x_{(\pm)})\big)\label{eq:energyequation-at-xpm}
\end{equation}
as well as
\begin{align}
    F(&x_{(\pm)}+2\delta,h)\label{eq:energyequation-at-xpm-plus-eps}\\
    &=\Big(\delta\pm\sqrt{\frac{2}{15}}\,\Big)\delta h^2\mp\frac{1}{\sqrt{30}}-\delta-\frac{e_2}{30}-\frac{e_1}{30}\,\frac{1}{h^2}-\Big(\delta\pm\frac{1}{\sqrt{30}}\Big)\big(2\log(h)+f(x_{(\pm)}+2\delta)\big)\,,\notag
\end{align}
respectively.

Note that, for our given bound on $\delta$ both the prefactors $\delta\pm\sqrt{\nicefrac{2}{15}}$ and $\delta\pm\frac{1}{\sqrt{30}}$ of the divergent terms in the limit $h\to\infty$ in \eqref{eq:energyequation-at-xpm-plus-eps} do not change sign in $(0,\nicefrac{1}{\sqrt{30}})\ni\delta$.

From \eqref{eq:energyequation-at-xpm} and \eqref{eq:energyequation-at-xpm-plus-eps} we can read off that
\[
F(x_{(\pm)},h)\to\mp\infty\qquad\textup{as well as}\qquad F(x_{(\pm)}+2\delta,h)\to\pm\infty
\]
as $h\to\infty$ and, consequently, there exists $h_0\in(0,\infty)$ such that for all $h>h_0$ there exists a solution of $F(x,h)=0$ in both the intervals
\begin{equation}
    (x_{(-)},x_{(-)}+2\delta) \qquad\textup{and}\qquad (x_{(+)},x_{(+)}+2\delta)\,.\label{eq:open_intervals_in_some_proof}
\end{equation}
By possibly enlarging $h_0$ we find that there is a third solution in some interval $(0,\widetilde{\delta})$ provided by Lemma~\ref{lem:large-h-small-x} and with Lemma~\ref{lem:solution_count_lemma_h_fixed} (which is where we need Assumption \ref{ass:second}) it is indeed \emph{the} third solution. Hence, each of the intervals \eqref{eq:open_intervals_in_some_proof} contains precisely one solution.

Since the above argument works for any (sufficiently small) $\delta>0$, and since the solutions we have found are in the open intervals in \eqref{eq:open_intervals_in_some_proof} (in particular they do not coincide with these interval's lower bounds), we can already conclude that the solutions ``come closer to $x_{(\pm)}$'' as $h\to\infty$, more precisely, we have the asymptotic solution curves as stated in the lemma, approximating the solution curves in $\smalloh(1)$. This approximation order suffices, since, by our above argumentation (relying on Assumption \ref{ass:second}), we have already identified all solution curves at sufficiently large $h$-values and there are no more solution curves from which the ones we have found need to be separated by elaborating higher-order approximations.

Finally, evaluating $F\big(X_{(e_1,\pm)}(h),h\big)$ (with $X_{(e_1,\pm)}(h)$ from Lemma~\ref{lem:def_of_Xpm}) we find the dominant term in the limit $h\to\infty$ to be the $\log$-term in $F_2$. In particular, we find that also $F\big(X_{(e_1,\pm)}(h),h\big)\to\mp\infty$ as $h\to\infty$, and we can replace \eqref{eq:open_intervals_in_some_proof} by the refined intervals
\begin{equation*}
    (X_{(e_1,-)}(h),x_{(-)}+2\delta) \qquad\textup{and}\qquad (X_{(e_1,+)}(h),x_{(+)}+2\delta)
\end{equation*}
at sufficiently large $h$. Again the solutions do not coincide with these refined interval's lower bounds. Hence, in every solution point found above we can solve the solution set to yield the graph of an analytic function of $x$ on some interval $(x_{(\pm)},x_{(\pm)}+\varepsilon)$.
\end{proof}
\begin{remark}
Recall that the values $x_{(\pm)}$ correspond, in the limit of large $h$, to the values $\xi_{(\pm)}$ which we have distinguished in the massless case, where we also have a divergent $(H\to\infty)$ solution branch around these values.
\end{remark}

\subsection{Minimal and conformal coupling}
\label{Section-Minimal-and-conformal-coupling}

In this section we want to study the physically distinguished $\xi$-values of minimal coupling $\xi=0$ and conformal coupling $\xi=\frac{1}{6}$. This allows us to exclude solutions or to specify a number of solutions along the curves defined by these $\xi$-values for certain choices of $e_1$ and $e_2$. For this analysis we use similar arguments as in Section~\ref{Section-Overall-solution-counting}.

Minimal and conformal coupling are realized in $(x,h)$-coordinates by the curves
\begin{align}
    (0,\infty)\to(0,\infty)\times(0,\infty),&\quad x\mapsto\big(x,h_\textup{mc}(x)\big):=\Big(x,\frac{1}{\sqrt{x}\,}\Big)\label{eq:minimal_coupling_curve}
\intertext{and}
    (2,\infty)\to(0,\infty)\times(0,\infty),&\quad x\mapsto\big(x,h_\textup{cc}(x)\big):=\Big(x,\frac{1}{\sqrt{x-2}\,}\Big)\,,\label{eq:conformal_coupling_curve}
\end{align}
respectively.

In order to study these particular cases, we substantiate Assumption \ref{ass:third} as follows. First, we treat the case of minimal coupling.
\begin{lemma}\label{lem:minimal_coupling_lemma_ass}
    There exist $x_1,x_2\in(0,\infty)$ such that the $(h$-, $e_1$- and $e_2$-independent$)$ function 
    \[
    (0,\infty)\to\mathbb{R},~x\mapsto\frac{\mathrm{d}^2}{\mathrm{d} x^2}\,F\big(x,h_\textup{mc}(x)\big)
    \]
    is negative on $(0,x_1)\cup(x_2,\infty)$.
\end{lemma}
\begin{proof}
   Evaluating the function $F$ from \eqref{energy-equation-massive-case} along the curve of minimal coupling \eqref{eq:minimal_coupling_curve} results in
\begin{equation}
    F\big(x,h_\textup{mc}(x)\big)=-\Big(\frac{1}{4}+\frac{e_1}{30}\Big)x-\Big(\frac{e_2}{30}+\frac{x}{2}\big(f(x)-\log(x)\big)\Big)+\frac{29}{30x}+f(x)-\log(x)\label{minimally-coupled-solutions-equation}
\end{equation}
for $x>0$ (where  we grouped the terms according to their relevance in the different regimes $x\to0$, finite $x$ or $x\to\infty$). For the second derivative of \eqref{minimally-coupled-solutions-equation} we compute that
\begin{align*}
        \frac{\mathrm{d}^2}{\mathrm{d} x^2}\, F\big(x,h_\textup{mc}(x)\big)\,&=-\frac{61}{15x^3}+\mathcal{O}\Big(\frac{1}{x^2}\Big)\quad\textup{as }x\to0
\intertext{as well as}
        \frac{\mathrm{d}^2}{\mathrm{d} x^2}\, F\big(x,h_\textup{mc}(x)\big)\,&=-\frac{74}{21x^4}+\mathcal{O}\Big(\frac{1}{x^5}\Big)\quad\textup{as }x\to\infty\,,
\end{align*}
for which we have employed the Puiseux expansion of $f-\log$ from Appendix~\ref{appendix-function-f}. 
\end{proof}
\begin{proof}[Numerical evidence for Assumption \textup{\ref{ass:third}.(i)}]
    The previous lemma imposes the assumption at sufficiently small and at sufficiently large $x$. That it holds also in the intermediate regime is numerically justified in Figure~\ref{Second-der-min-conf-coupling}.(i), were we have plotted \eqref{minimally-coupled-solutions-equation} for $x\in(0,\infty)$ (multiplied by $x^{\nicefrac{7}{2}}$) in log-log-scaling, together with its asymptotics (also multiplied by $x^{\nicefrac{7}{2}}$) of the previous lemma.
    \renewcommand\qedsymbol{\rotatebox{45}{$\square$}}
\end{proof}

\begin{figure}
\centering
    \includegraphics{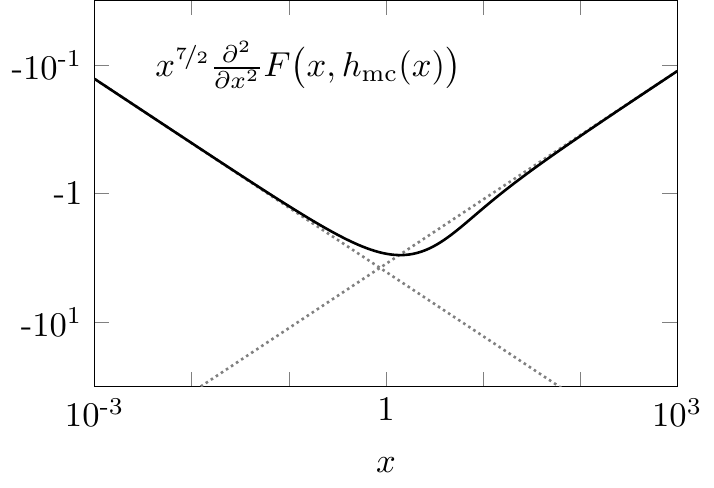}
    \includegraphics{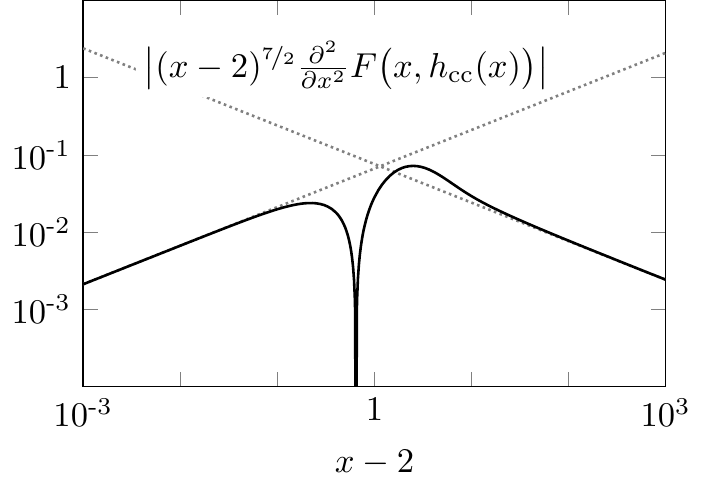}

        \hspace{.5cm}
        (i) minimal coupling
        \hspace{3.5cm}
        (ii) conformal coupling
    
    \begin{minipage}{.9\textwidth}
        \captionsetup{format=plain, labelfont=bf}
        \caption{\label{Second-der-min-conf-coupling}The graphics show the (in (ii), absolute value of the) second derivatives of $F$ along the curves of minimal (i) and conformal (ii) coupling, parameterized by $x$ and $x-2$, respectively. The dotted lines mark the asymptotics asserted in the text both at small and large arguments.  Here we multiplied in both plots with a power-$x^{\nicefrac{7}{2}}$-factor in order to maximize the angle between the asserted asymptotes in the logarithmic plot, that is, the lines which have a logarithmic slope of $\frac{1}{2}$ and $-\frac{1}{2}$.}
    \end{minipage}
\end{figure}

Before using this numerical evidence we substantiate Part (ii) of Assumption \ref{ass:third}.

\begin{lemma}\label{lem:conformal_coupling_lemma_ass}
    There exist $x_1,x_2\in(2,\infty)$ such that the $(h$-, $e_1$- and $e_2$-independent$)$ function 
    \[
    (2,\infty)\to\mathbb{R},~x\mapsto\frac{\mathrm{d}^2}{\mathrm{d} x^2}\,F\big(x,h_\textup{cc}(x)\big)
    \]
    is negative on $(2,x_1)$ and positive on $(x_2,\infty)$.
\end{lemma}
\begin{proof}
    For $x>2$ and $y=x-2$ we obtain
    \begin{align}
        F\big(x,h_\textup{cc}(x)\big)&=F\Big(y+2,\frac{1}{\sqrt{y}}\Big)\notag\\
        &=-\Big(\frac{1}{4}+\frac{e_1}{30}\Big)y-\frac{e_2}{30}-\frac{y}{2}\big(f(y+2)-\log(y)\big)-\frac{1}{30y}\label{eq:conformal_coupling_equation}
    \end{align}
    for the values of $F$ along the curve of conformal coupling \eqref{eq:conformal_coupling_curve}. Note that, employing some Puiseux series arithmetics\footnote{After we expanded $f(y+2)$ we, moreover, need to expand $\frac{1}{(y+2)^k}$ as well as $\log(\frac{y}{y+2})$ in terms of $\frac{1}{y^{l}}$.}, one can compute that
    \[
        f(y+2)=\log(y)+\frac{2}{3y}-\frac{1}{15y^2}-\frac{8}{315y^3}+\mathcal{O}\big(\tfrac{1}{y^4}\big)
    \]
    at large $y$ and thus, for the second derivative of \eqref{eq:conformal_coupling_equation}, obtain the asymptotic expansions
    \begin{align*}
        \frac{\mathrm{d}^2}{\mathrm{d} y^2}\,F\big(y+2,\tfrac{1}{\sqrt{y}}\big)\,&=-\frac{1}{15y^3}+\mathcal{O}\Big(\frac{1}{y^2}\Big)\quad\textup{as }y\to0
    \intertext{and}
        \frac{\mathrm{d}^2}{\mathrm{d} y^2}\,F\big(y+2,\tfrac{1}{\sqrt{y}}\big)\,&=\frac{8}{105y^4}+\mathcal{O}\Big(\frac{1}{y^5}\Big)\quad\textup{as }y\to\infty\,.
    \end{align*}
    Therefrom we read off the lemma.
\end{proof}
\begin{proof}[Numerical evidence for Assumption \textup{\ref{ass:third}.(ii)}]
    Numerically evaluating the second derivative of \eqref{eq:conformal_coupling_equation} we obtain Figure~\ref{Second-der-min-conf-coupling}.(ii). Therefrom, we conclude that in the intermediate regime not covered by Lemma \ref{lem:conformal_coupling_lemma_ass} this second derivative has precisely one zero (indicated by the dip in the log-log plot). Note that we could continue to numerically evaluate $\frac{\mathrm{d}^2}{\mathrm{d} x^2}\,F\big(x,h_\textup{cc}(x)\big)$ on the intervals $(2,x_0)$ and $(x_0,\infty)$, with its numerically determined zero $x_0$ ($\approx2.6457$) and with a suitable scaling of the horizontal axis. This would, however, show that the dip in the log-log plot \ref{Second-der-min-conf-coupling}.(ii) at $x_0-2$ is in fact a linear zero and that $\frac{\mathrm{d}^2}{\mathrm{d} x^2}\,F\big(x,h_\textup{cc}(x)\big)$ is negative  $(2,x_0)$ and positive $(x_0,\infty)$. We skip this step for a concise presentation.
    \renewcommand\qedsymbol{\rotatebox{45}{$\square$}}
\end{proof}
The following two lemmata use the previous numerical evidence (in terms of Assumption \ref{ass:third}) in order to find bounds on the number of solutions in the cases of minimal and conformal coupling.

\begin{lemma}\label{lem:minimal_coupling_lemma}
Suppose that Assumption \textup{\ref{ass:third}.(i)} holds.
\begin{itemize}
    \item[\textup{(i)}] For $e_1<-\frac{15}{2}$ the consistency equation \eqref{energy-equation-massive-case} has precisely one solution with minimal coupling $\xi=0$ for any value of $e_2$.
    \item[\textup{(ii)}] For $e_1=-\frac{15}{2}$ the consistency equation has precisely one solution with minimal coupling $\xi=0$ for $e_2<20$ and no solution with minimal coupling for $e_2\ge20$.
    \item[\textup{(iii)}] For $e_1>-\frac{15}{2}$ there exists $e_2^{(0)}\in\R$ such that the consistency equation has
    \begin{itemize}
        \item[$\circ$] precisely two solutions along minimal coupling if $e_2<e_2^{(0)}$,
        \item[$\circ$] precisely one solution along minimal coupling if $e_2=e_2^{(0)}$ and
        \item[$\circ$] no solution along minimal coupling if $e_2>e_2^{(0)}$.
    \end{itemize}
\end{itemize}
\end{lemma}
\begin{proof}
By Assumption \textup{\ref{ass:third}.(i)}, the map $x\mapsto F\big(x,h_\textup{mc}(x)\big)$ is strictly concave and hence for all parameter choices has at most two zeros.

Now compute the asymptotic expansions of \eqref{minimally-coupled-solutions-equation},
\begin{align}
    F\big(x,h_\textup{mc}(x)\big)\,&=-\frac{61}{30x}+\mathcal{O}(1)&&\textup{as }x\to0
    \notag
\intertext{and}
    F\big(x,h_\textup{mc}(x)\big)\,&=-\Big(\frac{1}{4}+\frac{e_1}{30}\Big)x-\Big(\frac{e_2}{30}-\frac{2}{3}\Big)+\mathcal{O}\Big(\frac{1}{x}\Big)\hspace{-1.5cm}&&\textup{as }x\to\infty\,.
    \label{eq:mc_coupling_asymptotic_expansion}
\end{align}
We read off that for $e_1<-\frac{15}{2}$ the limits $F\big(x,h_\textup{mc}(x)\big)\to-\infty$ as $x\to0$ and $F\big(x,h_\textup{mc}(x)\big)$ $\to+\infty$ as $x\to\infty$ hold. A concave function admitting these limits has precisely one zero and we conclude Assertion (i) of the lemma.

If $e_1=-\frac{15}{2}$ we have the same small-$x$-limit as before, but now we have $F\big(x,h_\textup{mc}(x)\big)\to-\frac{e_2}{30}+\frac{2}{3}$ as $x\to\infty$. That limit is positive if and only if $e_2<20$. If this is the case, a concave function with such limiting behavior has precisely one zero, whereas if this is not the case, such a function cannot have a zero at all. This shows (ii).

Finally, for fixed $e_1>-\frac{15}{2}$ we read off the limits $F\big(x,h_\textup{mc}(x)\big)\to-\infty$ for both $x\to0$ and $x\to\infty$, and using its concavity the map $x\mapsto F\big(x,h_\textup{mc}(x)\big)$ has precisely one local (and thus global) maximum. The value of $F\big(x,h_\textup{mc}(x)\big)$ at this maximum depends on $e_2$ in an affine linear manner (with non-vanishing slope coefficient, namely $-\frac{1}{30}$), in particular, there exist exactly one $e_2^{(0)}$ such that maximum value equals $0$. If $e_2=e_2^{(0)}$, the strictly concave function vanishes in its global maximum and hence that maximum is the only zero. If $e_2<e_2^{(0)}$, the maximum value is positive, and a concave function with the above limiting behavior and a positive maximum has exactly two zeros. If, on the other hand, $e_2>e_2^{(0)}$, the maximum value is negative and a function with negative global maximum has no zero at all. Together we conclude (iii) of the lemma.
\end{proof}

\begin{lemma}\label{lem:conformal_coupling_lemma}
Suppose that Assumption \textup{\ref{ass:third}.(ii)} holds.
\begin{itemize}
    \item[\textup{(i)}] If $e_1<-\frac{15}{2}$ the consistency equation \eqref{energy-equation-massive-case} has at most three solutions with conformal coupling $\xi=\frac{1}{6}$.
    \item[\textup{(ii)}] Let $e_1=-\frac{15}{2}$. If $e_2\le-10$ the consistency equation has exactly one solution with conformal coupling. Moreover, there exists $e_2^{(1)}>-10$ such that the consistency equation has
    \begin{itemize}
        \item[$\circ$] precisely two solutions with conformal coupling if $e_2\in(-10,e_2^{(1)})$,
        \item[$\circ$] precisely one solution with conformal coupling if $e_2=e_2^{(1)}$ and
        \item[$\circ$] no solution with conformal coupling if $e_2>e_2^{(1)}$.
    \end{itemize}
    \item[\textup{(iii)}] If $e_1>-\frac{15}{2}$ there exists $e_2^{(2)}\in\R$ such that the consistency equation has
    \begin{itemize}
        \item[$\circ$] precisely two solutions with conformal coupling if $e_2<e_2^{(2)}$,
        \item[$\circ$] precisely one solution with conformal coupling if $e_2=e_2^{(2)}$ and
        \item[$\circ$] no solution with conformal coupling if $e_2>e_2^{(2)}$.
    \end{itemize}
\end{itemize}
\end{lemma}
\begin{proof}
We proceed similar to Lemma~\ref{lem:minimal_coupling_lemma}, although we need to adjust the arguments at some points.

Assumption \ref{ass:third}.(ii) and the argument introduced in the beginning of Section \ref{Section-Overall-solution-counting} imply that \eqref{eq:conformal_coupling_equation} has at most three solutions, regardless of any parameter choices. Particularly, Part (i) of the lemma follows and we skip to improve the bounds in that setting.

For Parts (ii) and (iii), we determine the asymptotic expansions of \eqref{eq:conformal_coupling_equation} and obtain
    \begin{align*}
        F\big(x,h_\textup{cc}(x)\big)\,&=-\frac{1}{30y}+\mathcal{O}(1)&&\textup{as }y\to0
    \intertext{and}
        F\big(x,h_\textup{cc}(x)\big)\,&=-\Big(\frac{1}{4}+\frac{e_1}{30}\Big)y-\Big(\frac{e_2}{30}+\frac{1}{3}\Big)+\mathcal{O}\Big(\frac{1}{y}\Big)\hspace{-1.5cm}&&\textup{as }y\to\infty\,.
    \end{align*}

If $e_1=-\frac{15}{2}$, we read off that $F\big(y+2,\tfrac{1}{\sqrt{y}}\big)\to-\infty$ as $y\to0$ and $F\big(y+2,\tfrac{1}{\sqrt{y}}\big)\to-\frac{e_2}{30}-\frac{1}{3}$ as $y\to\infty$. Consequently, if $e_2\le-10$ the latter large-$y$-limit is non-negative and, having only one inflection point, the map $y\mapsto F\big(y+2,\tfrac{1}{\sqrt{y}}\big)$ must have precisely one zero. Therefore note that the second derivative is indeed positive at large $y$ by Lemma \ref{lem:conformal_coupling_lemma_ass}.

On the other hand, from the limiting behavior we can read off that the map $y\mapsto F\big(y+2,\tfrac{1}{\sqrt{y}}\big)$ has a global maximum, in which the maximum value depends affine linearly on $e_2$. Hence there exists $e_2^{(1)}$ such that this maximum value is zero, where in the context of the aforementioned argument necessarily $e_2^{(1)}>-10$ holds. The remaining claims of part (ii) in the lemma follow immediately.

Finally, for $e_1>-\frac{15}{2}$ we have $F\big(y+2,\tfrac{1}{\sqrt{y}}\big)\to-\infty$ in both limits $y\to0$ and $y\to\infty$, consequently the function $y\mapsto F\big(y+2,\tfrac{1}{\sqrt{y}}\big)$ has a global maximum in which its value, again, depends affine linearly on $e_2$. By the same arguments as above and in Lemma~\ref{lem:minimal_coupling_lemma} we obtain a $e_2^{(2)}\in\R$ as in Part (iii) of the lemma.
\end{proof}

\section{Numerical treatment of the solution set \label{sec:Numerics}}

In the previous sections we discussed properties of the set $\mathcal{S}_{e_1,e_2}$ of de Sitter solutions of the energy equation in the massive case, in particular, the structure of $\mathcal{S}_{e_1,e_2}$ and its asymptotic behavior. The non-polynomial terms of the energy equation prevented an explicit solution. In the present section have a look at the solution set by determining the zeros of $F$ numerically. In particular, we study at the behavior of the solution set around distinguished parameter sets that were found in the above analysis. Hereby we give attention to the topological changes of the solution set.

\begin{figure}[t!]
\begin{minipage}{.5\textwidth}
\centering
\includegraphics[scale=1]{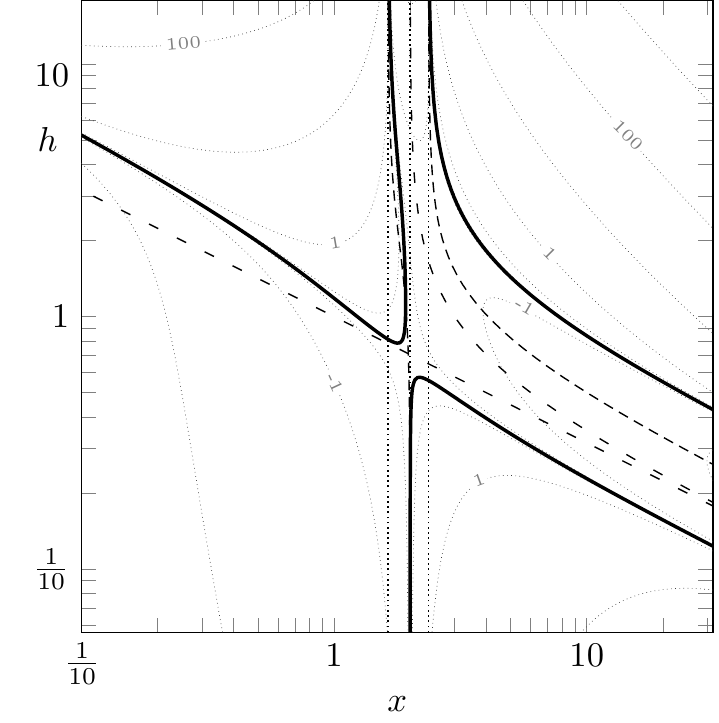}

~~(i) log-log scaling
\end{minipage}
\hfill
\begin{minipage}{.495\textwidth}
\centering
\includegraphics[scale=1]{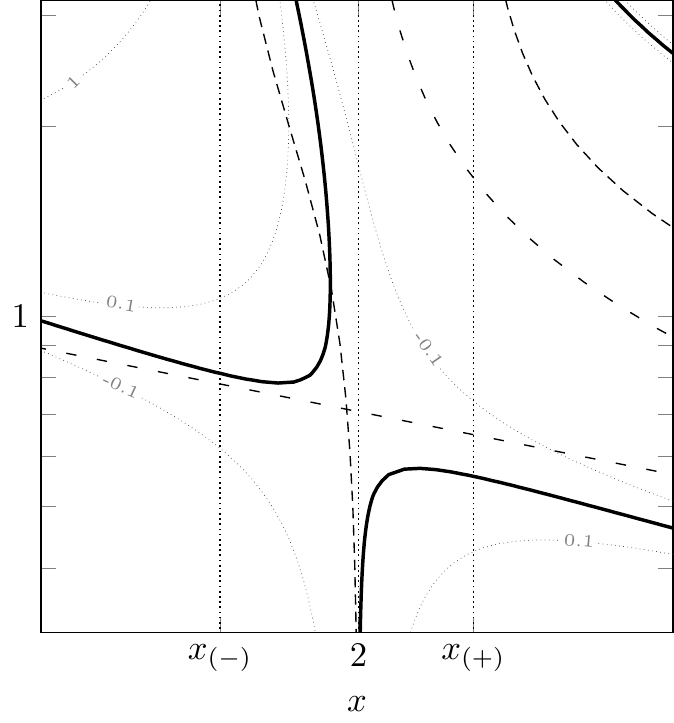}

~~(ii) log-log scaling on a smaller section
\end{minipage}

    \bigskip
\begin{minipage}{.5\textwidth}
\centering
\includegraphics[scale=1]{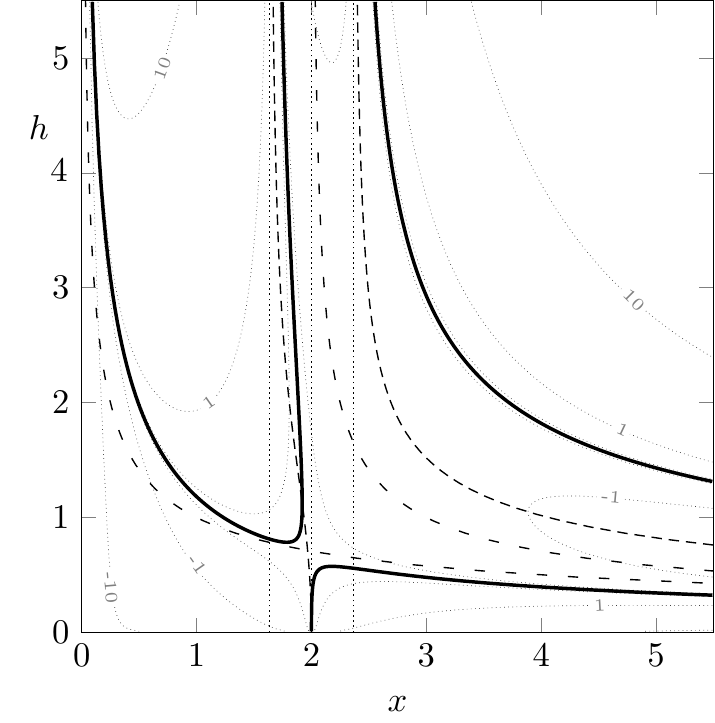}

~~(iii) normal scaling
\end{minipage}
\hfill
\begin{minipage}{.45\textwidth}
    \captionsetup{format=plain, labelfont=bf}
    \caption{\label{Figure-e1-0-e2-0-samplegraphic}~\,The graphics show the solu-}

    tion sets of the energy equation for $e_1=e_2=0$ in log-log and in normal scaling as the thick, black lines. The dotted black lines mark the values $x\in\{x_{(-)},2,x_{(+)}\}$, the densely dashed lines mark the points $(x,h)$ where $x=X_{(0,\pm)}(h)$ and the loosely dashed black lines mark minimal and conformal coupling. The dotted gray lines finally mark some more level sets of $F$ with the given values. In the following we mostly stick with the representation in log-log scaling as the straight outgoing lines show off the asymptotics of our solution curves. Note that in (i), although it seems slightly curved, the left branch of solutions also ends in a straight line with (logarithmic) slope $-\frac{1}{2}$.
\end{minipage}

\end{figure}

At first, in Figure~\ref{Figure-e1-0-e2-0-samplegraphic}, we look at the parameter pair $e_1=e_2=0$. There, and throughout this section, we use the following graphical conventions:
\begin{itemize}
    \item[$\circ$] Thick solid lines mark the solution set $\mathcal{S}_{e_1,e_2}$ of the consistency equation \eqref{energy-equation-massive-case}.
    \item[$\circ$] Densely dashed lines mark the graphs of $X_{(e_1,\pm)}$, where the equation $F(x,h)=0$ is not solvable for $h$.
    \item[$\circ$] Loosely dashed lines mark the curves from \eqref{eq:minimal_coupling_curve} and \eqref{eq:conformal_coupling_curve} of minimal ($\xi=0$) and conformal ($\xi=\frac{1}{6}$) coupling, respectively.
    \item[$\circ$]Dotted (black) lines mark the distinguished values of $x\in\{x_{(-)},2,x_{(+)}\}$.
    \item[$\circ$] If suitable, we include some additional level sets of the function $F$ in gray dotted. As $e_2$ enters the function $F$ as an offset, these are solution curves for some other value of $e_2$.
\end{itemize}

Note that we display the solution sets in log-log scaling since by our results of Section~\ref{Asymptotic-Behavior-of-the-Solution-Set} all solution curves are, in that scaling, asymptotically equivalent with straight lines, either vertical or with slope $\pm\frac{1}{2}$.

In Figure~\ref{Figure-e1-0-e2-0-samplegraphic} we identify most of the analytic assertions we have made in the previous sections. These are, we have a solution curve running into $(x,h)=(2,0)$ (cf.\ Section~\ref{Section-Asymptotics-at-finite-x}), three solution curves which are asymptotically equivalent with $x\mapsto\frac{\alpha}{\sqrt{x}}$, one of them at small $x$ (cf.\ Section~\ref{Section:small_x_asymptotics}) and two at large $x$ (cf.\ Section~\ref{section-Asymptotics-at-large-x}) and, finally, two solution curves approaching the asymptotes at $x=x_{(\pm)}$ for large $h$ (cf.\ also Section~\ref{Section-Asymptotics-at-finite-x}). Moreover, we have exactly one point where the equation is not locally solvable for $h$, precisely where the curve of solutions crosses the curve defined by the graph of $X_{(0,-)}$.

\begin{figure}[t!]
    \centering
    \begin{minipage}{.495\textwidth}

        \hspace{1.64cm}$\frac{1}{10}$\hspace{1.86cm}1\hspace{.8cm}\raisebox{2pt}{$x$}\hspace{.5cm}10

        \includegraphics[trim=0 18pt 0 0, clip]{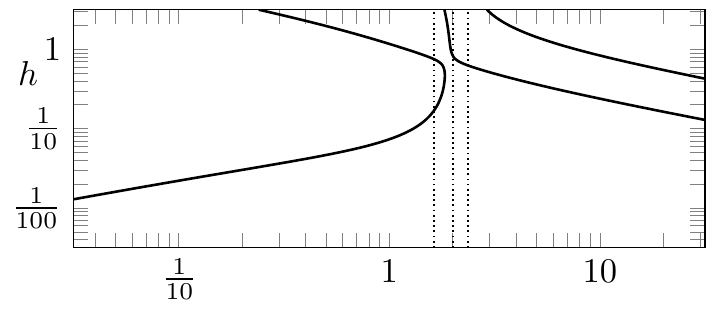}

        \hspace{2.5cm}(i) $e_1=-0.5$, $e_2=0$

        \includegraphics[trim=0 18pt 0 0, clip]{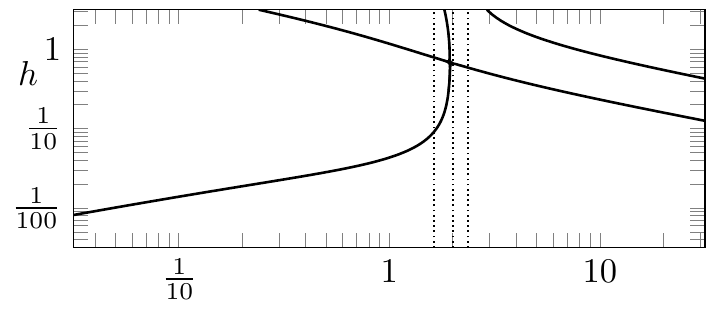}

        \hspace{2.5cm}(ii) $e_1=-0.20262$, $e_2=0$

        \includegraphics[trim=0 18pt 0 0, clip]{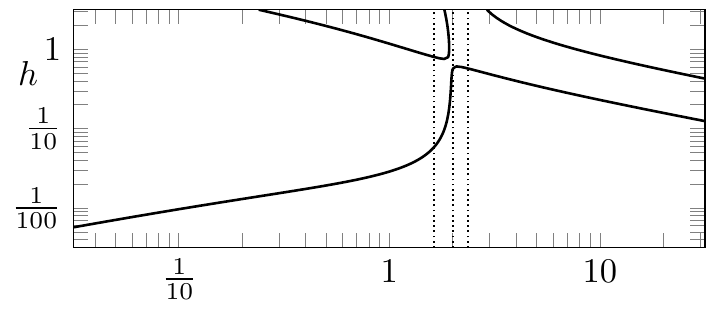}

        \hspace{2.5cm}(iii) $e_1=-0.1$, $e_2=0$

        \includegraphics[trim=0 18pt 0 0, clip]{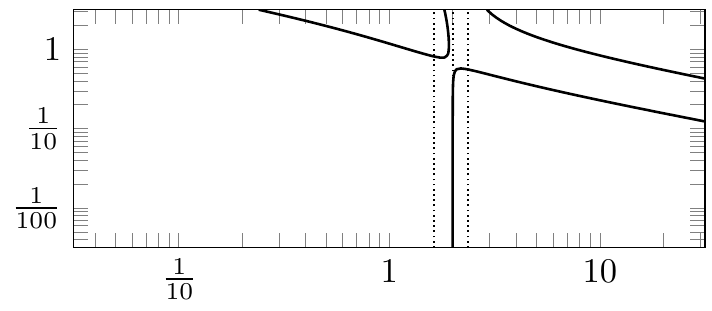}

        \hspace{2.5cm}(iv) $e_1=0$, $e_2=0$

        \includegraphics[trim=0 7pt 0 0, clip]{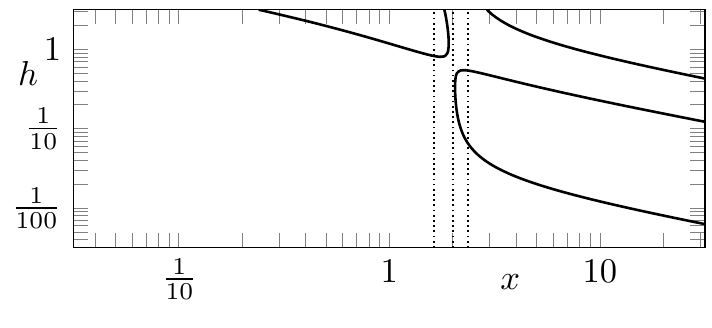}

        \hspace{2.5cm}(v)  $e_1=0.1$, $e_2=0$

    \end{minipage}\hfill
    \begin{minipage}{.495\textwidth}
    \includegraphics[trim=0 0 15pt 0, clip]{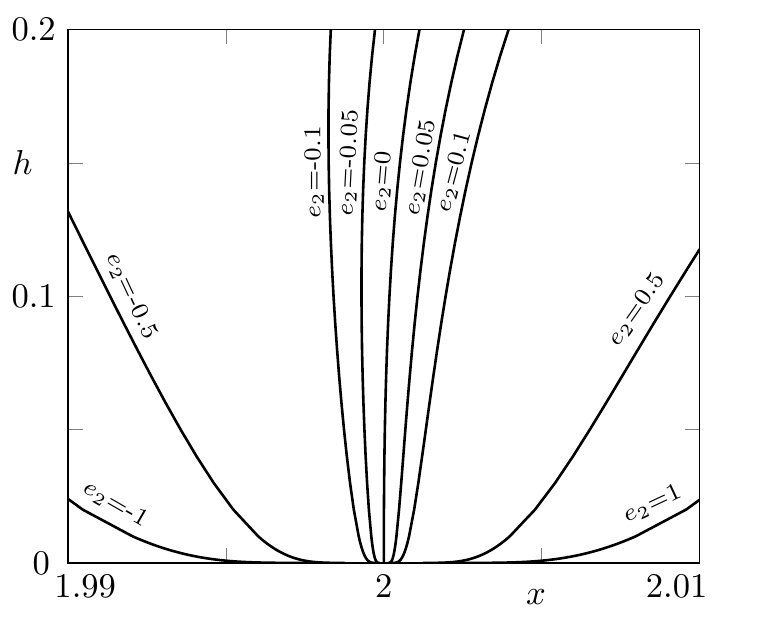}

    \vspace{-.2cm}
    \hspace{2.5cm}(vi) $e_1=0$

    \hspace{.06\textwidth}\begin{minipage}{.93\textwidth}
    \caption{\label{Figure-behavior-around-0-0} Behavior of the solution set of}

     the energy equation around the parameter point $e_1=e_2=0$. On the one hand, we show the solution sets at $e_2=0$ for several $e_1$-values around $0$ in (i)$-$(v). Above the shown section of the $h$-axis the curves are more or less identical with the plots in Figure~\ref{Figure-e1-0-e2-0-samplegraphic}, thus we cropped these figures.
     On the other hand, we show the solution sets close to the point $(x,h)=(2,0)$ for $e_1=0$ and for several $e_2$ values around 0 in (vi). This value of $e_1$ is the only case where we observe these curves running into $(2,0)$ and the distinction of how these curves run into that point is only visible if we show a non-log scaling of the $h$-axis. For other values of $e_1$ close to 0 a variation of $e_2$ is rather unspectacular.
    \end{minipage}
    \end{minipage}
\end{figure}

We have shown in Section~\ref{Section-Minimal-and-conformal-coupling} that the function $F$, restricted to the curve of minimal coupling $\xi=0$ is concave and divergent to $-\infty$ both at small and large $x$. However, as we can see in Figure~\ref{Figure-e1-0-e2-0-samplegraphic}, the maximum of the function $F$ is indeed negative and, consequently, there is no solution for the minimally coupled model. Moreover, there is also no solution for the conformally coupled case. Thereby we conclude that the six open ends of the present case constitute three connected components, namely the connected component below the minimal coupling curve, the one above the conformal coupling curve and the one in between these two curves. Moreover, we can separate the connected components by the line $x=2$, since for $e_1=e_2=0$ we have $F(2,h)=-\frac{h^2}{30}<0$ for all $h>0$. Finally, we can also observe the precise amount of solutions that we have asserted in Table~\ref{tabular-on-max-amount} on the lines at $x=x_{(-)}$ and $x=x_{(+)}$, namely exactly one solution for each $x$-value.

At next, we discuss the behavior of the solution set around the parameter point $e_1=e_2=0$, particularly we look at the asymptotic solution curve running into $(x,h)=(2,0)$. Recall that there exists such a curve only for $e_1=0$, where its precise shape is determined by the parameter $e_2$ (cf.\ Section~\ref{Section-Asymptotics-at-finite-x}). We have shown how for $e_2=0$ the solution curve running into $(x,h)=(2,0)$ is tangent to the $x=2$-line in the limit, whereas for $e_2>0$ and $e_2<0$ it is tangent to the $h=0$-axis, starting in positive or negative direction, respectively. This is what we observe in Figure~\ref{Figure-behavior-around-0-0}.(vi).

If we tune $e_1$ away from zero to values $e_1<0$, we observe that this curve, instead of running into $(x,h)=(2,0)$, now runs into $(x,h)=(0,0)$, and ends in the asymptotic solution curve at small $x$ we expect for such $e_1$ values. If we choose, however, values $e_1>0$ we observe that now this branch adds to a third asymptotic solution curve at large $x$. This behavior is shown in Figures~\ref{Figure-behavior-around-0-0}.(iii)-(v).

Moreover, in Figure~\ref{Figure-behavior-around-0-0} we capture how a saddle of the function $F$ can be tuned to coincide with a zero of $F$. While the value of $F$ at its saddle is positive or negative in Figure~\ref{Figure-behavior-around-0-0}.(i) or \ref{Figure-behavior-around-0-0}.(iii), respectively, the value of $e_1$ in Figure~\ref{Figure-behavior-around-0-0}.(ii) was chosen such that we indeed observe crossing solution curves, reducing the number of connected components of our solution set by one. That in such a case the zero set of $F$ can be locally decomposed into two analytic curves was shown in Section~\ref{Section-Non-existence-of-local-extrema}.

\begin{figure}[t!]
    \centering
    \includegraphics{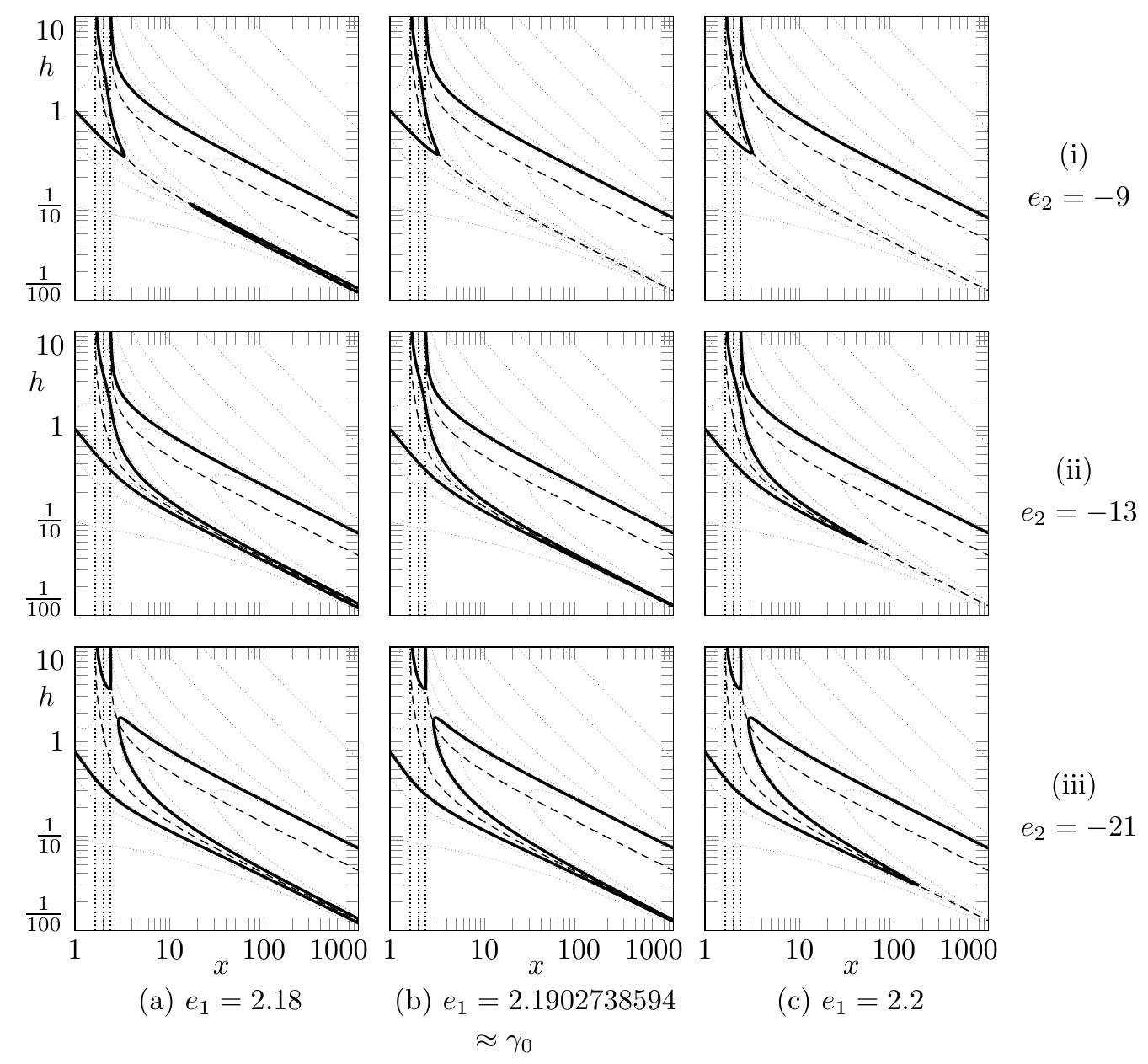}
    \begin{minipage}{.9\textwidth}
    \captionsetup{format=plain, labelfont=bf}
    \caption{The graphic shows the behavior of the solution set of the energy equation around the parameter $e_1=\gamma_0$. Hereby we also include a variation of the $e_2$-values. The values of the rows (i) and (ii) are chosen around the value $e_2=-10$ which only plays a particular role (for the asymptotics at large $x$) if $e_1=\gamma_0$. The value in row (iii) was chosen such that we additionally can observe how we passed a value with crossing curves.\label{Figure-e1-around-gamma-0}}
    \end{minipage}

\end{figure}

Furthermore, we want to look at the parameters nearby $e_1=\beta_0$ and $e_2=-10$. We have studied this setting in Section~\ref{section-Asymptotics-at-large-x}. The solution sets for certain choices of parameters are shown in Figure~\ref{Figure-e1-around-gamma-0}. At the value $e_1=\beta_0$ the solution set transitions from having three asymptotic solution curves to having only one. The lower two of these three curves for $e_1<\beta_0$ are connected to each other in Figure~\ref{Figure-e1-around-gamma-0}.(a).(i) and are connected to the two curves which are left at large $h$ without crossing in \ref{Figure-e1-around-gamma-0}.(a).(ii). In between these values there must, consequently, be an $e_2$-value at which they are connected to the same two curves, but now interchanged, and we can observe them to cross in the arising saddle of $F$. The same may be observed in each column between the $e_2$-values of (ii) and of (iii). Consequently, between the values of \ref{Figure-e1-around-gamma-0}.(c).(ii) and \ref{Figure-e1-around-gamma-0}.(c).(iii) there is a parameter point in which we only have one connected component of solutions.

Around the parameter point $e_1=\beta_{-1}$ we can observe a very similar behavior of the solution set as around $e_1=\beta_0$. We have shown in Section~\ref{section-Asymptotics-at-large-x} that above this particular value we have two asymptotic solution curves at large $x$, and that their asymptotic approximation is determined by the zeros of $s$. The latter two zeros in turn degenerate in the boundary case $e_1=\beta_{-1}$ and below this value there are no such curves anymore. Hence we can, similarly as in Figure~\ref{Figure-e1-around-gamma-0}, observe that the two asymptotic solution curves come closer to each other as $e_1\to\beta_{-1}$ from above, have the same asymptotic (first order) approximation at $e_1=\beta_{-1}$ and form a sling (now around the curve of $X_{(e_1,+)}$) which is more and more pulled to smaller values of $x$ as $e_1$ further decreases.

\section{Parameter choices for potential inflationary models}
\label{sec:Inflationary_models}

So far we mostly considered the consistency equation in the massive case as an equation of $x$ instead of the physical parameter $\xi$ in order to simplify the analysis of the function $f$. However, we can identify the curves of constant $\xi$ as parameterized by
\begin{equation}
    (0,\infty)\to(0,\infty)\times(0,\infty),~h\mapsto\big(12\xi+\tfrac{1}{h^2},h\big)\,,\label{constant-xi-parametrization}
\end{equation}
where the curves of minimal coupling \eqref{eq:minimal_coupling_curve} and conformal coupling \eqref{eq:conformal_coupling_curve} are special cases.
If $\xi<0$ these curves leave the domain $(0,\infty)\times(0,\infty)$ above a certain $h$-value, since the Bunch-Davies state exists for negative $\xi$ only if $h$ is sufficiently small.

In the present section we want to identify parameter settings of potential inflationary models. By a potential inflationary model we mean the following:

According to `standard' (i.e.\ $\Lambda$CDM) cosmology our universe will, eventually at late times, expand in a Dark Energy-dominated exponential manner. In other words, it will be well approximated by a cosmological de Sitter solution, say with Hubble rate $H_\textup{DE}$. On the other hand, assuming that the universe went through an inflationary early phase solves various problems of $\Lambda$CDM physics, most importantly, the so-called cosmic horizon problem. Also an inflationary phase is modelled by an (approximately) exponential expansion, with a much larger Hubble rate. Denote it by $H_\textup{I}$ ($\gg H_\textup{DE}$).

The aim of the present section is to identify parameter settings $(e_1,e_2,\xi)$ in which there exist multiple solutions of our model, say two values $h_1$ and $h_2$ with $h_1\gg h_2$, which restore the ratio $\frac{h_1}{h_2}=\frac{H_\textup{I}}{H_\textup{DE}}$ of some given physical data $H_\textup{I}$ and $H_\textup{DE}$. More advanced, we will show that by parameter tuning one is more or less free to produce an arbitrary ratio $\frac{h_1}{h_2}$ with a (more or less) arbitrary smaller solution $h_2$. To establish a realistic magnitude for the rates $h_1$ and $h_2$, we note that for today's Hubble rate and for the Higgs mass (i.e.\ the only mass of a scalar field occurring in the standard model of particle physics) one can compute
\[
h_2=\frac{H_\textup{today}}{M_\textup{Higgs}}\approx10^{-44}
\]
in our unit system (i.e.\ where $\hbar=c=1$). On the other hand, supposing that the inflationary phase lasted about $10^{-34}\,\textup{s}$ and caused an expansion by a factor of $10^{26}$ we can compute $H_\textup{I}\approx 10^{34}\,\frac{1}{\textup{s}}$ and thus a magnitude of
\[
    h_1=\frac{H_\textup{I}}{M_\textup{Higgs}}\approx10^7\,.
\]

Suppose that in a subsequent step, for a parameter setting as above, one is able to show how the dynamical system \eqref{dynamical-system} is unstable towards perturbations around the de Sitter solution defined by the larger value $h_1$ and stable around the solution defined by the smaller $h_2$. Moreover, suppose that one is able to show that the dynamical system \eqref{dynamical-system} (with the aforementioned stability properties) indeed possesses a solution starting close to the unstable de Sitter solution with rate $h_1$ and, after some intermediate phase, eventually approaches the stable solution with rate $h_2$. In that case one would provide a physical model for our universe in which the driving forces for both the inflationary and the late-time expansion are modelled by the presence of a quantum field. In particular, such a model omits introducing the new and unknown Dark Energy, that is, a form of energy which evades any observation other than enforcing an exponential late time expansion. In such an approach, the present section makes the first step of providing suitable parameter settings.

Note that similar to our hypothetical outline above, such effects have indeed been observed among solutions of the semiclassical Einstein equation. At first we refer to M.~H\"ansel~\cite{haensel2019} who found a parameter regime for the SCE of a massless, conformally coupled scalar field (with the conformal vacuum) in which precisely the stability properties of two (different) de Sitter solutions as described above are present. Moreover, the corresponding phase diagram (Figure 5.6.(a) of~\cite{haensel2019}) shows how solutions that start close to the larger de Sitter point approach the smaller de Sitter point under certain conditions. In another article~\cite{dappiaggi2008stable} the authors found two different de Sitter solutions of which the one with a larger rate is unstable and the one with a smaller rate is stable. They also work with a conformally coupled scalar field, using approximate KMS states. As a third reference, the authors of the present article found in~\cite{Numerik-Paper} a similar scenario in which a de Sitter solution is present and appears attractive towards perturbations. In that latter article, solutions are constructed using Minkowski-like states for a massless field. Note that while the analysis of this scenario is carried out explicitly only for $\xi=\frac{1}{6}$, it is also stated how this can be generalized for $\xi$ close but not equal to $\frac{1}{6}$. Finally, \cite{degner} identified the de Sitter-Bunch-Davies system as the asymptotic limit (in a suitable sense) of solutions to the SCE, at least in the class of so-called states of low energy.

In the following we provide a few examples of parameter settings in which there exist multiple solutions. Moreover, we show how to tune parameters in order to control the $h$-values of these solutions (particularly their ratio) by exploiting the knowledge acquired in Section~\ref{Asymptotic-Behavior-of-the-Solution-Set}. Without further notice we will make use of Assumptions \ref{ass:first}, \ref{ass:second} and \ref{ass:third} in the following.

\begin{figure}[t!]
    \begin{center}
    \includegraphics{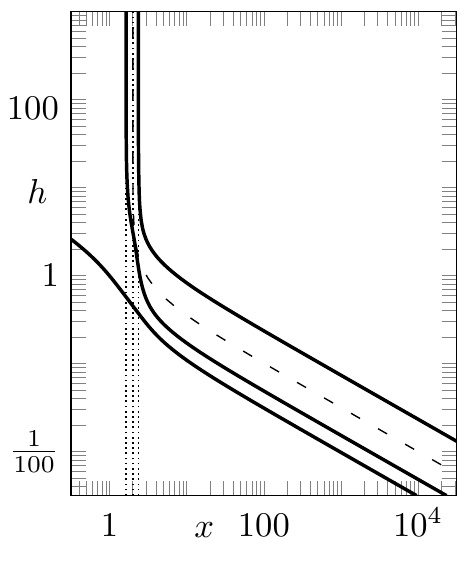}\quad
    \includegraphics{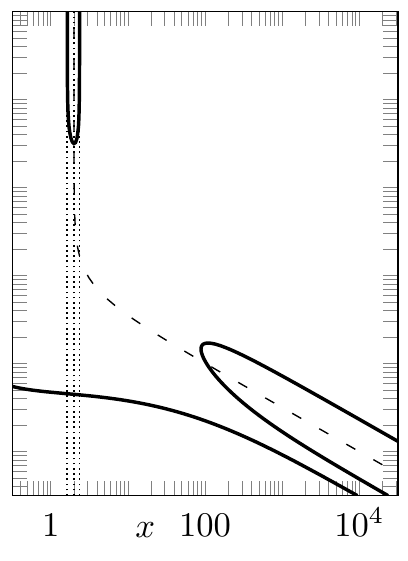}\quad
    \includegraphics{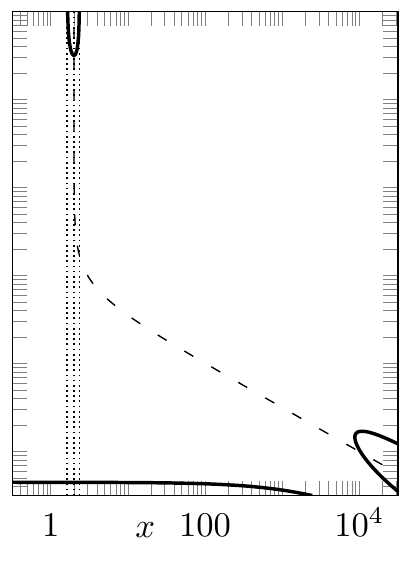}
    ~~~~~~~~~~ (i) $e_2=-10$\hspace{2.2cm}(ii) $e_2=-1000$\hspace{2.2cm}(iii) $e_2=-10^5$
    \end{center}
    \begin{minipage}{.39\textwidth}
    \centering
    \includegraphics{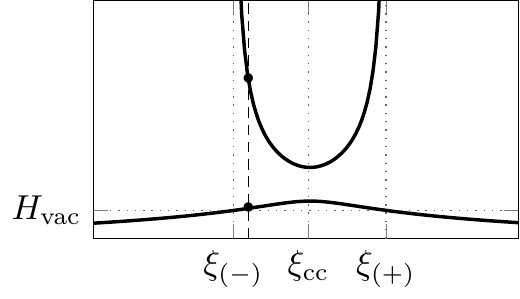}
    (iv) massless case

    \hspace{1.2cm}($x=12\xi$ for $m=0$)
    \end{minipage}~~
    \begin{minipage}{.57\textwidth}
    \captionsetup{format=plain, labelfont=bf}
    \caption{Solution sets of the massive equation for $e_1=2$ in (i)-(iii). We can see how for smaller and smaller values
    of $e_2$ the visible slings of solution curves are more and more pulled to large or small values of $h$, respectively. Hereby, both of them cross the line of $\xi=\xi_\textup{cc}=\frac{1}{6}$. (iv) shows a similar situation for the solution set of the massless equation for the same parameters as in Figure~\ref{schematic-H-of-xi}.(iii).\label{inflationary-plot}}
    \end{minipage}
\end{figure}

In Figure~\ref{inflationary-plot} we can see how the solution set behaves if $e_2\to-\infty$. We started in (i) with a parameter point $e_1=2$ and $e_2=-10$ close to the parameters studied in Figure~\ref{Figure-e1-around-gamma-0}, such that $e_1$ lies in the interval $(\beta_{\textup{-}1},\beta_0)$ where we have three asymptotic solution curves at small $h$, and therefrom lowered $e_2$.

In Section~\ref{Section-Minimal-and-conformal-coupling} we have seen that $F$, restricted to the curve of $\xi=\xi_\textup{cc}=\frac16$, parameterized as in \eqref{constant-xi-parametrization}, diverges to $-\infty$ for both $h\to0$ and $h\to\infty$ at our choice of $e_1$. Moreover, by Lemma~\ref{lem:conformal_coupling_lemma}, it has two zeros if $e_2$ lies below a certain bound and no zero above that bound. We have not computed this bound, but apparently in Figure~\ref{inflationary-plot}.(i) we are above this bound whereas in Figure~\ref{inflationary-plot}.(ii) we are below it. Since then $e_2$ enters $F$ simply as a (sign-reversed) offset, it is clear that, if we lower $e_2$, the two corresponding $h$-values of the two zeros of $F$ drift more and more apart, and the smaller solution's $h$ converges to 0, whereas the larger solution's $h$ diverges to $+\infty$ in the limit $e_2\to-\infty$.

What we have done above works not only for $\xi=\xi_\textup{cc}=\frac{1}{6}$, but for any $\xi$ with $|\xi-\frac{1}{6}|\le\nicefrac{1}{\sqrt{1080}}$, that is, for all $\xi$ such that the vertical asymptote of the corresponding curve \eqref{constant-xi-parametrization} lies in between the values $x_{(\pm)}=2\pm\sqrt{\nicefrac{2}{15}}$.

To see this, we observe that all constant-$\xi$-curves are asymptotically equivalent to $x\mapsto\nicefrac{1}{\sqrt{x}}$ at large $x$ and that the function $s$ from Section~\ref{section-Asymptotics-at-large-x} has one zero larger that 1 and two zeros smaller than 1 for our present value of $e_1$. In particular, all these curves of constant $\xi$ lie, asymptotically at large $x$, in between the upper two solution curves, and for sufficiently small $e_2$ intersect the solution set at small $h$. On the other hand, for the $\xi$-values specified above the solution set intersects all these curves at large $h$ as well. By lowering $e_2$ the $h$-values of these two solutions then drift more and more apart.

We have a similar situation for this particular interval of $\xi$-values in the massless case. We have included a graphic for this situation in Figure~\ref{inflationary-plot}.(iv). If the parameter ratio $\frac{\Lambda}{K}$ tends to zero, the two branches of the solution set drift more and more apart. More precisely, while for a fixed value $\xi\in(\xi_{(-)},\xi_{(+)})$ the larger solution diverges, the smaller solution remains close to $H_\textup{vac}=\sqrt{\nicefrac{\Lambda}{3}}$. Both these limiting behaviors can easily be read off from \eqref{H-of-xi-formula} computing $H_{(\pm)}/H_\textup{vac}$.

As a next example we want to study solutions for $\xi$-values at and around minimal coupling $\xi_\textup{mc}=0$. Figure~\ref{inflationary-plot2} shows the solution sets for the same parameters $e_1$ and $e_2$ as in Figure~\ref{inflationary-plot}. In Section~\ref{Section-Minimal-and-conformal-coupling} we have shown that also for minimal coupling we have basically the same situation as above, namely that $F$ along the minimal coupling curve diverges to $-\infty$ at both small and large $h$. Now already the value $e_2=-10$ in Figure~\ref{inflationary-plot2} lies below the upper bound on $e_2$-values from Lemma~\ref{lem:minimal_coupling_lemma}.(iii), consequently, in all three graphics of Figure~\ref{inflationary-plot2} we have two solutions with minimal coupling. The $h$-values of these solutions drift apart as $e_2\to-\infty$.

In this scenario, however, we can tune our parameter $\xi$ to positive small values such that the $h$-values of two solutions described above stay approximately stationary, but a third solution with a large $h$-value comes into play. Therefore, Figure~\ref{inflationary-plot2}.(i) includes the constant-$\xi$-curves for some such positive values. While it is obvious from the figure, we can also conclude the existence of such a solution with $h\to\infty$ as $\xi\to 0$ from Section~\ref{Section:small_x_asymptotics}. Therefore we reparameterize the (small-$x$-) asymptotic solution curve from Lemma~\ref{lem:large-h-small-x} into
\[
(0,\infty)\to(0,\infty)\times(0,\infty),~h\mapsto \big(\tfrac{90}{29\,h^2},h\big)
\]
and note that this curve lies above the line of $\xi=0$ in our $x$-$h$-parameter plane. On the other hand, the constant-$\xi$-curve for any $\xi>0$ has a vertical asymptote and, consequently, crosses the asymptotic solution curve above. For this crossing we have $x\to0$ and $h\to\infty$ as $\xi\to0$. Consequently, if $\xi$ is small enough (such that this crossing is at sufficiently small $x$-values for a good approximation of the actual solution curve by the asymptotic solution curve), we observe a third solution with the behavior as claimed above. Note that the smaller two solutions remain more or less stationary, obviously their $h$-values are continuous in $\xi$ around $\xi=0$.

\begin{figure}[t!]
    \begin{center}
    \includegraphics{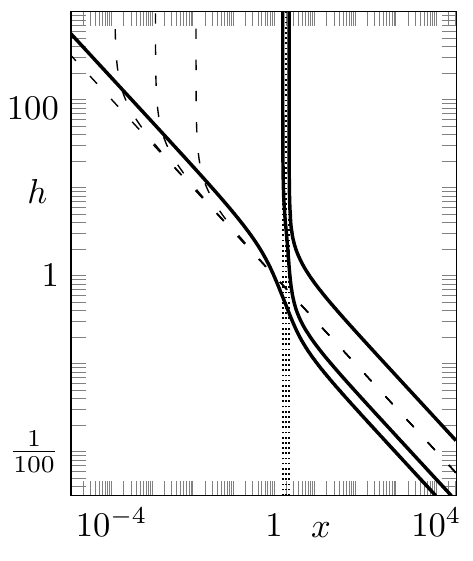}\quad
    \includegraphics{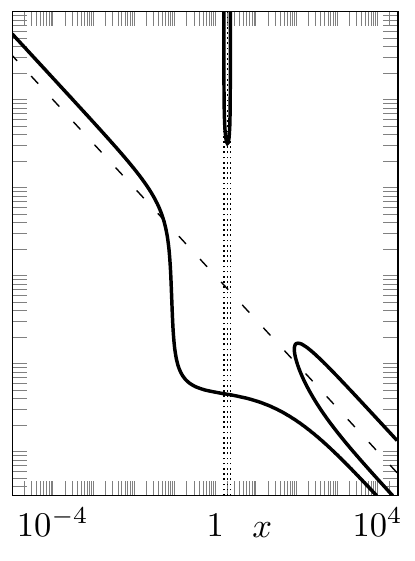}\quad
    \includegraphics{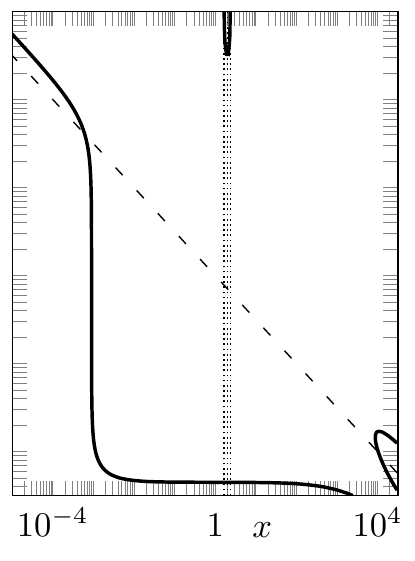}
    ~~~~~~~~~~ (i) $e_2=-10$\hspace{2.2cm}(ii) $e_2=-1000$\hspace{2.2cm}(iii) $e_2=-10^5$

    \begin{minipage}{.9\textwidth}
    \captionsetup{format=plain, labelfont=bf}
    \caption{Solution sets of the massive equation for $e_1=2$ in (i)-(iii), i.e.\ for the same parameters as in Figure~\ref{inflationary-plot}, but on a larger interval of (logarithmic) $x$-values. Moreover, in the present figure we included the curves of constant $\xi=0$, as well as, of constant $\xi\in\{10^{-3},10^{-4},10^{-5}\}$ in Subfigure (i).\label{inflationary-plot2}}
    \end{minipage}
    \end{center}
\end{figure}

Finally, we want to show the solution sets to a family of parameters with the reversed behavior as in the previous case, that is, we prescribe a certain (approximate) larger $h$-value and tune parameters so that the smaller solution's $h$-value tends to 0.

\begin{figure}
    \centering
        \includegraphics{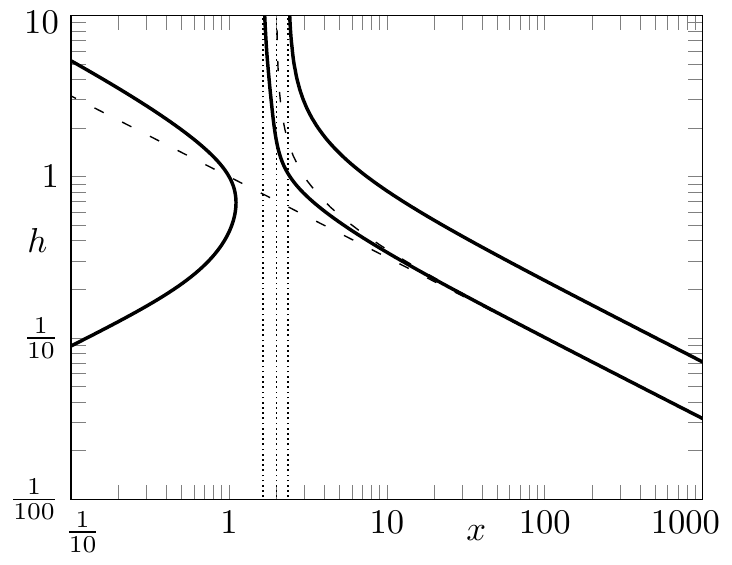}
        \includegraphics{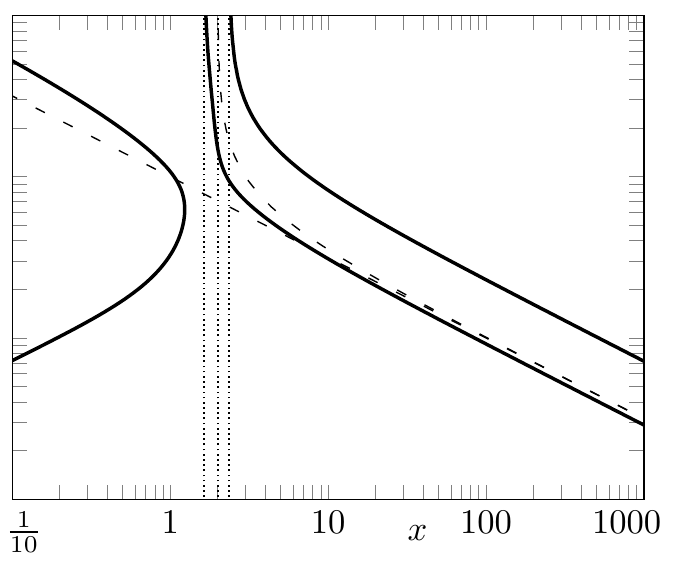}

        (i) $e_1=-\frac{15}{2}$ \hspace{5cm}(ii) $e_1=-5$
    \begin{minipage}{.9\textwidth}
    \captionsetup{format=plain, labelfont=bf}
    \caption{Solution sets with $e_2=0$ and the respective values of $e_1$. Note how the lower asymptotic curve for large $x$ admits the same asymptotic approximation as the curves of constant $\xi$ (here with $\xi_\textup{mc}=0$ and $\xi_\textup{cc}=\frac{1}{6}$) for $e_1=-\frac{15}{2}$, and lies below this asymptote for $e_1>-\frac{15}{2}$.\label{inflationary-plot3}}
    \end{minipage}

\end{figure}

Recalling the shape of the function $s$ from Section~\ref{section-Asymptotics-at-large-x}, we find that $\alpha=1$ being a zero of $s$ is equivalent with $e_1=-\frac{15}{2}$. Hence, for this value of $e_1$, we have a solution curve which is asymptotically equivalent with any curve of constant $\xi$ at large $x$. Moreover, for any value for $e_1$ larger than $-\frac{15}{2}$ but still sufficiently close we have an asymptotic solution curve below the curves of constant $\xi$. If we now approach $e_1\to-\frac{15}{2}$ from above the solution set intersects the curve of minimal coupling at larger and larger values of $x$, and hence at smaller and smaller values of $h$. This situation is depicted in Figure~\ref{inflationary-plot3}.

Also, we have shown in Lemma~\ref{lem:minimal_coupling_lemma} that on the curve of minimal coupling $\xi=0$ the values of $F$ are determined by
\[
    (0,\infty)\to\R,~x\mapsto F\big(x,\tfrac{1}{\sqrt{x}}\big)=-\big(\tfrac{1}{4}+\tfrac{e_1}{30}\big)x-\tfrac{e_2}{30}+\tfrac{2}{3}+\widetilde{F}(x)
\]
with some concave function $\widetilde{F}$ that fulfills $\widetilde{F}(x)\to-\infty$ as $x\to0$ and $\widetilde{F}(x)\to0$ as $x\to\infty$, cf.\ Section~\ref{Section-Minimal-and-conformal-coupling}. From this we conclude that, if $e_2$ is small enough (cf.\ Figure~\ref{inflationary-plot3}), then there exist two zeros for $e_1>-\frac{15}{2}$, and the zero with the larger $x$-value diverges to $+\infty$ as $e_1$ approaches $-\frac{15}{2}$ from above. Consequently, the $h$-value of the corresponding solution of our equation tends to 0. The solution with the larger value for $h$ essentially remains unaffected. By choosing a smaller $e_2$, we can freely adjust the $h$-value of the larger solution. Finally, by the same argument as above we can, moreover, observe a third solution if we tune $\xi$ to positive values close to 0, and the $h$-value of this solution diverges as $\xi\to0$.

Note that in last scenario we can first adjust the middle solution by tuning $e_2$, then the large one by tuning $\xi$ and, finally, the small one by tuning $e_1$. Thereby, one can in principle obtain an arbitrary triplet of positive numbers as the de Sitter solutions in a particular setting.
We emphasize that while the location (and probably even the existence) of what we called the middle solution depends on Assumptions \ref{ass:first}, \ref{ass:second} and \ref{ass:third}, the existence and the locations of the large and the small solution exclusively depends on the asymptotic analysis of Section \ref{Asymptotic-Behavior-of-the-Solution-Set}. More precisely, we can state the following:
\begin{theorem}\label{thm:inflationary_wo_any_assumptions}
    Let $e_2\in\mathbb{R}$. There exist sequences $(\xi^{(n)})_{n\in\mathbb{N}}$ and $(e_1^{(n)})_{n\in\mathbb{N}}$ such that the consistency equation \eqref{energy-equation-massive-case} with parameters $e_1^{(n)}$ and $e_2$ possesses two solutions $h_1^{(n)}$ and $h_2^{(n)}$ along the curve of constant $\xi=\xi^{(n)}$ for which 
    \[
    h_1^{(n)}\to\infty\qquad\textup{and}\qquad h_2^{(n)}\to0
    \]
    as $n\to\infty$. 
\end{theorem}
\begin{proof}
    It is clear from the analysis of Section \ref{Asymptotic-Behavior-of-the-Solution-Set} that a large solution is produced by the limit $\xi\to0$ from above and that a small solution is produced by the limit $e_1\to-\frac{15}{2}$. The difficulty is now to guarantee that, if we first tuned $\xi$ around 0 for the large solution and then $e_1$ around $-\frac{15}{2}$ for the small solution, we do not lose the large solution again. 
    
    This, in turn, can be achieved by restricting to a small interval $I=[-\frac{15}{2}-\varepsilon,-\frac{15}{2}+\varepsilon]$ for $e_1$-values and choosing a uniform upper bound $\widehat{\xi}$ for $\xi$-values below which the larger solution exists and depends continuously on $e_1$. Then, picking a sequence $(\xi^{(n)})_{n\in\mathbb{N}}$ such that $0<\xi^{(n)}<\widehat{\xi}$ for all $n$, we produce ``the large solution''  for any $e_1\in I$. 
    
    As a next step, for any $n\in\mathbb{N}$ we can pick a sequence $(e_1^{(n,k)})_{k\in\mathbb{N}}$ with $e_1^{(n,k)}\in I$ for all $n,k\in\mathbb{N}$, $e_1^{(n,k)}\to-\frac{15}{2}$ as $k\to\infty$ and such that we produce ``the smaller solution'' tending to $0$ as $k\to\infty$. 
    
    Finally, the consistency equation with parameters $e_1^{(n,n)}$ and $e_2$ has (at least) two solutions along the curve of constant $\xi^{(n)}$ for each $n\in\mathbb{N}$, one which tends to $\infty$ and one which tends to $0$ as $n\to\infty$.    
\end{proof}

\begin{remark}
    \begin{itemize}
        \item[\textup{(i)}] While in the proof we exemplary concentrated on the regime around $\xi=0$, the argument works completely analogous around the values $\xi=\xi_{(\pm)}$. In turn, the parameter $e_1=-\frac{15}{2}$ is distinguished since any constant-$\xi$-curve is asymptotically equivalent with an asymptotic solution curve of the consistency equation if and only if $e_1=-\frac{15}{2}$, that is, if and only if $\alpha=1$ is a zero of the function $s$ $($cf.\ the analysis of Section \textup{\ref{section-Asymptotics-at-large-x})}.
        \item[\textup{(ii)}] The theorem implies that given any values $H_\textup{I},H_\textup{DE}>0$ such that $\frac{H_\textup{I}}{H_\textup{DE}}$ is sufficiently large, then these values can be obtained as de Sitter rates of de Sitter-Bunch-Davies solutions to the cosmological SCE for a scalar quantum field with arbitrary $($positive$)$ mass.
    \end{itemize}
\end{remark}

\section{Conclusion and outlook}
\label{sec:conclusion}

Significantly improving upon earlier results in~\cite{JUAREZAUBRY} on semiclassical de Sitter solutions, which were limited to vacuum solutions with vanishing QSE tensor, we have given an extensive and complete analysis of the consistency equation~\eqref{finalenergyequation} to obtain a complete picture of cosmological de Sitter-Bunch-Davies solutions to the SCE.

Our main result is the description of the solution set of cosmological de Sitter-Bunch-Davies solutions as a union of analytic curves in the parameter space of the coupling to the scalar curvature $\xi$ and the expansion rate $H$ as a function of renormalization parameters. Also, as a function of the renormalization parameters, we have obtained a description of the structure of solution set as the union of up to three analytic curves in the plane of $(\xi, H)$ values or in the related $(x,h)$-plane. The techniques applied to generate this map range from elementary solutions of algebraic equations, over the asymptotic analysis of analytic functions, continuity and analyticity arguments based on the implicit function theorem to the reduction of analytic varieties. Also, at some points we had to employ numerical evidence in order to access certain intermediate results, however, we precisely tracked throughout this work which assertion depends on it and which does not (we emphasize that our physically perhaps most appealing result $-$ in terms of Theorem \ref{thm:inflationary_wo_any_assumptions} $-$ does not depend on any such numerical evidence). Note that a large number of solutions do not require a positive cosmological constant $\Lambda>0$ and even for particular cases with $\Lambda <0$ there exist solutions to the SCE with a positive rate of expansion, due to the nature of the QSE tensor. This is particularly interesting in the massless case, where this effect cannot simply be blamed on a positive \emph{renormalized} cosmological constant.

Based on these findings, in particular on the explicit asymptotic expansions of the solution curves, we have identified parameter settings which are compatible with multiple de Sitter solutions, both with very large and very small rates of expansion.

In such settings, studying the Lyapunov stability of the de Sitter solutions is a natural next step. While for the case of conformal coupling with massless fields the dissertation~\cite{haensel2019} clarifies the situation to a certain extent, the situation in the general case seems to be largely open.

Note that the question of Lyapunov stability can be answered on several levels. For special cases with a decoupling of the state degrees of freedom from the SCE like in~\cite{haensel2019,Numerik-Paper} this can be answered by the standard analysis of a dynamical system in finite dimension. Whenever the state dynamics couples non-trivially to the SCE, stability can either be answered in a reduced, cosmological setting~\cite{siemssen-gottschalk,Meda:2020}  or in the setting of the full SCE. See~\cite{meda2022linear} for some investigations on stability in the case of a toy model and \cite{hack2014quantization,froob2017compactly} for linearization techniques of the full SCE system.

It would be especially attractive to find unstable directions for de Sitter solutions with high expansion rates and stability for de Sitter solutions with low expansion rate in the situations described in Section~\ref{sec:Inflationary_models}.  Whether this is achievable or not, at present remains an open research question.

Finally, the inclusion of more general kinds of matter, modeled by fermionic fields or gauge fields, as well as the inclusion of positive or negative spatial curvature certainly is also of interest.

\paragraph{Acknowledgement.} Robert L.\ Bryant is gratefully acknowledged for providing the argument in Section~\ref{section-variety-reduction}. The authors thank Claudio Dappiaggi,  Nicolò Drago,  Hendrik Herrmann, Fernando Lledó, Paolo Meda, Valter Moretti and Nicola Pinamonti for interesting and useful discussions. A special thanks is expressed towards the reviewers of AHP for their carefulness while working through the initial manuscript.

\appendix
\section{Properties of the Bunch-Davies Digamma terms}
\label{appendix-function-f}

In this appendix we collect a few properties of the function
\begin{equation}
    f:(0,\infty)\to\R,~x\mapsto\psi^{(0)}\Big(\tfrac{3}{2}-\sqrt{\tfrac{9}{4}-x}\Big)+\psi^{(0)}\Big(\tfrac{3}{2}+\sqrt{\tfrac{9}{4}-x}\Big)\label{Digamma-dependency-appendix}
\end{equation}
as defined in Section~\ref{positive-mass-simplification}, with the Digamma function $\psi^{(0)}={\Gamma'}/{\Gamma}$.

At first, we want to be a bit more precise on its definition. Note that a priori the mapping in \eqref{Digamma-dependency-appendix} defines a meromorphic function $\widetilde{f}$ on a slit plane, that is, on the complex plane $\mathbb{C}$ from which a half-ray starting in $\frac{9}{4}$ was removed such that the square-root in fact yields a holomorphic function there. However, one can show that the values of \eqref{Digamma-dependency-appendix} do not depend on the choice of which ray was removed, hence \eqref{Digamma-dependency-appendix} defines a meromorphic function on $\mathbb{C}\backslash\{\frac{9}{4}\}$. In a final step, one can show that $\widetilde{f}(\frac{9}{4}):=2\psi^{(0)}$ defines a holomorphic continuation to the point in question and \eqref{Digamma-dependency-appendix} defines a meromorphic function on $\mathbb{C}$ whose poles lie at $\{-n^2-3n\,|\,n\in\mathbb{Z}_{\ge0}\}$. In particular, by restricting $f=\widetilde{f}\big|_{(0,\infty)}$ \eqref{Digamma-dependency-appendix} defines a (real) analytic function on the positive real axis.

Studying the poles and residues of the Gamma function we find that $\widetilde{f}$ has a pole of order 1 in $0\in\CC$. Employing the chain rule for residues we find that $\textup{Res}_0\,\widetilde{f}=-3$ and we conclude the asymptotic equivalence
\begin{align*}
    &f(x)\sim -\frac{3}{x}\qquad\hspace{.7pt}\textup{as}\quad x\to0\,.
\intertext{
On the other hand, we have asymptotically
}
    &f(x)\sim\log(x)\quad\textup{as}\quad x\to\infty
\end{align*}
which immediately follows from Lemma 1 in the appendix of Juárez-Aubry's article~\cite{JUAREZAUBRY}. More precisely, the latter lemma shows that
\[
    |f(x)-\log(x)|\le\frac{3}{x}
\]
for all $x>\frac{9}{4}$. The proof of the latter lemma can easily be extended to see that $f(x)<\log(x)$ and thus
\begin{equation}
    f(x)-\log(x)\in[-\tfrac{3}{x},0)\label{bound-on-f-minus-log}
\end{equation}
for all $x>\frac{9}{4}$. 

Without proof we state the first few Puiseux coefficients of $f-\log$ in the limit $x\to\infty$ as
\[
    f(x)=\log(x)-\frac{4}{3x}-\frac{11}{15x^2}-\frac{92}{315x^3}+\mathcal{O}(\tfrac{1}{x^4})\,.
\]

At last, we can show that $f$ is strictly increasing by estimating its derivative and, moreover, by comparing the values $f(2)=1-2\gamma_\textup{E}<0$ and $f(\frac{9}{4})=4+2\psi^{(0)}(\frac{1}{2})=4-4\log(2)-2\gamma_\textup{E}>0$ (with the Euler-Mascheroni number $\gamma_\textup{E}$) we find that it must have its only zero in the interval $(2,\frac{9}{4})$ (numerically $\approx$2.1646).

As an orientation, Figure~\ref{plot-of-f}.(i) in the text shows a plot of $f$ together with its asymptotics from above.

\bibliographystyle{cmphref}
\bibliography{bibliography.bib}{}
\end{document}